\documentclass[11pt]{article}
\usepackage[margin=1in,left=1in]{geometry}
\usepackage{amsmath}
\usepackage{amsfonts}
\usepackage{amssymb}
\usepackage{amsthm}
\usepackage{graphicx}
\usepackage{color}
\usepackage{hyperref}
\usepackage[all]{xy}

\definecolor{aqua}{rgb}{0, 1.0, 1.0}
\definecolor{fuschia}{rgb}{1.0, 0, 1.0}
\definecolor{gray}{rgb}{0.502, 0.502, 0.502}
\definecolor{lime}{rgb}{0, 1.0, 0}
\definecolor{maroon}{rgb}{0.502, 0, 0}
\definecolor{navy}{rgb}{0, 0, 0.502}
\definecolor{olive}{rgb}{0.502, 0.502, 0}
\definecolor{purple}{rgb}{0.502, 0, 0.502}
\definecolor{silver}{rgb}{0.753, 0.753, 0.753}
\definecolor{teal}{rgb}{0, 0.502, 0.502}

\makeatletter
\newdimen\itex@wd%
\newdimen\itex@dp%
\newdimen\itex@thd%
\def\itexspace#1#2#3{\itex@wd=#3em%
\itex@wd=0.1\itex@wd%
\itex@dp=#2ex%
\itex@dp=0.1\itex@dp%
\itex@thd=#1ex%
\itex@thd=0.1\itex@thd%
\advance\itex@thd\the\itex@dp%
\makebox[\the\itex@wd]{\rule[-\the\itex@dp]{0cm}{\the\itex@thd}}}
\makeatother

\makeatletter
\newif\if@sup
\newtoks\@sups
\def\append@sup#1{\edef\act{\noexpand\@sups={\the\@sups #1}}\act}%
\def\reset@sup{\@supfalse\@sups={}}%
\def\mk@scripts#1#2{\if #2/ \if@sup ^{\the\@sups}\fi \else%
  \ifx #1_ \if@sup ^{\the\@sups}\reset@sup \fi {}_{#2}%
  \else \append@sup#2 \@suptrue \fi%
  \expandafter\mk@scripts\fi}
\def\tensor#1#2{\reset@sup#1\mk@scripts#2_/}
\def\multiscripts#1#2#3{\reset@sup{}\mk@scripts#1_/#2%
  \reset@sup\mk@scripts#3_/}
\makeatother

\makeatletter
\newbox\slashbox \setbox\slashbox=\hbox{$/$}
\def\itex@pslash#1{\setbox\@tempboxa=\hbox{$#1$}
  \@tempdima=0.5\wd\slashbox \advance\@tempdima 0.5\wd\@tempboxa
  \copy\slashbox \kern-\@tempdima \box\@tempboxa}
\def\slash{\protect\itex@pslash}
\makeatother

\def\clap#1{\hbox to 0pt{\hss#1\hss}}

\let\oldroot\root
\def\root#1#2{\oldroot #1 \of{#2}}
\renewcommand{\sqrt}[2][]{\oldroot #1 \of{#2}}

\DeclareSymbolFont{symbolsC}{U}{txsyc}{m}{n}
\SetSymbolFont{symbolsC}{bold}{U}{txsyc}{bx}{n}
\DeclareFontSubstitution{U}{txsyc}{m}{n}

\DeclareSymbolFont{stmry}{U}{stmry}{m}{n}
\SetSymbolFont{stmry}{bold}{U}{stmry}{b}{n}

\DeclareFontFamily{OMX}{MnSymbolE}{}
\DeclareSymbolFont{mnomx}{OMX}{MnSymbolE}{m}{n}
\SetSymbolFont{mnomx}{bold}{OMX}{MnSymbolE}{b}{n}
\DeclareFontShape{OMX}{MnSymbolE}{m}{n}{
    <-6>  MnSymbolE5
   <6-7>  MnSymbolE6
   <7-8>  MnSymbolE7
   <8-9>  MnSymbolE8
   <9-10> MnSymbolE9
  <10-12> MnSymbolE10
  <12->   MnSymbolE12}{}

\makeatletter
\def\re@DeclareMathSymbol#1#2#3#4{%
    \let#1=\undefined
    \DeclareMathSymbol{#1}{#2}{#3}{#4}}
\re@DeclareMathSymbol{\neArrow}{\mathrel}{symbolsC}{116}
\re@DeclareMathSymbol{\neArr}{\mathrel}{symbolsC}{116}
\re@DeclareMathSymbol{\seArrow}{\mathrel}{symbolsC}{117}
\re@DeclareMathSymbol{\seArr}{\mathrel}{symbolsC}{117}
\re@DeclareMathSymbol{\nwArrow}{\mathrel}{symbolsC}{118}
\re@DeclareMathSymbol{\nwArr}{\mathrel}{symbolsC}{118}
\re@DeclareMathSymbol{\swArrow}{\mathrel}{symbolsC}{119}
\re@DeclareMathSymbol{\swArr}{\mathrel}{symbolsC}{119}
\re@DeclareMathSymbol{\nequiv}{\mathrel}{symbolsC}{46}
\re@DeclareMathSymbol{\Perp}{\mathrel}{symbolsC}{121}
\re@DeclareMathSymbol{\Vbar}{\mathrel}{symbolsC}{121}
\re@DeclareMathSymbol{\sslash}{\mathrel}{stmry}{12}
\re@DeclareMathSymbol{\bigsqcap}{\mathop}{stmry}{"64}
\re@DeclareMathSymbol{\biginterleave}{\mathop}{stmry}{"6}
\re@DeclareMathSymbol{\invamp}{\mathrel}{symbolsC}{77}
\re@DeclareMathSymbol{\parr}{\mathrel}{symbolsC}{77}
\makeatother

\makeatletter
\def\Decl@Mn@Delim#1#2#3#4{%
  \if\relax\noexpand#1%
    \let#1\undefined
  \fi
  \DeclareMathDelimiter{#1}{#2}{#3}{#4}{#3}{#4}}
\def\Decl@Mn@Open#1#2#3{\Decl@Mn@Delim{#1}{\mathopen}{#2}{#3}}
\def\Decl@Mn@Close#1#2#3{\Decl@Mn@Delim{#1}{\mathclose}{#2}{#3}}
\Decl@Mn@Open{\llangle}{mnomx}{'164}
\Decl@Mn@Close{\rrangle}{mnomx}{'171}
\Decl@Mn@Open{\lmoustache}{mnomx}{'245}
\Decl@Mn@Close{\rmoustache}{mnomx}{'244}
\makeatother

\makeatletter
\DeclareRobustCommand\widecheck[1]{{\mathpalette\@widecheck{#1}}}
\def\@widecheck#1#2{%
    \setbox\z@\hbox{\m@th$#1#2$}%
    \setbox\tw@\hbox{\m@th$#1%
       \widehat{%
          \vrule\@width\z@\@height\ht\z@
          \vrule\@height\z@\@width\wd\z@}$}%
    \dp\tw@-\ht\z@
    \@tempdima\ht\z@ \advance\@tempdima2\ht\tw@ \divide\@tempdima\thr@@
    \setbox\tw@\hbox{%
       \raise\@tempdima\hbox{\scalebox{1}[-1]{\lower\@tempdima\box
\tw@}}}%
    {\ooalign{\box\tw@ \cr \box\z@}}}
\makeatother

\makeatletter
\def\udots{\mathinner{\mkern2mu\raise\p@\hbox{.}
\mkern2mu\raise4\p@\hbox{.}\mkern1mu
\raise7\p@\vbox{\kern7\p@\hbox{.}}\mkern1mu}}
\makeatother

\newcommand{\infinity}{\infty}

\newcommand{\R}{\ensuremath{\mathbb R}}

\renewcommand{\(}{\begin{equation*}}
\renewcommand{\)}{\end{equation*}}
\newcommand{\bea}{\begin{eqnarray*}}
\newcommand{\eea}{\end{eqnarray*}}

\theoremstyle{italics}
\newtheorem{theorem}{Theorem}[section]
\newtheorem{lemma}[theorem]{Lemma}
\newtheorem{prop}[theorem]{Proposition}

\theoremstyle{definition}
\newtheorem{defn}[theorem]{Definition}

\newtheorem{example}[theorem]{Example}

\theoremstyle{remark}
\newtheorem{remark}[theorem]{Remark}
\newtheorem{note[theorem]}{Note}

\begin{document}

%

\title{T-Duality from super Lie $n$-algebra cocycles for super $p$-branes
}

 \author{Domenico Fiorenza\thanks{Dipartimento di Matematica, La Sapienza Universit\`a di Roma Piazzale Aldo Moro 2, 00185 Rome, Italy}, \;
     Hisham Sati\thanks{University of Pittsburgh, Pittsburgh, PA 15260, USA, and New York University, Abu Dhabi, UAE}, \;
     Urs Schreiber\thanks{Mathematics Institute of the Academy, {\v Z}itna 25, 115 67 Praha 1, Czech Republic}
     }

\maketitle
\begin{abstract}
  We compute the $L_\infty$-theoretic double dimensional reduction of the F1/D$p$-brane super $L_\infty$-cocycles
  with coefficients in rationalized twisted K-theory
  from the 10d type IIA and type IIB super Lie algebras down to 9d.
  We show that the two resulting coefficient $L_\infty$-algebras
  are naturally related by an $L_\infty$-isomorphism which we find to act on the super $p$-brane cocycles
  by the infinitesimal version of the rules of topological T-duality and inducing an isomorphism
  between $K^0$-cocycles in type IIA and $K^1$-cocycles in type IIB, rationally.
  In particular this is a derivation of the Buscher rules for RR-fields (Hori's formula) from first principles.
  Moreover, we show that these $L_\infty$-algebras are the homotopy quotients of the
  RR-charge coefficients by the ``T-duality Lie 2-algebra''.
  We find that the induced $L_\infty$-extension is a gerby extension of a $9+(1+1)$
  dimensional (i.e. ``doubled'') T-duality correspondence super-spacetime, which serves
  as a local model for T-folds.
  We observe that this still extends, via the D0-brane cocycle of its type IIA factor,
  to a $10+(1+1)$-dimensional super Lie algebra. Finally we show that this satisfies
  expected properties of a local model space for F-theory elliptic fibrations.
\end{abstract}

\tableofcontents

\section{Introduction}

Understanding and constructing fields and branes in string theory  and M-theory
in a manner compatible with
supersymmetry and with the various dualities is a fundamental problem,
both from the theoretical as well as the phenomenological point of view.
On the former, a complete solution to this problem would provide a solid ground for deriving the theory from
firm (mathematical) principles. On the latter, it would help in the
systematic classification of allowable vacua.

\medskip

The fundamental super $p$-branes
that have no gauge fields on their worldvolume and
which propagate on super-Minkowski spacetime
are defined via Green-Schwarz type action functionals.
These are higher-dimensional super coset WZW-type functionals \cite{HM},
for super-Minkowski regarded as the super-coset of super-Poincar{\'e} by the Spin cover of the
Lorentz group.
Accordingly, these $p$-branes are classified by the invariant super Lie algebra cohomology
of the supersymmetry algebras, a fact known as the ``old brane scan'' \cite{AETW87}.
When super-Minkowski target spacetime is generalized to curved super-spacetimes,
then this statement applies super-tangent-space wise: the bispinorial component of the
field strength super $(p+2)$-form $H_{p+2}$ to which the $p$-brane couples
is constrained to coincide in
each tangent space with the left-invariant form
$\tfrac{1}{p!} \left(\overline{\psi} \wedge \Gamma_{a_1\cdots a_p} \psi\right) \wedge e^{a_1} \wedge \cdots e^{a_p}$
corresponding to the super-cocycles in the old brane scan \cite{BST86, BST87}.
Notice that this is in direct analogy for instance to $G_2$-structures on 7-manifolds, which are
given by differential 3-forms that are constrained to coincide tangent-space wise with
a fixed 3-cocycle on $\mathbb{R}^7$. In particular, the bosonic part of the field strength $H_{p+2}$
may vanish identically, and still its bispinorial component is constrained to coincide with the
given cocycle super-tangent-space wise. In this way the super Lie algebra cocycles tightly control
the structure of super $p$-brane charges.

\medskip
We view the above as a powerful statement:
While in general these differential forms $H_{p+2}$ are just the image in real cohomology
(``rationalization'') of components of some more refined cohomology theory,
this says that at least the rational image of the charges of branes without gauge fields
on their worldvolume is classified by super Lie algebra cohomology.
Notice that in certain instances
rationalization of twisted generalized cohomology in the treatment of T-duality
is even forced upon us  (see \cite{LSW}).

\medskip
This phenomenon turns out to generalize also to those branes that do carry (higher) gauge fields on their worldvolume,
such as the D-branes and the M5-brane -- \emph{if} one generalizes the Chevalley-Eilenberg algebras
of super-Minkowski spacetimes to quasi-free differential-graded super-commutative algebras with
generators also in higher degree \cite{CAIB00, IIBAlgebra}. These DG-algebras are just the ``FDAs''
from the supergravity literature \cite{DAuriaFre82,CDF}. For instance what in \cite{CAIB00} is identified as the
cocycle for the M5-brane earlier appeared as an algebraic ingredient in the construction of
11d supergravity in \cite{DAuriaFre82}; similarly the algebra for the D$p$-brane charges found in
\cite{IIBAlgebra} and \cite{CAIB00} appears earlier as an ingredient for constructing type II
supergravity in \cite{Ca}.

\medskip
Now, while (super) Lie algebra cohomology is a respectable mathematical subject, what
are these ``extended super Minkowski algebras'' that carry the D-brane and M-brane charges,
really? In \cite{FSS13} we had pointed out, following \cite{SSS09}, that these ``FDA''s
are naturally identified as the Chevalley-Eilenberg algebras of \emph{super Lie $n$-algebras}
also called \emph{$n$-term super $L_\infty$-algebras}, for higher $n$
\footnote{Notice that these are Lie $n$-algebras in the sense of Stasheff
\cite{LadaStasheff93, LadaMarkl95, SSS09},
not ``$n$-Lie algebras'' in the sense of Filippov.
However, the two notions are not unrelated. At least the Filippov 3-algebras
that appear in the BLG model of coincident solitonic M2-branes may naturally be understood as Stasheff-type Lie
2-algebras equipped with a metric form. This observation is due to \cite[section 2]{PSa}, based on \cite{FF}.},
and that under this identification the Lie theory that underlies the ``old brane scan''
turns into the ``higher Lie theory'' or homotopy theory of super Lie $n$-algebras
that sees the entire super $p$-brane content \cite{FSS13}.

\medskip
This homotopy-theoretic perspective sheds further light on the classification of
super $p$-branes. For instance, it identifies the extended super-Minkowski spacetimes of
\cite{CAIB00, IIBAlgebra} as being the higher central super Lie $n$-algebra
extensions of super-Minkowski spacetime that are classified by the 3-cocycles for
the superstring and by the 4-cocycle for the M2-brane. This is in higher analogy to how
2-cocycles classify ordinary central extensions. Namely, these extended spacetimes
are the \emph{homotopy fibers} of the corresponding cocycles \cite[Prop. 3.5]{FSS13},
see example \ref{homotopyfiberofLinfinityCocycles} below.

\medskip
This means that by embedding super Lie algebra theory into the larger context
of homotopy super Lie algebra theory (super $L_\infty$-algebra theory) then
\emph{all} super $p$-branes are found by a sequence of consecutive higher
invariant extensions, yielding a ``bouquet of branes'' growing out of the super-spacetimes
\cite{FSS13}. This provides the generalization of the ``old brane scan'' that was argued to be needed
in \cite{LPSS}. There, the multiplicity of elementary and solitonic
$p$-brane solutions to supergravity theories was shown to cover many more values of
$(D,d)$ than the
classic $\kappa$-symmetric points on the brane scan, suggesting that the
original classification needs to be generalized.
Indeed, our work pins down that the required generalization is from super Lie
algebra cohomology to super $L_\infty$-algebra cohomology.

\medskip
In fact this bouquet of invariant higher super $L_\infty$-cocycles is
rooted in 0-dimensional super-space, the superpoint. 
It
is a diagram of super Lie $n$-algebras of the following form:
$$
  
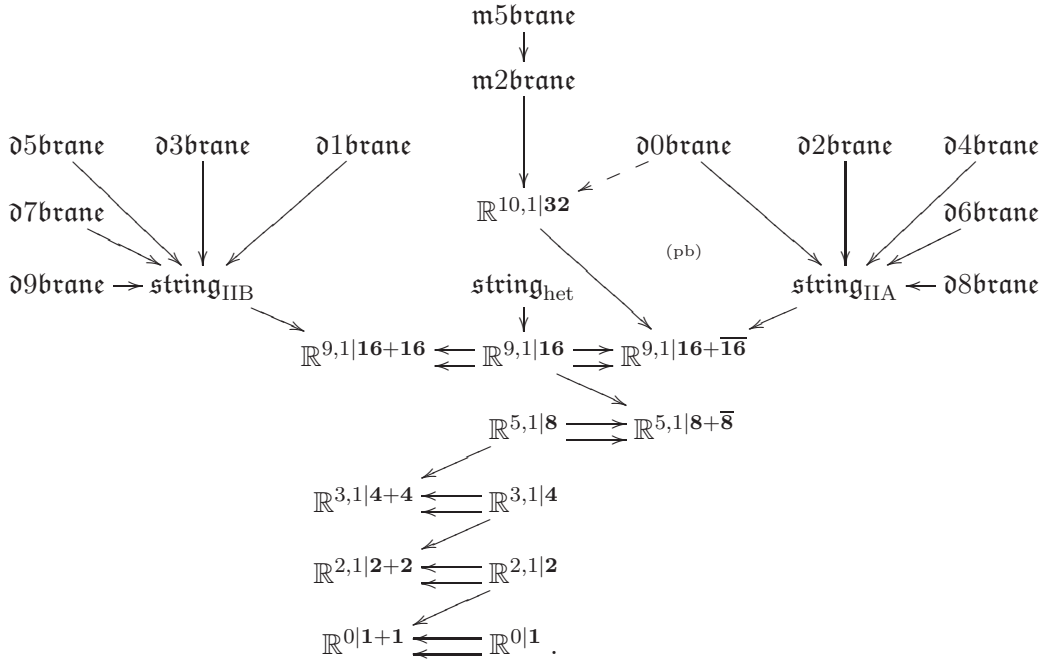
\hspace{-1.5cm}
  \xymatrix@=1em{
    &
    &&&& \mathfrak{m}5\mathfrak{brane}
     \ar[d]
    \\
    &
    &&
     && \mathfrak{m}2\mathfrak{brane}
    \ar[dd]
    &&
    \\
    &
    &
    \mathfrak{d}5\mathfrak{brane}
    \ar[ddr]
    &
    \mathfrak{d}3\mathfrak{brane}
    \ar[dd]
    &
    \mathfrak{d}1\mathfrak{brane}
    \ar[ddl]
    &
    & \mathfrak{d}0\mathfrak{brane}
    \ar@{}[ddd]|{\mbox{\tiny (pb)}}
    \ar[ddr]
    \ar@{-->}[dl]
    &
    \mathfrak{d}2\mathfrak{brane}
    \ar[dd]
    &
    \mathfrak{d}4\mathfrak{brane}
    \ar[ddl]
    \\
    &
    &
    \mathfrak{d}7\mathfrak{brane}
    \ar[dr]
    &
    &
    & \mathbb{R}^{10,1\vert \mathbf{32}}
      \ar[ddr]
    &&&
    \mathfrak{d}6\mathfrak{brane}
    \ar[dl]
    \\
    &
    &
    \mathfrak{d}9\mathfrak{brane}
    \ar[r]
    &
    \mathfrak{string}_{\mathrm{IIB}}
    \ar[dr]
    &
    & \mathfrak{string}_{\mathrm{het}}
      \ar[d]
    &&
    \mathfrak{string}_{\mathrm{IIA}}
    \ar[dl]
    &
    \mathfrak{d}8\mathfrak{brane}
    \ar[l]
    \\
    &
    &
    &
    &
    \mathbb{R}^{9,1 \vert \mathbf{16} + {\mathbf{16}}}
    \ar@{<-}@<-3pt>[r]
    \ar@{<-}@<+3pt>[r]
    & \mathbb{R}^{9,1\vert \mathbf{16}} \ar[dr]
    \ar@<-3pt>[r]
    \ar@<+3pt>[r]
    &
    \mathbb{R}^{9,1\vert \mathbf{16} + \overline{\mathbf{16}}}
    \\
    &
    &
    &
    &
    &
    \mathbb{R}^{5,1\vert \mathbf{8}}
    \ar[dl]
    &
    \mathbb{R}^{5,1 \vert \mathbf{8} + \overline{\mathbf{8}}}
    \ar@{<-}@<-3pt>[l]
    \ar@{<-}@<+3pt>[l]
    \\
    &
    &
    &
    &
    \mathbb{R}^{3,1\vert \mathbf{4}+ \mathbf{4}}
    \ar@{<-}@<-3pt>[r]
    \ar@{<-}@<+3pt>[r]
    &
    \mathbb{R}^{3,1\vert \mathbf{4}}
    \ar[dl]
    \\
    &
    &
    &
    &
    \mathbb{R}^{2,1 \vert \mathbf{2} + \mathbf{2} }
    \ar@{<-}@<-3pt>[r]
    \ar@{<-}@<+3pt>[r]
    &
    \mathbb{R}^{2,1 \vert \mathbf{2}}
    \ar[dl]
    \\
    &
    &
    &
    &
    \mathbb{R}^{0 \vert \mathbf{1}+ \mathbf{1}}
    \ar@{<-}@<-3pt>[r]
    \ar@{<-}@<+3pt>[r]
    &
    \mathbb{R}^{0\vert \mathbf{1}}\;.
  }
$$

\medskip
\medskip
\noindent Here every solid arrow denotes a central super $L_\infty$-algebra extension which is
invariant with respect to the maximal semisimple part of the bosonic body of the external automorphisms
(i.e. all automorphisms modulo R-symmetries) of
the super $L_\infty$-algebra that is being extended. These turn out to be the respective
Lorentz groups (their Spin-covers). 
Notice that
the claim is that it is the maximality of these invariant extensions which implies that the
extensions of the superpoint are super-Minkowski spacetimes, and that they are precisely
of dimension increasing from 0 through 3,4, 6, 10 to 11.
The top of this diagram is discussed in \cite{FSS13}.
The proof of the ``trunk'' of the bouquet is due to \cite{HuertaSchreiber}.
(For dimensions 0 to 3 the statement was observed earlier in \cite{SchreiberBristol}.)

\medskip
This shows that the core structure of string/M-theory
follows from first principles in higher super Lie algebra theory, with no need of an external
input from Lorentz geometry, or Spin geometry. Instead, Lorentzian geometry
and Spin geometry is discovered (re-discovered) by analyzing the super-point in
higher super Lie theory, as is the existence of all the super $p$-branes in their respective
super-spacetimes.\footnote{
 Compare to \cite[p.41]{Moore14}:
 ``Perhaps we need to understand the nature of time it self better. [...] One natural way
to approach that question would be to understand in what sense time itself is an emergent
concept, and one natural way to make sense of such a notion is to understand how
pseudo-Riemannian geometry can emerge from more fundamental and abstract notions such as categories of branes.''
}
This further suggests that higher super Lie theory knows much more about the inner working
of string/M-theory, and that we may check conjectures on M-theory and discover its missing details  by a systematic
homotopy theoretic analysis of the superpoint.
In the present article we are concerned with discovering and studying \emph{T-duality}
from this perspective (see e.g. \cite{AAL}). We will show that this allows to systematically derive
phenomena that have been proposed or conjectured but which seem to have
been lacking a derivation from first principles, such as the
rules of ``topological T-dulity'' on supermanifolds, the nature of super-T-folds
and the emergence of F-theory from M-theory.

\medskip
In order to describe the action of T-duality on the F1/D$p$ branes, we need the cocycles for the super D-branes
not as cocycles with coefficients in $\mathbb{R}$
on the type II extended super-Minkowski super Lie 2-algebras that are denoted
$\mathfrak{string}_{\mathrm{IIA}}$ and $\mathfrak{string}_{\mathrm{IIB}}$ in the above diagram,
but we need to descend them to cocycles on the type II super-Minkowski super Lie 1-algebras
themselves, where they will take values in more complicated richer coefficients.
This homotopical descent of the brane cocycles we have previously discovered and studied in
 \cite{FSS15} and \cite{FSS16}. There we showed that homotopy theory
allows to descend the iterated $\mathbb{R}$-valued cocycles for separate $p$-branes,
defined on a extended Minkowski spacetime, back to single cocycles on plain super-Minkowski
spacetime, but now taking values in more complex coefficients. We
showed in \cite[section 4]{FSS16} that applying this homotopy-theoretic descent to the
type IIA D$p$-brane cocycles gives that jointly they combine with the cocycle
for the type IIA superstring to one single cocycle, now with coefficients in the
$L_\infty$-algebra which is the image under Lie differentiation of the classifying
space $\mathrm{KU}/\mathbf{B}U(1)$ for twisted K-theory.
Of course
twisted K-theory has famously been argued earlier to be the correct cohomology
theory in which F1/D$p$-brane (background) charges properly take value \cite{Wi} \cite{MooreWitten00} \cite{Ka} \cite{BM}.
Therefore, it is of interest to explore how much more the homotopy theory of super Lie-$n$ algebra
may teach us about string/M-theory.

\medskip


In the present article, we first observe in  section \ref{SectionDoubleDimensionalReduction} that (double) dimensional
reduction is naturally encoded on super $L_\infty$-algebras by \emph{cyclification}, namely by the process which in terms of rational homotopy theory corresponds to passing to
homotopy quotients of free loop spaces by the rotation action on loops. We find that $L_\infty$-theoretically
this (double) dimensional reduction is an \emph{isomorphism}, hence has an inverse (``oxidation'', see \cite{LPSS})
that completely reconstructs the higher dimensional situation
(by the $L_\infty$-incarnation of D0-brane condensation \cite[Remark 3.11, 4.6]{FSS13}).
In this sense this process is non-perturbative, a fact that is important for our discussion of
F-theory further below.

\medskip
We first use this reduction isomorphism in section \ref{SectionDoubleDimensionalReduction} to recall
from \cite{FSS13} the
form of the descended type IIA F1/D$p$-brane cocycles for $p \in \{0,2,4\}$ as the dimensional reduction of the M-brane cocycles
from 11d. Then we we observe that down in 10d these are enhanced to cocycles for $p$-branes
for $p \in \{0,2,4,8, 10\}$, using the Fierz identity analysis in \cite{CAIB00}.
We also state the corresponding type IIB F1/D$p$-brane cocycles. These may be similarly extracted from
analysis of Fierz identities \cite{IIBAlgebra}, but our main theorem below (Theorem \ref{TheTDualityIsoOnLInfinityCocycles})
also implies the form of either (IIA or IIB) from the other.

\medskip
 In section \ref{TDuality} we first
compute  the dimensional reduction of the type IIA F1/D$p$-brane cocycles from the 10d super-Minkowski
super Lie algebra to 9d. We observe that the two $L_\infty$-algebras of coefficients of the descended
cocycles (for type IIA and type IIB) are manifestly related by a an $L_\infty$-isomorphism (Proposition \ref{CoefficientsIsoOnTDuality}.)
Then we show in Theorem  \ref{TheTDualityIsoOnLInfinityCocycles} that the action of this isomorphism on the descended type II $L_\infty$-cocycles
 implements \emph{T-duality} between the F1/D$p$ brane charges of IIA and type IIB string theory:
we discover that the super-cocycles follow the infinitesimal version of the rules of
``topological T-duality'', originally proposed by \cite{BouwknegtEvslinMathai04}
with precise formulation due to \cite{BunkeSchick05}\cite{BunkeRumpfSchick06},
this is the content of remarks \ref{NatureOfTheAlgebraicTDualityIso} and \ref{BunkeStyleTDuality} below.
Notice that even though the coefficients we obtain are just the rational image of the twisted K-theory
that appears in topological T-duality, this is the first time (to the best of our knowledge)
that the rules for topological T-duality are actually derived from string theoretic first principles,
and that topological T-duality is connected to local spacetime supersymmetry.
 In fact our derivation
shows that the existence and structure of (topological) T-duality acting of F1/D$p$-brane charges is
entirely controled by higher super Lie algebra theory. Similarly, we derive the
Buscher rules for RR-fields (\emph{Hori's formula}) from first principles this way (Proposition \ref{TDualityViaPullPush},
Remark \ref{TopologicalTDualityIII}).

\medskip
Our main theorem (Theorem  \ref{TheTDualityIsoOnLInfinityCocycles})
shows that in the category of super $L_\infty$-algebras
T-duality is incarnated as the right part of a diagram of the following form:

{\tiny
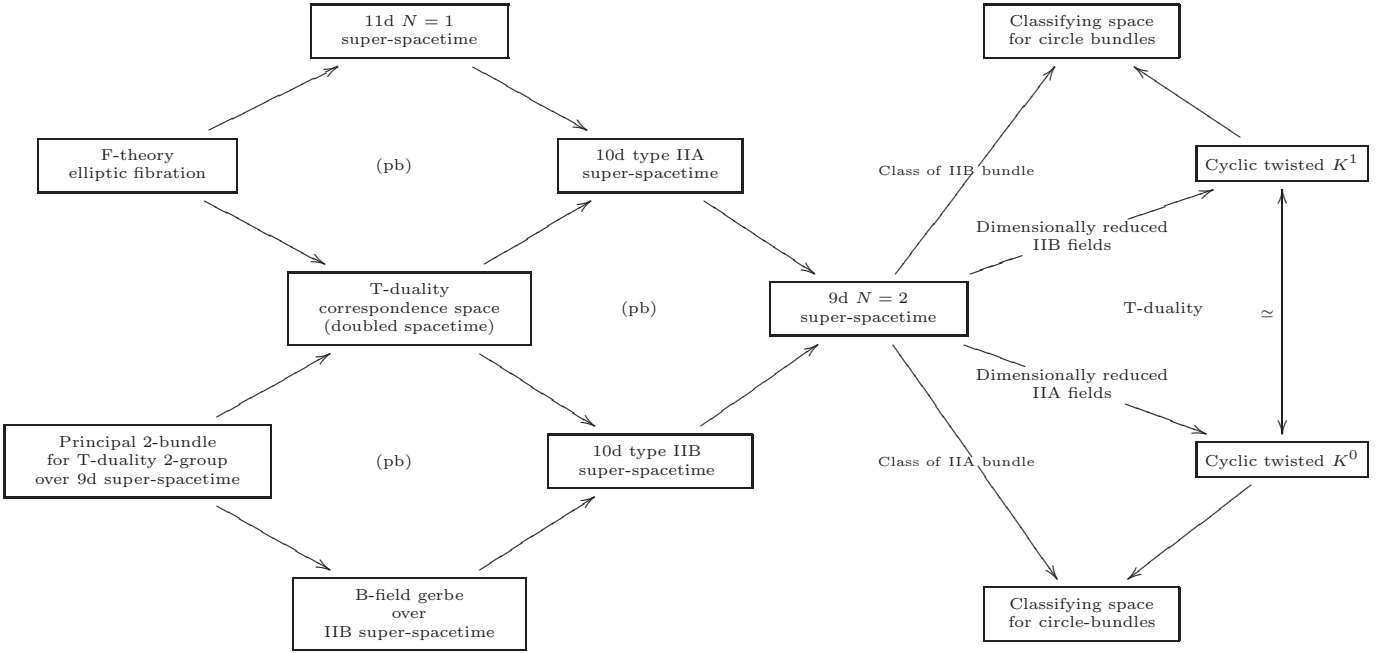
$$
  \hspace{-1cm}
  \xymatrix@C=.01em{
      &
      \fbox{
      \begin{tabular}{c}
        11d $N = 1$
        \\
        super-spacetime
      \end{tabular}
    }
    \ar[dr]
    & & & \fbox{\begin{tabular}{c}Classifying space \\ for circle bundles \end{tabular}}
    \\
    \fbox{
    \begin{tabular}{c}
      F-theory
      \\
      elliptic fibration
    \end{tabular}
    }
    \ar@{}[rr]|{\mbox{(pb)}}
    \ar[ur]
    \ar[dr]
    &
    &
    \fbox{\mbox{\tiny \begin{tabular}{c}10d type IIA \\super-spacetime \end{tabular}}}
    \ar[dr]
    & && \fbox{Cyclic twisted $K^1$}
    \ar[ul]
    \\
    &
    \fbox{
    \begin{tabular}{c}
      T-duality
      \\
      correspondence space
      \\
      (doubled spacetime)
    \end{tabular}
    }
    \ar@{}[rr]|{\mbox{\tiny (pb)}}
    \ar[ur]
    \ar[dr]
    &
    &
    \fbox{\mbox{ \begin{tabular}{c} 9d $N=2$ \\ super-spacetime \end{tabular} }}
    \ar[drr]|{ \mbox{\tiny  \begin{tabular}{c} Dimensionally reduced \\ IIA fields \end{tabular} }  }
    \ar[urr]|{ \mbox{\tiny \begin{tabular}{c} Dimensionally reduced \\ IIB fields \end{tabular} } }
    \ar@{}[rr]|>>>>>>>>>>>>>>>{\mbox{T-duality}}
    \ar[ddr]|{\hspace{-5mm}\rm Class\ of\ IIA\ bundle}
    \ar[uur]|{\hspace{-5mm} \rm Class\ of\ IIB\ bundle}
    &&
    \\
    \fbox{
    \begin{tabular}{c}
      Principal 2-bundle
      \\
      for T-duality 2-group
      \\
      over 9d super-spacetime
    \end{tabular}
    }
    \ar@{}[rr]|{\mbox{(pb)}}
    \ar[ur]
    \ar[dr]
    &
    &
    \fbox{\mbox{ \tiny \begin{tabular}{c} 10d type IIB \\ super-spacetime \end{tabular} } }
    \ar[ur]
    &
    && \fbox{Cyclic twisted $K^0$}
    \ar@{<->}[uu]^\simeq
    \ar[dl]
    \\
    &
    \fbox{
      \begin{tabular}{c}
        B-field gerbe
        \\
        over
        \\
        IIB super-spacetime
      \end{tabular}
    }
    \ar[ur]
    &
    &
    & \fbox{\begin{tabular}{c}Classifying space \\ for circle-bundles \end{tabular}}
  }
$$
}

The bottom left part of this diagram we discover in section \ref{TFolds}.
Namely topological T-duality in the alternative formulation of \cite{BunkeRumpfSchick06} is controlled
by the \emph{correspondence space} which is the fiber product of the IIA spacetime
with the T-dual IIB spacetime over their joint 9d base. In the string theory literature
this is essentially what is known as the \emph{doubled spacetime} \cite{Hull07}.
We first show that the correspondence space axiom for topological T-duality due to \cite{BunkeRumpfSchick06}
is satisfied by the relevant super Lie $n$-algebra extensions of super-Minkowski spacetimes
(Proposition \ref{BunkeStyleTDualityInfinitesimally}).

\medskip
Then we show that after
stripping off the RR-charge coefficients in 9d, the remaining coefficient $L_\infty$-algebra
is the delooping of the ``T-duality Lie 2-algebra'' (Def. \ref{Lie2AlgebraOfTDuality2Group}, Rem. \ref{TDuality2Group}). This is the homotopy fiber of the
cup product of two universal first Chern classes, in direct analogy to how the
``string Lie 2-algebra'' that controls the anomaly cancellation in heterotic string theory
is the homotopy fiber of the second Chern class/first Pontryagin class \cite{SSSIII}.
We show that the T-duality Lie 2-algebra has a natural $\infty$-action on the direct sum of twisted $K^0$ and $K^1$
(Prop. \ref{TDuality2GroupAction}).
Finally we construct the T-duality 2-group-principal 2-bundle over 9d super-Minkowski
spacetime that is classified by the descended F1/D$p$-brane cocycles, and find that this is equivalent to the total space
of the (either) gerbe extension of the T-duality correspondence spacetime
(Proposition \ref{TheTFoldAlgebra} below).

\medskip
It has been argued in \cite{Nikolaus14}
 that the principal 2-bundles of which this is the local model space are the right mathematical formulation of
Hull's concept of \emph{T-folds} \cite{Hull05} \cite{Hull07}.
This we will further discuss elsewhere.
Here we just remark that, as we already amplified for T-duality itself, while here we obtain
just the infinitesimal/rational image of this structure, we connect (for the first time, to our knowledge)
to spacetime super-symmetry and obtain what is in fact the local model for \emph{super T-folds}.

\medskip
Finally, in section \ref{SectionOnF} we observe that these doubled super T-correspondence spacetimes
still carry the D0-brane $L_\infty$-cocycle inherited through their type IIA fiber factor.
Hence there is a central super extension of this to a super Lie algebra of bosonic dimension 10+2.
We show in this last section that this super Lie algebra has the correct properties
to be expected of the local model space for an F-theory elliptic fibration,
according to \cite{Vafa96, Johnson97}. This gives the top left part of the above diagram, below
 this is Prop. \ref{TheFTheorySpacetimeInContext}.

\medskip
\noindent {\bf Related literature.}
Our T-duality takes place on superspaces, for a related discussion see \cite{Siegel}.
There are various other approaches to T-duality in relation to branes.
The brane worldvolume approach to T-duality transformations
as transformations which mix the worldvolume field equations with Bianchi identities
is discussed in \cite{PS} \cite{Se}.
 T-duality between D-branes is realized on the underlying $p$-brane solutions
of type IIA and type IIB supergravity in \cite{BdR}.
The relation to worldsheet and spacetime
 supersymmetry is discussed in \cite{BS},
  for  the Green-Schwarz superstring in
\cite{CLPS} \cite{KR}, and in the presence of RR fields in
 \cite{Ha}. In \cite{HKS}, a superspace with manifest T-duality including Ramond-Ramond
gauge fields is presented. The superspace is defined by the double nondegenerate
super-Poincare algebras where Ramond-Ramond charges are introduced by central
extension. In \cite{Saz} an interpretation of T-dualization procedure of type II superstring
theory in double space is given, taking into account compatibility between supersymmetry
and T-duality. A geometry of superspace corresponding to double field theory for type
II supergravity is introduced in \cite{Ced}  based on an orthosymplectic extension
${\rm OSp}(d,d|2s)$ of the continuous T-duality group.

\section{Supersymmetry super Lie $n$-algebras}
\label{Spinors}

Here we introduce what we need below on super Lie $n$-algebras associated with supersymmetry.
Similarly to how $L_\infty$-algebras are extensions of Lie algebras to include higher brackets,
super $L_\infty$-algebras are extensions of Lie superalgebras to include higher
graded brackets.  These turn out to be defined via more familiar differential graded (DG) algebras
when we restrict to finite-dimensional super-vector spaces.

\begin{defn}\label{superLInfinityAlgebra}
  Write
  $$
    \mathrm{CE}
      \;:\;
    \mathrm{sL}^{\mathrm{fin}}_\infty\mathrm{Alg}_{\mathbb{R}}
      \hookrightarrow
    \mathrm{dgAlg}_{/\mathbb{R}}^{\mathrm{op}}
  $$
  for the full subcategory of the opposite of that of differential graded-algebras augmented over $\mathbb{R}$
  (i.e. DG $\mathbb{R}$-algebras equipped with a  algebra homomorphism to $\mathbb{R}$) whose underlying
  graded algebra is freely generated as a graded super-commutative algebra on a
  $\mathbb{Z}$-graded super-vector space which is degreewise finite dimensional.
  This is the category of \emph{super $L_\infty$-algebras of finite type}.

  {
  If we consider just $\mathbb{R}$-algebras $\mathrm{Alg}_{\mathbb{R}}$
  instead of augmented $\mathbb{R}$-algebras $\mathrm{Alg}_{/\mathbb{R}}$, then the analogous full inclusion
  $$
    \mathrm{CE}
      \;:\;
    \mathrm{sL}^{\mathrm{cvd},\mathrm{fin}}_\infty\mathrm{Alg}_{\mathbb{R}}
      \hookrightarrow
    \mathrm{dgAlg}_{\mathbb{R}}^{\mathrm{op}}
  $$
  is called that of possibly \emph{curved} $L_\infty$-algebras of finite type, with possibly \emph{curved}
  homomorphism between them.

  Under the functor $\mathrm{dgAlg}_{/\mathbb{R}}^{\mathrm{op}} \to \mathrm{dgAlg}_{\mathbb{R}}^{\mathrm{op}}$ which forgets the augmentation,
  then every $L_\infty$-algebra is regarded as a curved $L_\infty$-algebra with \emph{vanishing curvature} and as
  such there are then possibly curved homomorphisms between non-curved $L_\infty$-algebras.
  }
\end{defn}
\begin{remark}\label{striuctureInLInfinity}
  This means that for $\mathfrak{g} \in \mathrm{sL}_\infty\mathrm{Alg}_{\mathbb{R}}$
  a super $L_\infty$-algebra, then
  every generator of its Chevalley-Eilenberg algebra $\mathrm{CE}(\mathfrak{g})$
  carries a bidegree $(n,\sigma) \in \mathbb{Z}\times (\mathbb{Z}/2)$ and
  for two such elements the following holds
  $$
    \omega_1 \wedge \omega_2 = (-1)^{n_1 n_2 + \sigma_1 \sigma_2} \omega_2 \wedge \omega_1
    \,.
  $$
  (these signs are as in \cite[II.2.109]{CDF} and \cite[appendix §6]{DeligneFreed}).
  These DG-algebras $\mathrm{CE}(\mathfrak{g})$ are the ``Cartan integrable systems'' of \cite{DAuriaFre82}
  and the ``free differential algebras'' (FDAs) of \cite{vN} \cite[III.6]{CDF}.\footnote{It is however crucial that they are
  not in general free as differential algebras, but just as graded-commutative algebras. In rational homotopy theory one
  also speaks of ``quasi-free'' or ``semi-free'' dg-algebras.}

  By forming the linear dual of a differential graded algebra of finite type, it equivalently becomes
  a differential graded co-algebra.
  That every $L_\infty$-algebra gives a differential co-algebra is originally due to \cite{LadaStasheff93}
  and that this faithfully reflects the original $L_\infty$-algebra is due to \cite{LadaMarkl95},
  see \cite[around def. 13]{SSS09}. In more modern parlance this is due to the Koszul duality between the
  operads for Lie algebras and that for commutative algebras.

  While differential co-algebras are less familar in practice, they have the advantage that they
  immediately reflect also $L_\infty$-algebras not of finite type. This gives a full inclusion
  $$
    \mathrm{sL}^{\mathrm{fin}}_\infty\mathrm{Alg}_{\mathbb{R}}
      \hookrightarrow
    \mathrm{sL}_\infty\mathrm{Alg}_{\mathbb{R}}
  $$
  into the category of possibly degreewise infinite dimensional super $L_\infty$-algebras, { and similarly for the curved case
  $$
    \mathrm{sL}^{\mathrm{cvd},\mathrm{fin}}_\infty\mathrm{Alg}_{\mathbb{R}}
      \hookrightarrow
    \mathrm{sL}^{\mathrm{cvd}}_\infty\mathrm{Alg}_{\mathbb{R}}
    \,.
  $$
  }

\medskip
  Notice that there is a fully faithful inclusion
  $$
    \mathrm{sLieAlg}_{\mathbb{R}}
      \hookrightarrow
    \mathrm{sL}_\infty\mathrm{Alg}_{\mathbb{R}}
  $$
  of ordinary super Lie algebras into super $L_\infty$-algebras, whose image is those $\mathfrak{g}$ for which all generators in
  $\mathrm{CE}(\mathfrak{g})$ are in degree $(1,\sigma)$, for some super-degree $\sigma$.
  Notice that for $\mathfrak{g}$ a super $L_\infty$-algebra structure on a graded super vector space $V$,
 its CE-differential may be co-restricted to its co-unary piece (which sends single generators to single
generators). This is the super cochain complex dual to the \emph{underlying super chain complex} of the $L_\infty$-algebra.
\end{remark}
\begin{example}\label{LineLienAlgebra}
  For $n \in \mathbb{N}$ we write $b^n \mathbb{R}$ for the super $L_\infty$-algebra for which
  $\mathrm{CE}(b^n \mathbb{R})$ has a single generator in bidegree $(n+1,\mathrm{even})$, and
  vanishing differential (the ``line Lie $n$-algebra'').
  For $\mathfrak{g}$ any super $L_\infty$-algebra, we call a homomorphism
  $$
    \mu \;:\; \mathfrak{g} \longrightarrow b^n \mathbb{R}
  $$
 an \emph{$L_\infty$-cocycle of degree $(n+1)$} on $\mathfrak{g}$ with coefficients in $\mathbb{R}$.
  Because dually this is, by definition, a closed element in $\mathrm{CE}(\mathfrak{g})$ of degree $(n+1)$
  (see \cite[section 6.3]{SSS09}).
  \end{example}

We are going to use homotopy theory of super $L_\infty$-algebras. A standard method
to present such are model categories (\cite{Hirschhorn}), but for our puposes here
a more lighweight structure is fully sufficient: that of a Brown category of fibrant objects.
See \cite[Def. 3.54]{NSS12b} for review in a context that we are concerned with here.
\begin{prop}[{\cite{pridham}}]\label{modelstructureOnLInfinityAlgebra}
  There is a model category whose category of fibrant objects
  is precisely the category $\mathrm{sL}_\infty\mathrm{Alg}_{\mathbb{R}}$
  of super $L_\infty$-algebras, and such that on these the
  weak equivalences are the morphisms
  for which the underlying morphism of dual super chain complexes (Remark \ref{striuctureInLInfinity})
  is a quasi-isomorphism, and whose fibrations
  are the morphisms that induce a surjection on the underlying chain complexes. In terms of the dual Chevalley-Eilenberg algebras, this corresponds to an injection on the graded linear subspaces spanned by the generators.
\end{prop}
\begin{proof}
  By \cite[Prop 4.36, Prop. 4.42]{pridham} there is a model category for ordinary (i.e. bosonic)
  $L_\infty$-algebras with these properties. By chasing through the proofs there, one finds that they
  immediately generalize to the super-algebraic situation.
\end{proof}
\begin{example}[{\cite[Prop. 3.5]{FSS13}}]
  \label{homotopyfiberofLinfinityCocycles}
  For $\mathfrak{g}$ any super $L_\infty$-algebra and
  $$
    \mu_{p+2} \;:\;
      \mathfrak{g}
        \longrightarrow
      b^{p+1} \mathbb{R}
  $$
  any homomorphism into the line Lie $(p+1)$-algebra (a $(p+2)$-cocycle on $\mathfrak{g}$,  Example \ref{LineLienAlgebra}),
  the CE-algebra $\mathrm{CE}(\widehat{\mathfrak{g}})$
  of its homotopy fiber $\widehat{\mathfrak{g}} {\xrightarrow{{\tiny \phantom{m}\mathrm{hofib(\mu_{p+2})}\phantom{m}}}}\mathfrak{g}$
  is given from that of $\mathrm{CE}(\mathfrak{g})$ by adjoining a single generator $b_{p+1}$ in degree $p+1$ and
  extending the differential by
  $$
    d_{\widehat{\mathfrak{g}}}\ b_{p+1} = \mu_{p+2}
    \,,
  $$
  i,e.
  $$
    \mathrm{CE}\left(
      \widehat{\mathfrak{g}}
    \right)
    \simeq
    \mathrm{CE}\left(
      \mathfrak{g}
    \right)[b]/(d b_{p+1} = \mu_{p+2})
    \,.
  $$
  Notice that for $\mathfrak{g}$ an ordinary super Lie algebra, and  for $p = 0$, i.e. for the case of a 2-cocycle,
  then $\widehat{\mathfrak{g}}$ thus defined is simply the ordinary central extension of super Lie algebras
  classified by the 2-cocycle. Therefore in the general case we may think of the homotopy fiber $\widehat{\mathfrak{g}}$
  as the \emph{higher central extension} of super $L_\infty$-algebras classified by a super $L_\infty$-cocycle.
\end{example}
We are interested in \emph{supersymmetry} super Lie algebras, and their extensions to super $L_\infty$-algebras.
To be self-contained,
we briefly collect now some basics on Majorana spinors and spacetime supersymmetry algebras,
as well as their interrelation, as the spacetime dimension ranges from 11 down to 9.
We use conventions as in \cite[II.7.1]{CDF}, except for the first two points to follow, where we use the opposite signs.
This means that our Clifford matrices behave as in \cite[II.7.1]{CDF}, the only difference is in a sign when raising a
spacetime index or lowering a spacetime index.
\begin{defn}\label{basicconventions}
{\bf (i)} The Lorentzian spacetime metric is $\eta := \mathrm{diag}  (-1,+1,+1,+1,\cdots)$.
\item
{\bf (ii)} The Clifford algebra relation is $\Gamma_a \Gamma_b + \Gamma_b \Gamma_a  = - 2 \eta_{a b}$;
\item
{\bf (iii)} The timelike index is $a = 0$, the spacelike indices range $a \in \{1,\cdots, d-1\}$.
\item
{\bf (iv)} A unitary Dirac representation of $\mathrm{Spin}(d-1,1)$ is on
$\mathbb{C}^{2^\nu}$ where $d  \in \{2\nu, 2\nu +1\}$,  via
Clifford matrices such that $\Gamma_0^\dagger = \Gamma_0$ and $\Gamma_a^\dagger = - \Gamma_a^\dagger$
for $a \geq 1$.
\item
{\bf (v)} For $\psi \in \mathrm{Mat}_{\nu \times 1}(\mathbb{C})$ a complex spinor, we write
$\overline{\psi} := \psi^\dagger \Gamma_0$
for its Dirac conjugate. If we have Majorana spinors forming a real sub-representation $S$
then restricted to these the Dirac conjugate coincides with the Majorana conjugate
$\psi^\dagger \Gamma_0 = \psi^T C$ (where $C$ is the Charge conjugation matrix).
\end{defn}
As usual we write
$$
  \Gamma_{a_1 \cdots a_p}
    :=
  \tfrac{1}{p!}
  \sum_{\mbox{\tiny permutations $\sigma$}}
    (-1)^{\vert \sigma\vert}
    \Gamma_{a_{\sigma(1)}} \cdots \Gamma_{a_{\sigma(p)}}
$$
for the anti-symmetrization of products of Clifford matrices.
  These conventions imply that all $\Gamma_a$ are self-conjugate with respect to the pairing $\overline{(-)}(-)$, hence
  that
  $$
    \left(\overline{\psi} \Gamma_{a_1\cdots a_p} \psi\right)^\ast
      =
    (-1)^{p(p-1)/2}\,  \overline{\psi} \Gamma_{a_1 \cdots a_p} \psi
  $$
  holds for all $\psi$.
  This means that the following expressions are real numbers
  $$
     \overline{\psi}\psi
\;, \quad
     \overline{\psi}\Gamma_a \psi
  \;, \quad
     i  \, \overline{\psi}\Gamma_{a_1 a_2} \psi
  \;, \quad
     i  \, \overline{\psi}\Gamma_{a_1 a_2 a_3}\psi
   \;, \quad
      \overline{\psi}\Gamma_{a_1\cdots a_4} \psi
    \;, \quad
      \overline{\psi}\Gamma_{a_1\cdots a_5} \psi
    \;, \quad
      i\, \overline{\psi}\Gamma_{a_1\cdots a_6} \psi
    \;, \quad
      \cdots
      \;.
  $$
\begin{defn}\label{SuperMinkowski}
  Given $d\in \mathbb{N}$ and $N$ a real $\mathrm{Spin}(d-1,1)$-representation
  (hence some direct sum of Majorana and Majorana-Weyl representations), the
  corresponding super-Minkowski super Lie algebra
  $$
    \mathbb{R}^{d-1,1\vert N}
      \;\;
      \in
      \mathrm{sLieAlg}_{\mathbb{R}}
  $$
  is the super Lie algebra defined by the fact that its Chevalley-Eilenberg algebra
  is the $(\mathbb{N}, \mathbb{Z}/2)$-bigraded differential-commutative differential algebra
  generated from elements $\{e^a\}_{a = 0}^{d-1}$ in bidegree $(1,\mathrm{even})$
  and from elements $\{\psi^\alpha\}_{\alpha = 1}^{\mathrm{dim}N}$ in bidegree $(1,\mathrm{odd})$
  with differential given by
  $$
    d\psi^\alpha = 0
    \;\;\;\,,
    \;\;\;
    d e^a = \overline{\psi} \wedge \Gamma^a \psi
    \,.
  $$
  Here on the right we use the spinor-to-vector bilinear pairing, regarded as a super 2-form, i.e.
  in terms of the charge conjugation matrix $C$ this is
  $$
    \overline{\psi} \wedge \Gamma^a \psi
      =
    (C \Gamma^a)_{\alpha \beta} \, \psi^\alpha \wedge \psi^\beta
    \,,
  $$
  where summation over repeated indices is understood. Notice that we do not include
  a factor of $\tfrac{1}{2}$ in the definition of $d e^a$.
\end{defn}

\begin{defn}
  \label{CliffordGeneratorsInDimensions9And10And11}
  Let $\{\gamma_a\}_{a = 0}^{d-1}$ be a Dirac representation on $\mathbb{C}^{16}$ of the Lorentzian $d= 9$
  Clifford algebra as above. We obtain a Dirac representation of the $d = 10$ and $d = 11$
  Clifford algebra by taking the following block matrices acting on $\mathbb{C}^{16} \oplus \mathbb{C}^{16}$
  $$
  \Gamma_{a \leq 8}
  :=
 \begin{pmatrix}
       0 & \gamma^a
       \\
       \gamma^a & 0
    \end{pmatrix}
  \;, \qquad
  \Gamma_9
  :=
    \begin{pmatrix}
      0 & \mathrm{I}
      \\
      -\mathrm{I} & 0
    \end{pmatrix}
  \;, \qquad
  \Gamma_{10}
  :=
    \begin{pmatrix}
      i \mathrm{I} & 0
      \\
      0 & -i \mathrm{I}
    \end{pmatrix}
  \,,
$$
where $I$ is the identity matrix.
\end{defn}

\begin{remark}
  \label{SpinRepsInDimensions11And10And9}
  The unique irreducible Majorana representation of $\mathrm{Spin}(10,1)$ is
  of real dimension 32.
  Under the inclusions
  $$
    \mathrm{Spin}(8,1) \hookrightarrow \mathrm{Spin}(9,1) \hookrightarrow \mathrm{Spin}(10,1)
  $$
  this representation branches as
  $$
    \mathbf{32} \mapsto \mathbf{16}\oplus \overline{\mathbf{16}} \mapsto \mathbf{16} \oplus \mathbf{16}
    \,,
  $$
  where in the middle $\mathbf{16}$ and $\overline{\mathbf{16}}$ are the left and right chiral Majorana-Weyl
  representations in 10d, while on the right the $\mathbf{16}$ is again the unique irreducible real
  representation in 9d.
 Under this branching we decompose a Majorana spinor $\psi \in \mathbf{32}$ as
  $$
    \psi
    =
      \begin{pmatrix}
        \psi_1
        \\
        \psi_2
      \end{pmatrix}
  $$
  with $\psi_1 \in \mathbf{16}$ and $\psi_2 \in \overline{\mathbf{16}}$ or $\mathbf{16}$.
\end{remark}
When we consider T-duality and S-duality below, the Clifford algebra generators will receive
various re-interpretations. To make this transparent we introduce the following notation.
\begin{defn}
  \label{IIBCliffordGenerators}
  Define another set of matrices $\{\Gamma_a^{\mathrm{IIB}}\}_{a = 0}^{9}$
  by
  $$
    \Gamma_a^{\mathrm{IIB}}
     :=
    \left\{
      \begin{array}{cl}
         \Gamma_a & \vert \; a \leq 8\;,
         \\
         \begin{pmatrix}
              0 & \mathrm{I}
              \\
              \mathrm{I} & 0
           \end{pmatrix}
         &
         \vert \; a = 9\;.
      \end{array}
    \right.
  $$
  For emphasis we write the original matrices also as $\Gamma_a^{\mathrm{IIA}} := \Gamma_a$, for $a \leq 9$.

  Moreover we also write
  $$
    \sigma_1 := \Gamma_9
    \;\,,\;\;\;\;\;
    \sigma_2 := -\Gamma_{9} \Gamma_{10}
    \;\,,\;\;\;\;\;
    \sigma_3 := \Gamma_{10}
    \,.
  $$
\end{defn}
\begin{remark}
  \label{The12GammaMatricesDoNotFormACliffordAlgebra}
  The matrices $\{\Gamma^{\mathrm{IIB}}_a\}_{a = 0}^9$ in Def. \ref{IIBCliffordGenerators} do not represent
  a Clifford algebra, but the product of any even number of them represents the correct such
  product acting on $\mathbf{16} \oplus \mathbf{16}$. For instance $\exp(\omega^{a b}\Gamma^{\mathrm{IIB}}_{a b})$
  are the elements of the $\mathrm{Spin}(d-1,1)$-representation on $\mathbf{16}\oplus \mathbf{16}$.
  Also, for \emph{odd} $p = 2k+1$, each of the pairings
  $$
    \overline{\psi}\Gamma^{\mathrm{IIB}}_{a_1 \cdots a_p}\psi
      =
    \psi^\dagger \Gamma^{\mathrm{IIB}}_0 \Gamma^{\mathrm{IIB}}_{a_1 \cdots a_p} \psi
  $$
  is the sum of the corresponding pairings on two copies of $\mathbf{16}$.
\end{remark}
\begin{remark}\label{SpinorToVectorPairingIIB}
 {\bf (i)} By Def. \ref{IIBCliffordGenerators} and Def. \ref{CliffordGeneratorsInDimensions9And10And11} we have
  $$
    \Gamma_9^{\mathrm{IIB}} = i \, \Gamma_9 \Gamma_{10} = \Gamma_9 \Gamma_{11}
  $$
  or, equivalently,
  $$
    \Gamma_9 = i \, \Gamma_9^{\mathrm{IIB}} \Gamma_{10}
    \,.
  $$
  This simple relation is crucial in the proof of T-duality in Theorem \ref{TheTDualityIsoOnLInfinityCocycles}.

  \item {\bf (ii)} This relation also makes it manifest that $\Gamma_9^{\mathrm{IIB}}$ commutes not only with all
  $\Gamma^{\mathrm{IIB}}_{a b}$ for $a,b \leq 8$, but also with all $\Gamma_a^{\mathrm{IIB}}\Gamma_9^{\mathrm{IIB}}$.
  Consequently, $\Gamma_{10}$ as well as $\Gamma_9$ are invariant under the
  IIB Spin-action, in that (with the notation in Def. \ref{IIBCliffordGenerators})
  $$
    \exp(-\omega^{a b}\Gamma^{\mathrm{IIB}}_{a b}) \; \sigma_i \; \exp(\omega^{a b} \Gamma^{\mathrm{IIB}}_{a b})
      =
    \sigma_i
  $$
  for $i \in \{1,2,3\}$.

  \item {\bf (iii)} Conversely, rotation in the $(9,10)$-plane leaves all the $\Gamma_a^{\mathrm{IIB}}$ invariant, in that
  $$
    \exp(- \tfrac{\alpha}{4} \Gamma_9 \Gamma_{10}) \,\Gamma_a^{\mathrm{IIB}}\, \exp(\tfrac{\alpha}{4} \Gamma_9 \Gamma_{10})
      \;=\;
    \Gamma_a^{\mathrm{IIB}}
    \,.
  $$
\end{remark}
\begin{defn}\label{def:SuperMinkowskiIn10dAnd9d}
  We write $\mathbb{R}^{10,1\vert \mathbf{32}}$ (in M-theory) and
  $\mathbb{R}^{9,1\vert \mathbf{16} + \overline{\mathbf{16}}}$ (in type IIA)
  and
  $\mathbb{R}^{9,1\vert \mathbf{16} + {\mathbf{16}}}$ (in type IIB)
  and
    $\mathbb{R}^{8,1\vert \mathbf{16} + {\mathbf{16}}}$ (in common 9d)
    for the super-Minkowski super Lie algebras (Def. \ref{SuperMinkowski})
    given by the Spin representations of Remark \ref{SpinRepsInDimensions11And10And9}.
\end{defn}

We now start with an observation
relating  the algebraic structures of super-Minkowski spacetimes (Def. \ref{SuperMinkowski}) in dimensions
 nine and ten.
\begin{prop}\label{IIBAsExtension}
  The bilinear spinor-to-vector pairings $\psi \mapsto (\overline{\psi} \Gamma^a \psi)$ in dimensions 11 and 10
  constitute 2-cocycles on the super-Minkowski super Lie algebras of one dimension
  lower (Def. \ref{def:SuperMinkowskiIn10dAnd9d}):
    \item {\bf (i)} $c_2^M := \overline{\psi} \wedge \Gamma_{10} \psi
      \;\;\;\;\in \mathrm{CE}(\mathbb{R}^{9,1\vert \mathbf{16} + \overline{\mathbf{16}}})$;
    \item {\bf (ii)} $c_2^{\mathrm{IIA}} :=
        \overline{\begin{pmatrix} \psi_1 \\ \psi_2\end{pmatrix}}
           \wedge \Gamma_9^{\mathrm{IIA}}
        \begin{pmatrix} \psi_1 \\ \psi_2\end{pmatrix}
        \;\;
        \in
        \mathrm{CE}(\mathbb{R}^{8,1\vert \mathbf{16}+ \mathbf{16}})
    $;
    \item {\bf (iii)} $c_2^{\mathrm{IIB}} :=
        \overline{\begin{pmatrix}\psi_1 \\ \psi_2 \end{pmatrix}}
           \wedge \Gamma_9^{\mathrm{IIB}}
        \begin{pmatrix}{\psi_1}\\{ \psi_2}\end{pmatrix}
        \;\;
        \in
        \mathrm{CE}(\mathbb{R}^{8,1\vert \mathbf{16}+ \mathbf{16}})
    $.

\vspace{1mm}
  \noindent Moreover, these 2-cocycles classify consecutive central super Lie algebra extensions
  of super-Minkowski spacetime from 9d to 11d (Def. \ref{def:SuperMinkowskiIn10dAnd9d})
  in that we get the following diagram of super $L_\infty$-algebras
  $$
    \xymatrix@C=4em{
      && \mathbb{R}^{10,1\vert \mathbf{32}}
      \ar[d]^-{\pi_{10}}
      \\
      \mathbb{R}^{9,1\vert \mathbf{16} + \mathbf{16}}
      \ar[dr]^{\pi_9^{\mathrm{IIB}}}
      & &
      \mathbb{R}^{9,1\vert \mathbf{16}  + \overline{\mathbf{16}}}
      \ar[dl]_{\pi_9^{\mathrm{IIA}}}
      \ar[r]^-{c^M_2} & b \mathbb{R}
      \\
      & \mathbb{R}^{8,1 \vert \mathbf{16} + \mathbf{16}}
      \ar[dl]^{c_2^{\mathrm{IIB}}}
      \ar[dr]_{c_2^{\mathrm{IIA}}}
      \\
      b \mathbb{R} && b \mathbb{R}
    }
  $$
  where each ``hook''
  $$
    \xymatrix{
      \widehat{\mathfrak{g}}
      \ar[d]^-\pi
      \\
      \mathfrak{g}
      \ar[r]^{\omega_2}
      &
      b\mathbb{R}
    }
  $$
  corresponds to the central extension $\widehat{\mathfrak{g}}$ of $\mathfrak{g}$ classified by the
  2-cocycle $\omega_2$ (i.e. is a homotopy fiber sequence of super $L_\infty$-algebras, according
  to Example \ref{homotopyfiberofLinfinityCocycles}).
\end{prop}
\begin{proof}
To see that the given 2-forms are indeed cocycles: they are trivially closed (by Def. \ref{SuperMinkowski}),
and so all that matters is that we have a well-defined super-2-form in the first place.
Since the $\psi^\alpha$ are in bidegree $(1,\mathrm{odd})$, they all commute with each other. and hence
the condition is that the spinor-to-vector pairing is symmetric. This is the case for Majorana spinors.
(This is a simple but deep fact, highlighted before in \cite[(2.4)]{CAIB00}, \cite[Prop. 4.5]{FSS13}).

Now we consider the extensions. Notice that for $\mathfrak{g}$ any super Lie algebra (of finite dimension),
and for $\omega \in \wedge^2\mathfrak{g}^\ast$ a Lie algebra 2-cocycle on it,  the Lie algebra extension
$\widehat{\mathfrak{g}}$ that this classifies is neatly characterized in terms of its dual Chevalley-Eilenberg
algebra; it is simply the original CE algebra with one new generator $e$ in degree $(1, {\rm even})$
adjoined, and with the differential of $e$ taken to be $\omega$:
$$
  \mathrm{CE}(\widehat{\mathfrak{g}})
   =
  \big((\mathrm{CE}(\mathfrak{g}) \otimes \langle e\rangle), ~d e = \omega\big)
  \,.
$$
Hence in the case of $\omega = c_2^{\mathrm{IIA}}$ we identify the new generator with $e^9$.
Furthermore, we  see
from Prop. \ref{SpinorToVectorPairingIIB} that the equation $d e^9 = c_2^{\mathrm{IIA}}$ is
precisely what distinguishes the CE-algebra of $\mathbb{R}^{8,1\vert \mathbf{16}+ \mathbf{16}}$
from that of $\mathbb{R}^{9,1\vert \mathbf{16} + \overline{\mathbf{16}}}$. This follows
by Def. \ref{SuperMinkowski}, Def. \ref{CliffordGeneratorsInDimensions9And10And11}, and
using the fact that the spinors all have the same underlying representation space by
Remark \ref{SpinRepsInDimensions11And10And9}.

The other two cases are directly analogous.
\end{proof}

\section{Double dimensional reduction}
\label{SectionDoubleDimensionalReduction}

In this section we formalize the concept of
\emph{double dimensional reduction} of branes on a circle bundle
(originally due to \cite{DuffHoweInamiStelle87}) as an equivalence between
higher homotopy algebra structures (theorem \ref{ddReduction} below), which is at the same time
a formalization and a generalization of earlier treatments (as for instance \cite{FF2}).
This may be iterated to deal with reducing on successive circles, i.e. on tori. The free and cyclic
loop space construction will play a major role in this formalization (see remark \ref{rem:SullivanModelForFreeLoopSpace} below),
which in terms of $L_\infty$-cohomolgy
translates into passage to Hochschild and cyclic $L_\infty$-homology (see remark \ref{JonesTheorem} below).
We close this section below in remark \ref{GeometricInterpretationOfdd} with spelling out how the
$L_\infty$-theoretic statement of theorem \ref{ddReduction} indeed reflects the informal idea of
double dimensional reduction.

Here we make use of one of the main result of \emph{rational homotopy theory} (see e.g. \cite{Hess06} for review); namely, that
the rational homotopy type of a
connected topological space $X$
with nilpotent fundamental group is completely and faithfully encoded into a suitable $L_\infty$-algebra $\mathfrak{l}(X)$, which one may think of as being the $\infinity$-Lie algebra of the loop group $\Omega X$ (itself an ``$\infty$-group'', namely a grouplike $A_\infty$-space).
The \emph{Sullivan model} dg-algebras of rational homotopy theory \cite{Sullivan77}, are nothing but the
CE-algebras of (particularly good representatives) of connective $L_\infty$-algebras:

\begin{center}
\begin{tabular}{|c|c|c|c|}
  \hline
  topological space & loop $\infty$-group &  $L_\infty$-algebra & Sullivan model
  \\
  \hline
  $X$ & $\Omega X$ &  $\mathfrak{l}(X)$ & $\mathrm{CE}(\mathfrak{l}X)$
  \\
  \hline
\end{tabular}
\end{center}
The observation that traditional rational homotopy theory sits inside the homotopy theory of $L_\infty$-algebras
was implicit alredy in the original \cite{Quillen69}, but was made fully explicit only in \cite{Hinich},
on which the model \cite{pridham} is based, which we use in Proposition \ref{modelstructureOnLInfinityAlgebra}.
A review of rational homotopy theory from this modern perspective is in \cite[section 2]{BuijsFelixMurillo12}. For more
exposition in our context see also \cite[appendix A]{FSS16}.

\medskip

We recall the following fact due to Vigu{\'e} et al.:
\begin{remark}[Sullivan model for free and cyclic loop spaces]\label{rem:SullivanModelForFreeLoopSpace}
  The {\it cyclic loop space} of a topological space $X$ is the homotopy quotient $\mathcal{L}X / S^1$, where
  $\mathcal{L}X := \mathrm{Maps}(S^1,X)$ denotes the free loop space of $X$ and the $S^1$-action is given by rotation of loops.
If $X$ is simply connected, and   $( \wedge^\bullet V, d_{{}_X})$ is a
 minimal Sullivan model for the rationalization of $X$, then
   a Sullivan model for the rationalization of the free loop space $\mathcal{L}X$ of $X$ is given by \cite{VigueSullivan}
      $$
        \mathrm{CE}(\mathfrak{l}(\mathcal{L}X)) =
        \big(\wedge^\bullet ( V \oplus s V ),~ d_{{}_{\mathcal{L}X}}\big)
        \,,
      $$
      where $s V$ is $V$ with degrees shifted down by one, and with $d_{\mathcal{L}X}$
      acting for $v \in V$ as
      $
        d_{{}_{\mathcal{L}X}} \, v = d_{{}_X}  v$, $
        d_{{}_{\mathcal{L}X}}  \, s v = - s  \, d_{{}_X}  v
        $
      where on the right $s \colon V \to s V$ is extended uniquely as a graded derivation.
    A Sullivan model for the rationalization of the cyclic space $\mathcal{L}X / S^1$
     is given by \cite{VigueBurghelea}
     $$
       \mathrm{CE}(\mathfrak{l}( \mathcal{L}X/S^1 ))
       =
       \Big( \wedge^\bullet\big( V \oplus s V \oplus \langle \omega_2\rangle\big), ~d_{{}_{\mathcal{L}X/S^1}} \Big)
     $$
     with
     $
       d_{{}_{\mathcal{L}X/S^1}}  \omega_2 = 0
     $
     and with $d_{{}_{\mathcal{L}X/S^1}}$ acting on $w \in \wedge^1 V \oplus sV $ as
     $
       d_{{}_{\mathcal{L}X/S^1}} w = d_{{}_{\mathcal{L}X}}  w + \omega_2 \wedge s w
       $.
     Moreover, the canonical sequence of $L_\infty$-homomorphisms
     $$
         \mathfrak{l}(\mathcal{L}X)
           \longrightarrow
          \mathfrak{l}(\mathcal{L}X/S^1)
           \longrightarrow
          b \mathbb{R}
     $$
     is a rational model for the homotopy fiber sequence
     $
         \mathcal{L}X
     \to
         \mathcal{L}X/S^1
     \to
         B S^1
     $
     that exhibits the homotopy quotient.
\end{remark}
\begin{remark}[Jones' theorem]\label{JonesTheorem}
  The Sullivan models for free and cyclic loop spaces as in remark \ref{rem:SullivanModelForFreeLoopSpace} will appear
  (in a super $L_\infty$-theoretic generalization to follow in Definition \ref{def:cyclification})
  below in Theorem \ref{ddReduction} purely as part of an $L_\infty$-theoretic incarnation of the
  concept of double dimensional reduction. All the more does it seem interesting to briefly recall the close relation of
  these dg-algebras to Hochschild homology and to cyclic homology. This is known as \emph{Jones' theorem}
  \cite{Jones87} (see \cite{Loday11} for expository review):
Let $X$ be a simply connected topological space. Then the ordinary cohomology $H^\bullet(-)$ of its free loop space
  is isomorphic to the Hochschild homology $\mathrm{HH}_\bullet$ of its dg-algbra $C^\bullet(X)$ of singular cochains:
  $$
    H^\bullet(\mathcal{L}X) \simeq \mathrm{HH}_\bullet(C^\bullet(X))
    \,.
  $$
  Similarly, the $S^1$-equivariant cohomology of the free loop space, hence the cohomology of the cyclic
  loop space, is isomorphic to the cyclic homology $\mathrm{HC}_\bullet(-)$ of the cochains:
  $$
    H^\bullet(\mathcal{L}X/S^1)
      \simeq
    \mathrm{HC}_\bullet(C^\bullet(X))
    \,.
  $$
  Specifically if $X$ happens to carry the structure of a smooth manifold, then the dg-algebra
  of singular chains with real coefficients is
  quasi-isomorphic to the de-Rham dg-algebra of smooth differential forms $\Omega^\bullet(X)$, and hence in this case
  $$
    H^\bullet(\mathcal{L}X,\mathbb{R})
      \simeq
    \mathrm{HH}_\bullet(\Omega^\bullet(X))
 \qquad
\text{and}
\qquad
    H^\bullet(\mathcal{L}X/S^1, \mathbb{R})
      \simeq
    \mathrm{HC}_\bullet(\Omega^\bullet(X))
    \,.
  $$
  Therefore in this case Remark \ref{rem:SullivanModelForFreeLoopSpace} together with the central
  theorem of Quillen-Sullivan rational homotopy says that the cochain cohomology
  of the Chevalley-Eilenberg complexes of the $L_\infty$-algebras $\mathfrak{l}(\mathcal{L}X)$
  and $\mathfrak{l}(\mathcal{L}X/S^1)$ coincide with the Hochschild homology and the
  cyclic homology, respectively, of the CE-algebra of the $L_\infty$-algebra $\mathfrak{l}(X)$:
  $$
    H^\bullet( \mathrm{CE}( \mathfrak{l}(\, \mathcal{L}X\, ) ) )
      \simeq
    \mathrm{HH}_\bullet( \mathrm{CE}( \mathfrak{l}(\, X \,) )   )
  \qquad
\text{  and}
\qquad
    H^\bullet( \mathrm{CE}( \mathfrak{l}(\, \mathcal{L}X/S^1\, ) ) )
      \simeq
    \mathrm{HC}_\bullet( \mathrm{CE}( \mathfrak{l}(\, X \,) )   )
    \,.
  $$
\end{remark}
Remark \ref{rem:SullivanModelForFreeLoopSpace} motivates the following generalization:
\begin{defn}\label{def:cyclification}
Let $\mathfrak{h}$ be a super $L_\infty$-algebra.
\begin{enumerate}
\item
The \emph{free loop algebra} of $\mathfrak{h}$
is the super $L_\infty$-algebra defined by
\[
        \mathrm{CE}(\mathfrak{L}\mathfrak{h}) =
        \big(\wedge^\bullet ( \mathfrak{h}^\ast \oplus s \mathfrak{h}^\ast ),~
        d_{\mathfrak{Lh}}\big)
        \,,
\]
where $s \mathfrak{h}^\ast$ is a copy of $\mathfrak{h}^\ast$ with bi-degrees shifted by $(-1,\mathrm{even})$,
with $d_{\mathfrak{Lh}}$
acting as
$$
  \begin{aligned}
  d_{\mathfrak{Lh}}  v & = d_{\mathfrak{h}}  v
  \\
  d_{\mathfrak{Lh}}  s v & = - s d_{\mathfrak{h}}  v
  \end{aligned}
$$
for  all $v \in \mathfrak{h}^\ast$,
where  $s$ is the graded derivation of degree $(-1,\mathrm{even})$ which sends the original generators $v \in \mathfrak{h}^\ast$
to their shifted image $s v \in s \mathfrak{h}^\ast$.

\item
The \emph{cyclification} of $\mathfrak{h}$ is the $L_\infty$-algebra $\mathfrak{L}\mathfrak{h}/\mathbb{R}$ defined by
\[
\mathrm{CE}(\mathfrak{L}\mathfrak{h}/\mathbb{R})= \big( \wedge^\bullet( \mathfrak{h}^\ast \oplus s \mathfrak{h}^\ast \oplus \langle \omega_2\rangle),~ d_{\mathfrak{L}\mathfrak{h}/\mathbb{R}} \big)
\]
     with
$$
       d_{\mathfrak{L}\mathfrak{h}/\mathbb{R}} \; \omega_2 = 0
$$
     and with $d_{\mathfrak{L}\mathfrak{h}/\mathbb{R}}$
     acting on $w \in  \mathfrak{h} \oplus s\mathfrak{h} $ as
$$
       d_{\mathfrak{L}\mathfrak{h}/\mathbb{R}} \, w
         =
       d_{\mathfrak{L}\mathfrak{h}} \, w + \omega_2 \wedge s w
       \,.
$$
\end{enumerate}
\end{defn}
\begin{lemma}\label{FunctorialLoopAlgebra}
  The free loop algebra construction from Definition \ref{def:cyclification}
  extends to a functor on the category of super $L_\infty$-algebras
  $$
    \mathfrak{L}(-)
    \;\colon\;
    sL_\infty\mathrm{Alg}_{\mathbb{R}}
      \longrightarrow
    sL_\infty\mathrm{Alg}_{\mathbb{R}}
  $$
  by taking a super $L_\infty$-homomorphism $f \colon \mathfrak{h}_1 \to \mathfrak{h}_2$
  to the homomorphism $\mathfrak{L}f$ whose dual $(\mathfrak{L}f)^\ast$ is given on generators
  $v \in \wedge^1 \mathfrak{h}_2^\ast$ by
  $$
    \begin{aligned}
      (\mathfrak{L}f)^\ast(v) & = f^\ast(v)
      \\
      (\mathfrak{L}f)^\ast(s_2v) & = s_1 f^\ast(v)
      \,.
    \end{aligned}
  $$
Similarly the cyclification operation from Definition \ref{def:cyclification}
  extends to a functor
  $$
    \mathfrak{L}(-)/\mathbb{R}
    \;\colon\;
    sL_\infty\mathrm{Alg}_{\mathbb{R}}
      \longrightarrow
    sL_\infty\mathrm{Alg}_{\mathbb{R}}
  $$
  by taking a homomorphism $f$ to the homomorphism $\mathfrak{L}f/\mathbb{R}$
  whose dual is given on $v$ and $s v$ by $(\mathfrak{L}f)^\ast$ as before, and which sends the copy of the generator $\omega_2$
  in $\mathrm{CE}(\mathfrak{L}\mathfrak{h}_2/\mathbb{R})$ to the generator of the same name in
  $\mathrm{CE}(\mathfrak{L}\mathfrak{h}_1/\mathbb{R})$.
\end{lemma}
\begin{proof}
  To see that $\mathfrak{L}(-)$ is functorial on homomorphisms of graded algebras
  it is sufficient to observe that the relation $(\mathfrak{L}f)^\ast(s_2 v) = s_1 (\mathfrak{L}f)^\ast(v)$
  on generators implies that $(\mathfrak{L}f)^\ast$ commutes with $s$ generally. For instance on binary wedge products
  of generators we get
  $$
    \begin{aligned}
      (\mathfrak{L}f)^\ast( s(v_1 \wedge v_2) )
      & =
      (\mathfrak{L}f)^\ast( (s v_1) \wedge v_2 + (-1)^{\mathrm{deg}(v_1+1)} (s v_2)    )
      \\
      & =
      (\mathfrak{L}f)^\ast(s v_1) \wedge (\mathfrak{L}f)^\ast(v_2)
      +
      (-1)^{\mathrm{deg}(v_1)+1}
      (\mathfrak{L}f)^\ast(v_1) \wedge (\mathfrak{L}f)^\ast(s v_2)
      \\
      & =
      s((\mathfrak{L}f)^\ast(v_1)) \wedge (\mathfrak{L}f)^\ast(v_2)
      +
      (-1)^{\mathrm{deg}((\mathfrak{L}f)^\ast(v_1))+1}
      (\mathfrak{L}f)^\ast(v_1) \wedge s((\mathfrak{L}f)^\ast(v_2))
      \\
      & =
      s (\mathfrak{L}f)^\ast( v_1 \wedge v_2 )
      \,.
    \end{aligned}
  $$

  It remains to see that $(\mathfrak{L}f)^\ast$ respects the differential. On the unshifted generators $v$
  this is so because  $f^\ast$ does respect the differential. For shifted generators it
  follows by this computation:
  $$
    \begin{aligned}
    (\mathfrak{L}f)^\ast(d_{\mathfrak{L}\mathfrak{h}_2} s_2 v)
    & =
    -(\mathfrak{L}f)^\ast(s_2 d_{\mathfrak{L}\mathfrak{h}_2} v)
    \\
    & =
    - s_1 (\mathfrak{L}f)^\ast(d_{\mathfrak{h}_2} v)
    \\
    & =
    - s_1 f^\ast ( d_{\mathfrak{h}_2} v )
    \\
    & =
    - s_1   d_{\mathfrak{h}_1} f^\ast (v)
    \\
    & =
    - s_1  d_{\mathfrak{L}\mathfrak{h}_1} f^\ast(v)
    \\
    & =
    d_{\mathfrak{L}\mathfrak{h}_1} s_1 (\mathfrak{L}f)^\ast(v)
    \\
    & =
    d_{\mathfrak{L}\mathfrak{h}_1}  (\mathfrak{L}f)^\ast (s_2 v)
    \end{aligned}
  $$

  Finally, that also $(\mathfrak{L}f/\mathbb{R})^\ast$ respects the differential
  follows by this computation:
  $$
    \begin{aligned}
    (\mathfrak{L}f/\mathbb{R})^\ast(d_{\mathfrak{L}\mathfrak{h}_2/\mathbb{R}} w)
    &=
    (\mathfrak{L}f/\mathbb{R})^\ast(d_{\mathfrak{L}\mathfrak{h}_2} w + \omega_2 \wedge s w )
    \\
    & =
    (\mathfrak{L}f)^\ast(d_{\mathfrak{L}\mathfrak{h}_2} w)
    +
    (\mathfrak{L}f/\mathbb{R})^\ast(\omega_2 \wedge s w)
    \\
    & =
    d_{\mathfrak{L}\mathfrak{h}_2} (\mathfrak{L}f)^\ast(w)
    +
    \omega_2\wedge (\mathfrak{L}f)^\ast(s w)
    \\
    &=
    d_{\mathfrak{L}\mathfrak{h}_2} (\mathfrak{L}f)^\ast(w)
    +
    \omega_2\wedge s (\mathfrak{L}f)^\ast(w)
    \,.
    \end{aligned}
  $$
  where we used the previous statements about $(\mathfrak{L}f)^\ast$.
\end{proof}

The two super $L_\infty$-algebras from Definition \ref{def:cyclification} are related as follows.

\begin{prop}\label{def:FreeLoopQuotientFiberSequence}
For any super $L_\infty$-algebra $\mathfrak{h}$, its free loop algebra and its cyclification
(Def. \ref{def:cyclification}) sit in a homotopy fiber sequence of the form
$$
  \xymatrix{
    \mathfrak{Lh}
      \ar[r]
    &
    \mathfrak{Lh}/\mathbb{R}
      \ar[r]^-{\omega_2}
      &
    b\mathbb{R}
  }\;,
$$
where $\omega_2$ is the 2-cocycle given by the CE-element of the same name in Def. \ref{def:cyclification}.
\end{prop}
\begin{proof}
  The second morphism is a
  fibration (according to Proposition \ref{modelstructureOnLInfinityAlgebra}),
  by the surjection evident from their definition. Moreover, from the form of the
  differentials, the first morphism
  is directly seen to be the ordinary fiber of the second.
  Hence it models the homotopy fiber.
\end{proof}

The key observation now, noticed in a particular case in \cite[Prop. 3.8]{FSS16}, is that the cyclification
of coefficients is what formalizes ``double dimensional reduction'' of branes \cite{DuffHoweInamiStelle87}, i.e. the
joint process of reducing spacetime dimension and reducing in parallel the dimension
of branes in this spacetime, or dually, the cohomological degrees of their charges.
Or rather: those branes/charges that ``wrap'' the dimension being reduced are to
reduce in parallel, while those that do not wrap should just descend.
We formalize this by Theorem \ref{ddReduction} below, see remark \ref{GeometricInterpretationOfdd} further below
for discussion of the physical interpretation.
To state the formalization, we first need to make explicit the following basic construction:

\begin{defn}\label{SliceCategory}
  Given any super $L_\infty$-algebra $\mathfrak{b} \in sL_\infty\mathrm{Alg}_{\mathbb{R}}$,
  then the \emph{slice over $\mathfrak{g}$} is the category
  $(sL_\infty \mathrm{Alg}_{\mathbb{R}})_{/\mathfrak{b}}$
  whose objects are super $L_\infty$-homomorphisms
  $\mathfrak{g} \stackrel{\phi}{\to} \mathfrak{b}$
  into $\mathfrak{b}$, and whose morphisms are super $L_\infty$-homomorphisms $\mathfrak{g}_1 \to \mathfrak{g}_2$
  that respect the morphisms down to $\mathfrak{b}$ in that they make the diagrams shown on the
  right commute:
  $$
    \mathrm{Hom}_{/\mathfrak{b}}\left( (\mathfrak{g}_1,\phi_1), (\mathfrak{g}_2, \phi_2) \right)
    \;:=\;
    \left\{
      \raisebox{20pt}{
      \xymatrix{
        \mathfrak{g}_1 \ar@{-->}[rr] \ar[dr]_{\phi_1} && \mathfrak{g}_2 \ar[dl]^{\phi_2}
        \\
        & \mathfrak{b}
      }
      }
    \right\}\;.
  $$
\end{defn}
\begin{example}\label{FunctorialHofib}
  Consider the operation that sends any super $L_\infty$-homomorphism of the form
  $$
    \mu_{p+2} \colon \mathfrak{g} \longrightarrow b^{p+1}\mathbb{R}
  $$
  (i.e. a super $L_\infty$-cocycle, according to Example \ref{LineLienAlgebra}) to
  its homotopy fiber $\widehat{ \mathfrak{g}}$ represented via Example
   \ref{homotopyfiberofLinfinityCocycles}
  $$
    \mathrm{CE}\left(
      \widehat{\mathfrak{g}}
    \right)
    :=
    \mathrm{CE}(\mathfrak{g})[b_{p+1}]/(d b_{p+1} = \mu_{p+2})\;.
  $$
This operation  extends to a functor on the slice category (Def. \ref{SliceCategory})
over $b^{p+1}\mathbb{R}$
  $$
    \mathrm{hofib}
    \;\colon\;
    (sL_\infty \mathrm{Alg}_{\mathbb{R}})_{/b^{p+1}\mathbb{R}}
    \longrightarrow
    s L_\infty \mathrm{Alg}_{\mathbb{R}}\;,
  $$
  by taking any super $L_\infty$-homomorphism $f$ in
  $$
      \raisebox{20pt}{
      \xymatrix{
        \mathfrak{g}_1 \ar[rr]^f \ar[dr]_{\mu^1_{p+2}} && \mathfrak{g}_2 \ar[dl]^{\mu^2_{p+2}}
        \\
        & b^{p+1}\mathbb{R}
      }
      }
  $$
  to the homomorphism $\mathrm{hofib}(f)$ whose dual $(\mathrm{hofib}(f))^\ast$ is given on
  $\mathrm{CE}(\mathfrak{g}_2)$ by $f^\ast$ and sends the generator $b^2_{p+1}$ to $b^1_{p+1}$.
  This respects the differential on the original generators because $f^\ast$ does, and
  it respects the differential on the new generator because

  $$
    \begin{aligned}
      (\mathrm{hofib}(f))^\ast( d_{\widehat{\mathfrak{g}_2}} b^2_{p+1} )
      &=
      (\mathrm{hofib}(f))^\ast( \mu^2_{p+2} )
      \\
      & =
      f^\ast ( \mu^2_{p+2} )
      \\
      & =
      \mu^1_{p+1}
      \\
      & =
      d_{\widehat{\mathfrak{g}_1}} b^1_{p+1}
      \\
      & =
      d_{\widehat{\mathfrak{g}_1}} (\mathrm{hofib}(f))^\ast b^2_{p+1}
      \,,
    \end{aligned}
  $$
  where the third equality used is equivalent to the commutativity of the triangular diagram above.
\end{example}

\begin{theorem}[$L_\infty$-theoretic double dimensional reduction and oxidation]
  \label{ddReduction}
  Let $\mathfrak{g}$ and $\mathfrak{h}$ be  super $L_\infty$-algebras,
  { such that $\mathrm{CE}(\mathfrak{h})$ has no generators in degree 1}.
  Let moreover
  $c_2 : \mathfrak{g} \to b \mathbb{R}$ be a 2-cocycle (Definition \ref{LineLienAlgebra}) and
  $\pi : \widehat{\mathfrak{g}} \to \mathfrak{g}$
  be the
  central extension classified by $c_2$, according to Example \ref{homotopyfiberofLinfinityCocycles}.
  Then there is a bijection
  $$
    \xymatrix{
      \mathrm{Hom}( \widehat {\mathfrak{g}}, \mathfrak{h} )
        \ar@{<-}@<+5pt>[rr]^-{\mbox{\tiny oxidation}}
        \ar@<-5pt>[rr]_-{\mbox{\tiny reduction}}^-\simeq
        &&
       \mathrm{Hom}_{/b \mathbb{R}}( \mathfrak{g}, \mathfrak{Lh}/\mathbb{R} )
    }
  $$
  between super $L_\infty$-homomorphisms out of $\widehat{\mathfrak{g}}$ into $\mathfrak{h}$
  and super $L_\infty$-homomorphism over $b\mathbb{R}$ from $\mathfrak{g}$ (in the sense of
  Definition \ref{SliceCategory})
  to the cyclification $\mathfrak{L}\mathfrak{h}/\mathbb{R}$ of $\mathfrak{h}$ (Def. \ref{def:cyclification}):
  $$
    \left(
      \xymatrix{
      \widehat{\mathfrak{g}}
        \ar[r]&
      \mathfrak{h}
      }
      \
    \right)
     \;\;\;\;\longleftrightarrow\;\;\;\;
    \left(
      \raisebox{20pt}{
      \xymatrix{
        \mathfrak{g}
          \ar[rr]
          \ar[dr]_{c_2}
        &&
        \mathfrak{Lh}/\mathbb{R}
        \ar[dl]^{\omega_2}
        \\
        & b\mathbb{R}
      }}
    \right)\;.
  $$
  {
  More generally, if we consider $\mathfrak{g}$ and $\mathfrak{h}$ as curved $L_\infty$-algebras
  with vanishing curvature, according to Definition \ref{superLInfinityAlgebra}, then
  the bijection holds without any condition on $\mathfrak{h}$ as a bijection between hom-sets of
  possibly curved super $L_\infty$-homomorphism. Moreover, in this case}
  the bijection is \emph{natural} in its arguments (i.e. compatible with pre- and postcomposition with
  curved super $L_\infty$-homomorphisms). In other words, the functors from Example \ref{def:cyclification}
  (Lemma \ref{FunctorialLoopAlgebra}) and Example \ref{FunctorialHofib}
  then form an adjoint pair (e.g. \cite[chapter 3]{Borceux94})
  with $\mathrm{hofib}$ left adjoint to the cyclification functor $\mathfrak{L}(-)/\mathbb{R}$ from Def. \ref{def:cyclification}:
  $$
    \xymatrix{
      sL_\infty \mathrm{Alg}^{{\mathrm{cvd}}}_{\mathbb{R}}
        \ar@<-6pt>[rr]_{\mathfrak{L}(-)/\mathbb{R}}^{\bot}
        \ar@<+6pt>@{<-}[rr]^{\mathrm{hofib}}
      &&
     (sL_\infty \mathrm{Alg}^{{ \mathrm{cvd}}}_{\mathbb{R}})_{/b\mathbb{R}}
    }
  $$
  {
  Finally, the unit $\eta_{(\mathfrak{g}), \omega_2}$ of this adjunction,
  $$
    \xymatrix{
       \mathfrak{g}
         \ar[dr]_{c_2}
         \ar[rr]^-{\eta_{\mathfrak{g},c_2}}
         &&
       \mathfrak{L} {\widehat {\mathfrak{g}}}/\mathbb{R}
       \ar[dl]^{\omega_2}
       \\
       & b\mathbb{R}
    }
  $$
  (i.e. the image under the above correspondence of the identity on some $\widehat{\mathfrak{g}}$)
  is dually given by the map that
  \begin{enumerate}
    \item sends unshifted generators to themselves, except for $e$
  (the extra generator with $d_{\widehat {\mathfrak{g}}} e = c_2$ according to example \ref{homotopyfiberofLinfinityCocycles})
  which is sent to zero,
   \item sends
  all shifted generators to zero, except for the shift $s(e)$ of the extra generator,
  that instead is sent to minus the algeba unit $s(e) \mapsto -1$ {(this makes it a curved $L_\infty$-morphism)}.
  \end{enumerate}
  }
\end{theorem}

\begin{proof}
We discuss the adjunction for curved $L_\infty$-homomorphism. The statement for the non-curved case then follows
immediately as a special case.

First consider the bijection as such. Dually we need to show that there is an identification
 $$
    \mathrm{Hom}( \mathrm{CE}(\mathfrak{h}), \mathrm{CE}(\widehat {\mathfrak{g}})  )
     \simeq
    \mathrm{Hom}^{\mathrm{CE}(b\mathbb{R})/}( \mathrm{CE}(\mathfrak{Lh}/\mathbb{R}), \mathrm{CE}(\mathfrak{g})  )
    \,.
  $$
  between homomorphisms of dg-algebras (not augmented) out of $\mathrm{CE}(\mathfrak{h})$
  into  $\mathrm{CE}(\widehat{\mathfrak{g}})$
  and dg-algebra homomorphisms under $\mathrm{CE}(b\mathbb{R})$ out of
  $\mathrm{CE}(\mathfrak{L}\mathfrak{h}/\mathbb{R})$ into $\mathrm{CE}(\mathfrak{g})$
$$
  \left(
    \xymatrix{
      \mathrm{CE}(\mathfrak{h})
      \ar[r]
      &
      \mathrm{CE}(\widehat{\mathfrak{g}})
    }
  \right)
     \;\;\;\;\longleftrightarrow\;\;\;\;
  \left(
    \raisebox{24pt}{
    \xymatrix{
      & \mathrm{CE}(b\mathbb{R})
      \ar[dr]^{c_2}
      \ar[dl]_{\omega_2}
      \\
      \mathrm{CE}(\mathfrak{L}\mathfrak{h}/\mathbb{R})
      \ar[rr]
      &&
      \mathrm{CE}(\mathfrak{g})
    }
    }
  \right)
  \,.
$$
Since, after forgetting the differential, $\mathrm{CE}(\mathfrak{h})$ is a free graded polynomial algebra,
we may write
\[
\mathrm{CE}(\mathfrak{h})=\big(\mathbb{R}[\{x_p\}],\,
d_{{}_{\mathrm{CE}(\mathfrak{h})}}x_i= P_i(\{x_p\})\big)
\]
for suitable polynomials $P_i$ in the variables $\{x_p\}$. Then a DGCA homomorphism from $\mathrm{CE}(\mathfrak{h})$ to $\mathrm{CE}(\widehat {\mathfrak{g}})$ is just a collection of elements $\alpha_p$ in $\mathrm{CE}(\widehat {\mathfrak{g}})$, one for each generator of $\mathrm{CE}(\mathfrak{h})$, such that
\[
d_{{}_{\mathrm{CE}}(\widehat {\mathfrak{g}})}\alpha_i = P_i(\{\alpha_p\})\;.
\]
Since $\hat{\mathfrak{g}}$ is the  the central extension of $\mathfrak{g}$ classified by $\omega_2$, we have
\[
\mathrm{CE}(\widehat {\mathfrak{g}})=\bigl( \mathrm{CE}(\mathfrak{g})[e], \,
d_{{}_{\mathrm{CE}(\hat{\mathfrak{g}})}}e=\gamma_2\bigr)\;,
\]
where $\gamma_2$ is the image of the generator $\omega_2$ of $\mathbb{R}[\omega_2] \simeq \mathrm{CE}(b\mathbb{R})$ under the map $c_2:\mathbb{R}[\omega_2]\to \mathrm{CE}(\mathfrak{g})$. Since $e$ has degree 1, we then see that there is a vector space decomposition
\[
\mathrm{CE}(\widehat {\mathfrak{g}})=\mathrm{CE}(\mathfrak{g})\oplus \big(e\wedge \mathrm{CE}(\mathfrak{g})\big)\;.
\]
 Consequently, we can uniquely decompose all the
  elements $\alpha_{p}$ as
    \begin{equation}
  \label{decomp}
    \alpha_{p} =   \beta_{p} - e \wedge \tilde{\alpha}_{p-1}
    \,,
  \end{equation}
  with $\tilde{\alpha}_{p-1}, \beta_{p} \in \mathrm{CE}(\mathfrak{g})$.
  We denote this by
  \begin{equation}\label{fiberintegration}
    \pi_\ast (\alpha_{p}) := \tilde{\alpha}_{p-1}
    \;\;\;\;\;\mbox{and}
    \;\;\;\;\;
    \alpha_p|_{\mathfrak{g}} := \beta_{p}
  \end{equation}
  in order to amplify the geometric interpretation.
Forgetting the differential, we have
\[
\mathrm{CE}(\mathfrak{L}\mathfrak{h}/\mathbb{R})=\mathbb{R}[\{x_p,y_{p-1}\},\omega_2]\;,
\]
so that the assignment
\begin{equation}\label{TheAdjunctOfAMapFromHatgToh}
x_p\mapsto \beta_{p}; \qquad y_{p-1}\mapsto \tilde{\alpha}_{p-1}
\end{equation}
precisely defines a morphism of graded commutative algebras under $\mathbb{R}[\omega_2]$ from $\mathrm{CE}(\mathfrak{L}\mathfrak{h}/\mathbb{R})$ to $\mathrm{CE}(\mathfrak{g})$. That is, the image of the additional generator $\omega_2$ of $\mathrm{CE}(\mathfrak{L}\mathfrak{h}/\mathbb{R})$ is prescribed to be $\gamma_2$ by the requirement of having a morphism under $\mathbb{R}[\omega_2]$. We are therefore left with checking that this is indeed a morphism of DGCAs. In terms of the generators $\{x_p,y_{p-1}\}$ and $\omega_2$, the differential in $\mathrm{CE}(\mathfrak{L}\mathfrak{h}/\mathbb{R})$ reads

\vspace{-7mm}
\[
d_{{}_{\mathrm{CE}(\mathfrak{L}\mathfrak{h}/\mathbb{R})}}x_i=
P_i(\{x_p\})+\omega_2\wedge y_{p-1}\;,
\quad
d_{{}_{\mathrm{CE}(\mathfrak{L}\mathfrak{h}/\mathbb{R})}}y_{i-1}=
-\sum_jy_{j-1} {\small \frac{\partial P_i}{\partial x_j}(\{x_p\})}\;,
\quad
d_{{}_{\mathrm{CE}(\mathfrak{L}\mathfrak{h}/\mathbb{R})}}\omega_2=0\;.
\]

\vspace{-2mm}
\noindent The fact that the above assignment is a morphism of DGCAs then  follows from
the matching
\begin{align*}
P_i(\{\beta_p\})-
e\wedge  \sum_j\tilde{\alpha}_{j-1} {\small \frac{\partial P_i}{\partial x_j}}\big(\{\beta_p\}\big)
&=
P_i(\{\beta_p-e\wedge \tilde{\alpha}_p\})
\\
&=P_i(\{\alpha_p\})
\\
&=d_{{}_{\mathrm{CE}(\widehat {\mathfrak{g}})}}\alpha_i
\\
&=d_{{}_{\mathrm{CE}(\widehat {\mathfrak{g}})}}(\beta_{i}- e\wedge \tilde{\alpha}_{i-1} )
\\
&= d_{{}_{\mathrm{CE}({\mathfrak{g}})}} \beta_{i}
-\gamma_2\wedge \tilde{\alpha}_{i-1}  + e\wedge d_{{}_{\mathrm{CE}( {\mathfrak{g}})}}\tilde{\alpha}_{i-1}\;.
\end{align*}

This establishes the bijection for any fixed $\widehat{\mathfrak{g}}$ and $\mathfrak{h}$. For this
bijection to be \emph{natural} we need to show that for every morphism
$$
  \xymatrix{
    \mathfrak{g}_2
    \ar[dr]_{(\omega_2)_2}
    \ar[rr]^f
    &&
    \mathfrak{g}_1
    \ar[dl]^{(\omega_2)_1}
    \\
    & b \mathbb{R}
  }
$$
in $(sL_\infty \mathrm{Alg}_{\mathbb{R}})_{/b\mathbb{R}}$  (Def. \ref{SliceCategory}) and every morphism
$$
  \xymatrix{
    \mathfrak{h}_1
    \ar[rr]^g
    &&
    \mathfrak{h}_2
  }
$$
in $sL_\infty \mathrm{Alg}_{\mathbb{R}}$
the following diagram (of functions between hom-sets) commutes:
$$
  \xymatrix{
    \mathrm{Hom}( \widehat{\mathfrak{g}_1}, \mathfrak{h}_1 )
    \ar[rr]^-{\simeq}_-{\mathrm{reduction}}
    \ar[d]_{g\circ (-)\circ \mathrm{hofib}(f)}
    &&
    \mathrm{Hom}_{/b\mathbb{R}}(\mathfrak{g}_1, \mathfrak{L}\mathfrak{h}_1/\mathbb{R})
    \ar[d]^{ \mathfrak{L}g/\mathbb{R} \circ (-) \circ f}
    \\
    \mathrm{Hom}( \widehat{\mathfrak{g}_2}, \mathfrak{h}_2 )
    \ar[rr]^-{\simeq}_-{\mathrm{reduction}}
    &&
    \mathrm{Hom}_{/b\mathbb{R}}(\mathfrak{g}_2, \mathfrak{L}\mathfrak{h}_2/\mathbb{R})\,,
  }
$$
where the vertical maps are those given by pre- and postcomposition, as indicated.

By unwinding the definition, one finds that indeed both ways of going around this square
take a homomorphism $\phi \colon \widehat{\mathfrak{g}_1} \to \mathfrak{h_1}$
to the homomotphism $\mathfrak{g}_2 \to \mathfrak{L}\mathfrak{h}_2/\mathbb{R}$
which is given, dually, for any generator $v \in \wedge^1 \mathfrak{h}_1^\ast$ by

$$
  \begin{aligned}
    v & \mapsto f^\ast((\phi^\ast(g^\ast(v)))\vert_{\mathfrak{g}_1})
    \\
    s v & \mapsto f^\ast((\pi_1)_\ast(\phi^\ast(g^\ast(v))))
  \end{aligned}
$$
where we are using (\ref{fiberintegration}).

This establishes the adjunction.
{ Finally we show that the adjunction unit $\eta$ is as claimed.
One way to see this is to apply the above corespondence to the identity morphism on $\widehat{\mathfrak{g}}$.
But it is instructive to spell this out:
First observe that the map as claimed is indeed well defined: Its dual incarnation respects the differential on $e$
because
$$
  d_{\mathfrak{L}\widehat{\mathfrak{g}}/\mathbb{R}}
  e
  =
  c_2 + \omega_2 \wedge s(e)
  \mapsto
  c_2 + c_2 \wedge (-1)
  = 0
  \,,
$$
Showing that this is indeed the adjunction unit
is equivalent to showing that image of any $\widehat{\mathfrak{g}} \overset{\phi}{\longrightarrow} \mathfrak{h}$
under the adjunction correspondence is equal to the composite
$$
  \xymatrix{
    \mathfrak{g}
      \ar[rr]^-{\eta_{\mathfrak{g}, c_2}}
      &&
    \mathfrak{L}\widehat{\mathfrak{g}}/\mathbb{R}
      \ar[rr]^-{\mathfrak{L}\phi/\mathbb{R}}
    &&
    \mathfrak{L}\mathfrak{h}/\mathbb{R}
  }
  \,.
$$
This now follows by direct inspection. By the above, the dual of this composite sends any
unshifted generator $c_{p}$ to
$$
  c_p \mapsto \alpha_p = \beta_p - e \wedge \tilde \alpha_{p-1} \mapsto \beta_p
$$
and any shifted generator $s(c_p)$ to
$$
  s(c_p)
    \mapsto
  s(\alpha_p) = s(\beta_p) - s(e) \wedge \tilde \alpha_{p-1} - e \wedge s(\tilde \alpha_{p-1})
    \mapsto
  -(-1) \wedge \tilde \alpha_{p-1}
   =
     \tilde \alpha_{p-1}
   \,.
$$
This is indeed as in assignment (\ref{TheAdjunctOfAMapFromHatgToh}) above.}

\vspace{-6mm}
\end{proof}

\begin{remark}[Geometric interpretation of $L_\infty$-algebraic dimensional reduction]
 \label{GeometricInterpretationOfdd}
  The general theory of adjoint functors (e.g. \cite[chapter 3]{Borceux94}) provides insight as to the geometric nature of the
  super $L_\infty$-algebraic formalization of dimensional reduction  from Theorem \ref{ddReduction}.
  Namely given any pair of adjoint functors $L \dashv R$, then
  the unit of the adjunction $\eta_x : x \to R L x$ is such that
  the natural bijection between hom-sets
  $$
    \mathrm{Hom}(L x, y) \overset{\simeq}{\longrightarrow} \mathrm{Hom}(x,R x)
  $$
  is given by sending any morphism of the form $\phi : L x \to y$ to the composite
  $$
    x \overset{\eta_x}{\longrightarrow} R L x
     \overset{R \phi}{\longrightarrow}
    R y
    \,.
  $$
  Specified to the situation in Theorem \ref{ddReduction}, this means that the
  $L_\infty$-theoretic double dimensional reduction
  of a super $L_\infty$-homomorphism
  $$
    \widehat{\mathfrak{g}} \overset{\phi}{\longrightarrow} \mathfrak{h}
  $$
  on a central $\mathbb{R}$-extension $\widehat{\mathfrak{g}}$ of some super
   $L_\infty$-algebra $\mathfrak{g}$ is the following composite:
  $$
  \xymatrix{
    \mathfrak{g}
      \ar[r]^-{\eta_{\mathfrak{g}}}
      &
    \mathfrak{L} \widehat{\mathfrak{g}} / \mathbb{R}
      \ar[rr]^-{\mathfrak{L}(\phi)/\mathbb{R}}
      &&
    \mathfrak{L}\mathfrak{h}/\mathbb{R}
    }
    \,.
  $$
  In terms of the geometric interpretation via rational homotopy theory from Remark \ref{rem:SullivanModelForFreeLoopSpace}
  the morphism $\eta_{\mathfrak{g}}$ here has the following interpretation:

  Let $\widehat X \to X$ be a principal circle bundle. Then there is a map
  $$
    X \longrightarrow \mathcal{L}\widehat X/ S^1
  $$
  which sends each point of $X$ to the loop that winds around the circle fiber over that point, at unit parameter speed.
  As a map to the free loop space $\mathcal{L}\widehat{X}$ this would not be well defined unless the circle bundle
  were trivial, because by definition of principal circle bundles its fibers are identified with the typical
  fiber (the circle) only up to rigid rotation of that circle. But this is precisely the relation that is
  divided out by passing to the cyclified space $\mathfrak{L}\widehat{X}/S^1$, which makes the assignment
  of points to the loops that wind around their fibers be well defined.

  Hence given any map of spaces $\widehat{X} \overset{f}{\longrightarrow} H$, then
  we may pass to the induced map on loops in $\widehat{X}$ modulo rigid rotation, and then
  precompose with the above fiber-assigning map
  $$
  \xymatrix{
    X
     \ar[r]
     &
    \mathcal{L}\widehat{X}/S^1
      \ar[r]^{\mathcal{L}f/S^1}
      &
    \mathcal{L}H/S^1
    }
    \,.
  $$
  Under the Quillen-Sullivan functor from spaces to their associated $L_\infty$-algebras in rational homotopy theory,
  this is the $L_\infty$-algebraic construction above.

  And this shows just how this formalizes the intuitive picture of double dimensional reduction:
  For let $\Sigma_{p+1}$ be a manifold of dimension $p+1$
  and let $\Sigma_{p+1} \longrightarrow X$ be the worldvolume of some $p$-brane in $X$. Thinking of
  $f : \widehat{X} \to H$
  as classifying a background field for $(p+1)$-branes on $\widehat{X}$
  (for instance for $H = K(\mathbb{Z},p+3)$, the classifying space for ordinary cohomology)
  then the dimensionally reduced coupling term is given by the composite
  $$
  \xymatrix{
    \Sigma_{p+1}
      \ar[r] &
    X
      \ar[r]
      &
      \mathcal{L}\widehat{X}/S^1
      \ar[r]^{\mathcal{L}f/S^1}
      &    \mathcal{L}H/S^1
    }
    \,.
  $$
  To see what this does, consider what happens locally over some chart $U$ on which the circle extension $\widehat{X} \to X$
  is topologically trivial. Then by the ordinary product/hom adjunction $S^1 \times (-) \dashv [S^1,-] = \mathcal{L}(-)$ this is equivalently
  the composite
  $$
    \Sigma_{p+1} \times S^1
      \longrightarrow
    \widehat U
      \overset{f|_U}{\longrightarrow}
    H
    \,.
  $$
  But this is nothing than the value of the background field $f$ not on the $p$-brane worldvolume
  $\Sigma_{p+1}$, but on the worldvolume $\Sigma_{p+1} \times S^1$ of a $(p+1)$-brane, which
  ``wraps'' the circle fiber in $\widehat U = U \times S^1$.
  This is precisely the physical picture of double dimensional reduction, originally due to \cite{DuffHoweInamiStelle87}.

\end{remark}

\section{The brane supercocycles}
\label{SecBraneCocycles}

We now naturally associate to systems of super $p$-brane species certain
 supercocycles taking values in super $L_\infty$-algebras arising from spheres and related
 topological spaces/spectra.
This follows the geometric approach to cocycles in supergravity, as in
\cite{DAuriaFre82} \cite{CDF}.
The algebras for type IIA and type IIB that we obtain may also be found
in \cite{Ca} \cite{CAIB00} \cite{IIBAlgebra} \cite{DFGT}.

\begin{defn}\label{SphereFourAlgebra}
 Write $\mathfrak{l} S^4 \in \mathrm{sL}_\infty \mathrm{Alg}$ for the Chevalley-Eilenberg algebra
 which is the minimal Sullivan model of the 4-sphere according to \cite{Sullivan77}:
 $$
   \mathrm{CE}(\mathfrak{l}S^4)
    =
   \left\{
     \begin{aligned}
       d g_4 & = 0
       \\
       d g_7 &= -\tfrac{1}{2} g_4 \wedge g_4
     \end{aligned}
   \right\}
   \,.
 $$
\end{defn}
\begin{defn}\label{M2M5Cocycles}
  On the super Lie algebra $\mathbb{R}^{10,1\vert \mathbf{32}}$  (Def. \ref{SuperMinkowski}, Remark \ref{SpinRepsInDimensions11And10And9})
  define the cochains
  $$
    \begin{aligned}
      \mu_{{}_{M2}}
        & :=
      \tfrac{i}{2} \; \left(\overline{\psi} \wedge \Gamma_{a_1 a_2} \psi \right) \wedge e^{a_1} \wedge e^{a_2}\;,
      \\
      \mu_{{}_{M5}}
        &:=
      \tfrac{1}{5!}  \;
      \left(\overline{\psi} \wedge \Gamma_{a_1 \cdots a_5} \psi \right)
      \wedge e^{a_1} \wedge \cdots \wedge e^{a_5}\;,
    \end{aligned}
  $$
  where the indices run in the set $\{0,1,\cdots, 10\}$.
\end{defn}
\begin{prop}\label{TheMBraneS4Cocycle}
  The elements in Def. \ref{M2M5Cocycles} satisfy
  $$
    d\mu_{{}_{M5}} = -\tfrac{1}{2} \mu_{{}_{M2}} \wedge \mu_{{}_{M2}}
    \,,
  $$
  hence consititute a super $L_\infty$-algebra homomorphism to the 4-sphere (Def. \ref{SphereFourAlgebra}):
  $$
    (\mu_{{}_{M2}}, ~\mu_{{}_{M5}})
    \;:\;
    \mathbb{R}^{10,1 \vert \mathbf{32}}
      \longrightarrow
    \mathfrak{l}S^4
    \,.
  $$
\end{prop}
\begin{proof}
  Using the differential relations $d \psi^\alpha = 0$ and $d e^a = \overline{\psi}\Gamma^a \psi$
  from Def. \ref{SuperMinkowski},
  we have
  $$
    \begin{aligned}
      d \mu_{{}_{M5}}
        & =
        \tfrac{1}{5!} d \left(\left(\overline{\psi} \Gamma_{a_1 \cdots a_5} \psi\right)\wedge e^{a_1} \wedge \cdots \wedge e^{a_5}\right)
        \\
        & =
          \tfrac{1}{4!}
          \underset{= 3 \left(\overline{\psi}\Gamma_{[a_1 a_2}\psi\right) \left(\overline{\psi}\Gamma_{a_3 a_4]}\psi \right) }{\underbrace{
          \left(\overline{\psi} \Gamma_{[a_1 \cdots a_4 a]} \psi \right)
           \left(
         \overline{\psi}
           \Gamma^a
         \psi
         \right)
         }}
         \wedge e^{a_1} \wedge \cdots e^{a_4}
         \\
         & =
         -\tfrac{1}{2}
         \left( \tfrac{i}{2} (\overline{\psi}\wedge \Gamma_{a_1 a_2} \psi) \wedge e^{a_1}\wedge e^{a_2} \right)
         \wedge
         \left( \tfrac{i}{2} (\overline{\psi} \wedge \Gamma_{a_3 a_4} \psi) \wedge e^{a_3}\wedge e^{a_4} \right)
         \\
         & = -\tfrac{1}{2}\mu_{{}_{M2}} \wedge \mu_{{}_{M2}}\;,
    \end{aligned}
  $$
  where the equality under the brace is the Fierz identity from \cite[(3.27a)]{CDF}.
\end{proof}
\begin{remark}\label{M2M5Origin}
  In fact there is a stronger statement: A priori the WZW term for the M5-brane is a 7-cocycle
  not on $\mathbb{R}^{10,1\vert \mathbf{32}}$, but on the super Lie 3-algebra $\mathfrak{m}2\mathfrak{brane}$ which is
  its higher extension (according to example \ref{homotopyfiberofLinfinityCocycles}) by the 4-cocycle $\mu_{M2}$ \cite{FSS13}:
  $$
    \xymatrix{
      \mathfrak{m}2\mathfrak{brane}
      \ar[d]_{\mathrm{hofib}(\mu_{M2})}
      \ar[rr]^{\mu'_{M5}}
      &&
      b^6 \mathbb{R}
      \\
      \mathbb{R}^{10,1\vert \mathbf{32}}
      \ar[dr]_-{\mu_{M2}}
      \\
      & b^3 \mathbb{R}
    }
    \,.
  $$
  Given such a situation, $\mathfrak{m}2\mathfrak{brane}$ is exhibited as a $b^2 \mathbb{R}$-principal $\infty$-bundle
  over $\mathbb{R}^{10,1\vert \mathbf{32}}$
  \cite{NSS12}
  and one may ask for the $b^2 \mathbb{R}$-equivariant homotopy
  descent of $\mu'_{M5}$ down to the base space $\mathbb{R}^{9,1\vert \mathbf{16}+ \overline{\mathbf{16}}}$.
  As discussed in \cite{FSS15} this exists and is given, up to equivalence of $L_\infty$-cocycles,
  by the $S^4$-valued cocycle $(\mu_{M2}, \mu_{M5})$ from Proposition \ref{TheMBraneS4Cocycle}.
\end{remark}
\begin{example}[{\cite[prop. 3.8]{FSS16}}]\label{ReductionOfTheMBraneCocycles}
The double dimensional reduction according to Theorem \ref{ddReduction} of the M2/M5-brane cocycle from Definition \ref{TheMBraneS4Cocycle}
and Remark \ref{M2M5Origin}
is a cocycle of the form
$$
  \xymatrix{
    \mathbb{R}^{10,1 \vert \mathbf{32}}
    \ar[drr]_{\mu_{{}_{D0}}}
    \ar[rrrr]^{\left\{{\mu_{{}_{F1}}^{\mathrm{IIA}}, ~\mu_{{}_{NS5}}^{\mathrm{IIA}}},~  { \mu_{{}_{D2}}, ~\mu_{{}_{D4}} } \right\}}
    &&&&
    \mathfrak{L}S^4/b\mathbb{R}
    \ar[dll]^{\omega_2}
    \\
    &&
    b \mathbb{R}
  }
  \,.
$$
\end{example}

However, we highlight that there is gauge enhancement: more cocycles
appear after dimensional reduction, notably the D6-brane cocycle,
with all cocycles being assembled together in a model for rational
twisted complex topological K-theory  in cohomological degrees 0 and 1.
As a precise statement this is Proposition \ref{IIACocyclesCoeffsInKU} below, for which we need the following
definitions:
\begin{defn}[{\cite[section 4]{FSS16}}]
  \label{RationalTwistedku}
  Write $\mathfrak{l}( \mathrm{KU}/BU(1) )$ for the $L_\infty$-algebra
  with CE-algebra
  $$
    \mathrm{CE}(\mathfrak{l}( \mathrm{KU}/BU(1) ))
    :=
    \big(\mathbb{R}[
      h_3,
      \{\omega_{2p}\}_{p \in \mathbb{Z}}],\, ~d \omega_{2p+2} = h_3 \wedge \omega_{2p}
    \big)
  $$
  and $\mathfrak{l}( \Sigma\mathrm{KU}/BU(1) )$ for the $L_\infty$-algebra
  with CE-algebra
  $$
    \mathrm{CE}(\mathfrak{l}( \Sigma\mathrm{KU}/BU(1) ))
    :=
    \big(\mathbb{R}[
      h_3,
      \{\omega_{2p+1}\}_{p \in \mathbb{Z}}], ~d \omega_{2p+3} = h_3 \wedge \omega_{2p+1}
    \big)
    \,.
  $$
\end{defn}

Starting with cocycles in M-theory on $\R^{9, 1|{\bf 32}}$ we get
cocycles in type IIA on $\R^{8, 1| {\bf 16} + \overline{\bf 16}}$ by integration over
the fiber $(\pi_{10})_*$ of the rational circle bundle $S^1_\R \to \R^{9, 1|{\bf 32}}
\overset{\pi_{10}}{\longrightarrow} \R^{8, 1| {\bf 16} + \overline{\bf 16}}$, established
in \cite[Prop. 4.5]{FSS13}.
\begin{defn}\label{IIACocycles}
  We denote the components of the double dimensional reduction of the M-brane cocycles
  as follows, using the notation in Eq. (\ref{fiberintegration}):
  $$
    \begin{aligned}
      \mu_{{}_{D0}}
        &:=
      \overline{\psi}\Gamma^{10} \psi
      \\
&       =
      \overline{\psi} \Gamma_{10} \psi\;,
      \\
      \mu_{{}_{F1}}^{\mathrm{IIA}}
        &:=
      (\pi_{10})_\ast( \mu_{{}_{M2}} )
       \\
     &   = i  \left( \overline{\psi} \wedge \Gamma_a \Gamma_{10} \psi \right) \wedge e^a\;,
      \\
      \mu_{{}_{D2}}
        &:=
        (\mu_{{}_{M2}})|_{8+1}
        \\
      &  =
        \tfrac{i}{2}  \left(\overline{\psi} \wedge \Gamma_{a_1 a_2} \psi \right) \wedge e^{a_1} \wedge e^{a_2}\;,
        \\
      \mu_{{}_{D4}}
        &:=
      (\pi_{10})_\ast( \mu_{{}_{M5}} )
      \\
     &  =
      +\tfrac{1}{4!}
         \left(
         \overline{\psi}
           \wedge
          \Gamma_{a_1\cdots a_4}
          \Gamma_{10}\psi
          \right)
          \wedge e^{a_1} \wedge \cdots \wedge e^{a_4}\;.
    \end{aligned}
    $$
    Furthermore, consider the following elements in
    $\mathrm{CE}(\mathbb{R}^{9,1\vert \mathbf{16}+\overline{\mathbf{16}}})$:
    $$
    \begin{aligned}
      \mu_{{}_{D6}}
       &  :=
       \tfrac{i}{6!}
      \left(\overline{\psi} \Gamma_{a_1 \cdots a_6}\psi \right) \wedge e^{a_1} \wedge \cdots \wedge e^{a_6}
      ,
      \\
      \mu_{{}_{D8}}
       & :=
      \tfrac{1}{8!}
        \left(
        \overline{\psi}
           \Gamma_{a_1 \cdots a_8}
           \Gamma_{10}
        \psi
        \right)
        \wedge
        e^{a_1}\wedge \cdots \wedge e^{a_8}
      ,
      \\
      \mu_{{}_{D10}}
&        :=
      \tfrac{i}{10!}
        \left(
        \overline{\psi}
           \Gamma_{a_1 \cdots a_{10}}
        \psi
        \right)
        \wedge
        e^{a_1}\wedge \cdots \wedge e^{a_{10}}
    \,,
  \end{aligned}
  $$
  where the indices run through $\{0,1,\cdots, 9\}$
\end{defn}
\begin{prop}\label{IIACocyclesCoeffsInKU}
  The elements in Def. \ref{IIACocycles} satisfy the following differential
  conditions
  $$
    d \mu_{{}_{F1}}^{\mathrm{IIA}} = 0\;,
  $$
  $$
    d \mu_{{}_{D 0}} = 0
    \;\,,\;\;
    d \mu_{{}_{D 10 }} = \mu_{{}_{F1}}^{\mathrm{IIA}} \wedge \mu_{{}_{D8}} = 0\;,
  $$
  $$
    d \mu_{{}_{D2 }} = \mu_{{}_{F1}}^{\mathrm{IIA}} \wedge \mu_{{}_{D0}}
    \;\,,\;\;
    d \mu_{{}_{D4}} = \mu_{{}_{F1}}^{\mathrm{IIA}} \wedge \mu_{{}_{D2}}
    \;\,,\;\;
    d \mu_{{}_{D6}} = \mu_{{}_{F1}}^{\mathrm{IIA}} \wedge \mu_{{}_{D4}}
    \;\,,\;\;
    d \mu_{{}_{D8}} = \mu_{{}_{F1}}^{\mathrm{IIA}} \wedge \mu_{{}_{D6}}\;,
  $$
  hence they constitute an $L_\infty$-cocycle shown as the top morphism of the following diagram:
  $$
    \xymatrix{
      \mathbb{R}^{9,1\vert \mathbf{16}+\overline{\mathbf{16}}}
      \ar@/^2pc/[rrrrr]^{\{ \mu_{{}_{F1}}^{\mathrm{IIA}},~ \mu_{{}_{D0}}, ~\mu_{{}_{D2}}, ~
      \mu_{{}_{D4}},~ \mu_{{}_{D6}},~ \mu_{{}_{D8}}, ~\mu_{{}_{D10}} \}}
      \ar[rrrr]_{\mathfrak{L}(\mu_{{}_{M2}},~ \mu_{{}_{M5}})/b\mathbb{R}}
      &&&&
      \mathfrak{L}S^4/b\mathbb{R}
      \ar[dr]
      &
      \mathfrak{l}( \mathrm{ku}/B U(1) )
      \ar[d]
      \\
      && && &
      \mathfrak{l}( \mathrm{ku}_{\leq 6}/BU(1) )\;.
    }
  $$
\end{prop}
  Here for emphasis we also displayed the double dimensional reduction of the M-brane cocycle
  from example \ref{ReductionOfTheMBraneCocycles}, and
  indicated that these coincide with the IIA cocycles on the F1, the D0, D2, and D4. Note that
   the M-brane cocycles also produce the NS5, but not the D6 and higher, which appear only
   in 10d  (``gauge enhancement'').

\begin{proof}
The first equation follows for instance from $d \mu_{{}_{M2}} = 0$ under dimensional reduction.
The two equations in the second row follow trivially, by $d \psi^\alpha = 0$ and since there is no bosonic $11$-form
on $\mathbb{R}^{9,1}$.
Regarding the equations in the third row:
Using $d \psi^\alpha = 0$ and $d e^a = \overline{\psi} \Gamma^a \psi$ (Def. \ref{SuperMinkowski})
we find that they are equivalently rewritten as follows:
$$
  \begin{array}{crcl}
    &
    d \mu_{{}_{D2}}
    &=& \mu_{{}_{F1}}^{\mathrm{IIA}} \wedge \mu_{{}_{D 0}}
    \\
    \Leftrightarrow
    &
    - i   \left(\overline{\psi} \Gamma_{a_1 a}\psi\right)
      \wedge
      \left(
          \overline{\psi} \Gamma^{a}\psi
      \right)
     & =  &
     + i  \left(  \overline{\psi} \Gamma_{a_1} \Gamma_{10} \psi \right)
        \wedge
     \overline{\psi} \Gamma_{10} \psi
     \\
     \\
    &
    d \mu_{{}_{D4}}
     & = &
    \mu_{{}_{F1}}^{\mathrm{IIA}} \wedge \mu_{{}_{D2}}
    \\
    \Leftrightarrow
    &
    - \tfrac{1}{3!}
     \left(
       \overline{\psi} \Gamma_{[a_1 a_2 a_3] a} \Gamma_{10} \psi
     \right)
     \left(
       \overline{\psi}
         \Gamma^a
       \psi
     \right)
     & = &
     -\tfrac{1}{2}
     \left(
       \overline{\psi} \Gamma_{[a_1 a_2} \psi
     \right)
     \left(\overline{\psi}\Gamma_{a_3]} \Gamma_{10}\psi\right)
     \\
     \\
     &
     d \mu_{{}_{D6}} &=& \mu_{{}_{F1}}^{\mathrm{IIA}}\wedge \mu_{{}_{D4}}
     \\
     \Leftrightarrow &
     - \tfrac{i}{5!} \left(\overline{\psi}\Gamma_{[a_1 \cdots a_5] a} \psi \right)
      \left(\overline{\psi}\Gamma^a \psi\right)
      &=&
      \tfrac{i}{4!}
        \left(
        \overline{\psi}
          \Gamma_{[a_1 \cdots a_4}
          \Gamma_{10}
        \psi
        \right)
        \left(
        \overline{\psi}
          \Gamma_{a_5]}
          \Gamma_{10}
        \psi
        \right)
     \\
     \\
     & d \mu_{{}_{D8}} & = & \mu_{{}_{F1}}^{\mathrm{IIA}} \wedge \mu_{{}_{D6}}
     \\
     \Leftrightarrow &
     -\tfrac{1}{7!}
       \left(
       \overline{\psi}
         \Gamma_{[a_1\cdots a_7] a} \Gamma_{10}
       \psi
       \right)
       \left(
       \overline{\psi}
         \Gamma^a
       \psi
       \right)
       & = &
       -\tfrac{1}{6!}
       \left(
       \overline{\psi}
         \Gamma_{[a_1 \cdots a_6}
       \psi
       \right)
       \left(
       \overline{\psi}
         \Gamma_{a_7]}
         \Gamma_{10}
       \psi
       \right)\;.
  \end{array}
$$
That these conditions hold may be checked to be equivalent to the statement of \cite[expressions (6.8) with coefficients as found above (6.9)]{CAIB00}.

Alternatively, our main theorem \ref{TheTDualityIsoOnLInfinityCocycles} below implies that the
F1/Dp-cochains for type IIA are cocycles precisely if those for type IIB are, which we state below as
def. \ref{IIBBraneCEElements}. This implies that the above Fierz identities in type IIA hold precisely
if those in type IIB hold. The latter has been checked independently
in \cite[section 2]{IIBAlgebra}, see prop. \ref{JointCocycleforIIB} below.

\end{proof}

We now consider a similar construction for the type IIB theory.

\begin{defn}\label{IIBBraneCEElements}
  Define the following elements in the Chevalley-Eilenberg algebra of the
  type IIB super-Minkowski spacetime $\mathbb{R}^{9,1\vert \mathbf{16}+\mathbf{16}}$ (Def. \ref{def:SuperMinkowskiIn10dAnd9d}):
  $$
    \begin{aligned}
    c_2^{\mathrm{IIB}}
      &:=
      \overline{\psi} \Gamma_9^{\mathrm{IIB}} \psi = \overline{\psi} \Gamma^9_B \psi\;,
      \\
    \mu_{F1}^{\mathrm{IIB}} &:= i \left(\overline{\psi} \Gamma_a^{\mathrm{IIB}} \Gamma_{10} \psi \right) \wedge e^a
     \\
      \mu_{{}_{D1}}
        &:=
      i \left(\overline{\psi} \Gamma_a^{\mathrm{IIB}} \Gamma_9 \psi\right) \wedge e^{a}\;,
      \\
    \mu_{{}_{D3}}
      &:=
    \tfrac{1}{3!}
    \left(
      \overline{\psi}
        \Gamma^{\mathrm{IIB}}_{a_1 \cdots a_3} (\Gamma_9 \Gamma_{10})
      \psi
    \right)
    \wedge
    e^{a_1} \wedge \cdots \wedge e^{a_3}\;,
    \\
    \mu_{{}_{D5}}
      &:=
     \tfrac{i}{5!}
       \left(
       \overline{\psi}
         \Gamma^{\mathrm{IIB}}_{a_1 \cdots a_5} \Gamma_9
       \psi
      \right)
      \wedge
      e^{a_1}\wedge \cdots \wedge e^{a_5}\;,
      \\
      \mu_{{}_{D7}}
        & =
              \tfrac{1}{7!}
        \left(
       \overline{\psi}
         \Gamma^{\mathrm{IIB}}_{a_1 \cdots a_7} (\Gamma_9 \Gamma_{10})
       \psi
       \right)
       \wedge
       e^{a_1}
       \wedge
         \cdots
       \wedge
       e^{a_7}\;,
     \\
     \mu_{{}_{D9}} & =
        \tfrac{i}{9!}
    \left(
    \overline{\psi}
      \Gamma^{\mathrm{IIB}}_{a_1\cdots a_9}\Gamma_9
    \psi
    \right)
    \wedge
    e^{a_1} \wedge \cdots \wedge e^{a_9}\;.
    \end{aligned}
  $$
  \label{IIB-cocycles}
\end{defn}

\begin{prop}\label{JointCocycleforIIB}
  The elements in Def. \ref{IIB-cocycles} constitute an $L_\infty$-cocycle in IIA super-Minkowski
  spacetime with coefficients in the model for twisted $K^1$ from Def. \ref{RationalTwistedku}:
  $$
    \mathbb{R}^{9,1\vert \mathbf{16}+\mathbf{16}}
      \longrightarrow
    \mathfrak{l}( \Sigma \mathrm{KU}/BU(1) )\;.
  $$
\end{prop}
\begin{proof}
  By matching Clifford algebra conventions via remark \ref{IIBCliffordGenerators},
  one finds that this is the statement in \cite[section 2]{IIBAlgebra}.
  But we may also re-derive this as a consequence of Theorem \ref{TheTDualityIsoOnLInfinityCocycles} below, which
  says that, under the dimensional reduction isomorphism from Theorem \ref{ddReduction},
  the IIA cocycles of Prop. \ref{IIACocyclesCoeffsInKU} are sent to the IIB elements from Def. \ref{IIB-cocycles}
  by the $L_\infty$-isomorphism of Prop. \ref{CoefficientsIsoOnTDuality}. Since $L_\infty$-homomorphisms
  preserve cocycles, the above claim follows via Theorem \ref{ddReduction} and Theorem \ref{TheTDualityIsoOnLInfinityCocycles}
  from Proposition \ref{IIACocyclesCoeffsInKU}.
\end{proof}

In summary, Proposition \ref{IIACocyclesCoeffsInKU} and Proposition \ref{JointCocycleforIIB}
say that the homotopical descent of the cocycles for the WZW-terms of
the super $p$-pranes from extended super-Minowski spacetime down to the actual
super-Minkowski spacetime is of the following form:

\vspace{.1cm}

\begin{center}
\begin{tabular}{|c|c|c|}
  \hline
  String theory & Unified brane cocycles & Rational image of
  \\
  \hline
  type IIA & $ \xymatrix{\mathbb{R}^{9,1\vert \mathbf{16} + \overline{\mathbf{16}}}  \ar[rr]^{\mu^{\mathrm{IIA}}_{{}_{F1/Dp}}} && \mathfrak{l}(\mathrm{KU}/BU(1))}$ & twisted $K^0$-theory
  \\
  \hline
  type IIB & $\xymatrix{\mathbb{R}^{9,1\vert \mathbf{16}+ {\mathbf{16}}}
  \ar[rr]^{\mu^{\mathrm{IIB}}_{{}_{F1/Dp}}}
    && \mathfrak{l}(\Sigma \mathrm{KU}/BU(1))}$ & twisted $K^1$-theory
  \\
  \hline
\end{tabular}
\end{center}

\vspace{.1cm}

Notice that in this article, for ease of terminology we say ``brane charge'', for
the force field (``flux'') that the given brane feels, sourced by the (``magnetic'') background charge.
For D-branes these
are the \emph{RR-field strengths}, while in the literature it is usual to say ``D-brane charge''
for the (``electric'') charge carried by the D-brane itself. That the classification of the RR-fields
is in $K^0$/$K^1$ for type II A/B was first argued in \cite[p. 6]{MooreWitten00} for the untwisted case.
An explicit extension to the twisted case is indicated in
\cite{BouwknegtEvslinMathai04} \cite{MS}, which corresponds to the fields found above.

\section{Super $L_\infty$-algebraic T-duality}
\label{TDuality}

The goal of this section is to describe how T-duality appears as an isomorphism
between the F1/D$p$-brane $L_\infty$-cocycles on the type IIA and type IIB
supersymmetry super Lie algebra in 10d, after double dimensional
reduction, in the sense of Theorem \ref{ddReduction}, down to 9d.

\medskip
We have seen (Sec. \ref{SectionDoubleDimensionalReduction}) how the double
dimensional reduction crucially involves cyclification. We will apply this
for the corresponding spectra, i.e. rational twisted K-theory, to connect to
the cocycles that we have just encountered in Sec. \ref{SecBraneCocycles}.

\begin{prop}
    \label{CoefficientsIsoOnTDuality}
  \label{cyclofku/BU1}
  The cyclification $\mathfrak{L}\mathfrak{l}(\mathrm{KU}/BU(1))/\mathbb{R}$
  (Def. \ref{def:cyclification})
  of $\mathfrak{l}(\mathrm{KU}/BU(1))$
  (Def. \ref{RationalTwistedku})
  has CE-algebra
  $$
    \mathrm{CE}(\mathfrak{L}\mathfrak{l}(\mathrm{KU}/BU(1))/\mathbb{R})
    =
    \left\{
       \begin{array}{l}
         d c_2 = 0\,,\;\;\;\;d \tilde c_2 = 0
         \\
         d h_3 = - c_2 \wedge \tilde c_2
         \\
         d \omega_{2p+2} = h_3 \wedge \omega_{2p} + c_2 \wedge \omega_{2p+1}
         \\
         d \omega_{2p+1} = h_3 \wedge \omega_{2p-1} + \tilde c_2 \wedge \omega_{2p}
       \end{array}
    \right\}\;.
  $$
  The cyclification $\mathfrak{L}\mathfrak{l}(\Sigma\mathrm{KU}/BU(1))/\mathbb{R}$
  of $\mathfrak{l}(\Sigma\mathrm{KU}/BU(1))$ 
  has CE-algebra
  $$
    \mathrm{CE}(\mathfrak{L}\mathfrak{l}(\Sigma\mathrm{KU}/BU(1))/\mathbb{R})
    =
    \left\{
       \begin{array}{l}
         d c_2 = 0 \,, \;\;\;\; d \tilde c_2 = 0
         \\
         d h_3 = - c_2 \wedge \tilde c_2
         \\
         d \omega_{2p+2} = h_3 \wedge \omega_{2p} + \tilde c_2 \wedge \omega_{2p+1}
         \\
         d \omega_{2p+1} = h_3 \wedge \omega_{2p-1} + c_2 \wedge \omega_{2p}
       \end{array}
    \right\}\;.
  $$
  Hence there is an $L_\infty$-isomorphism of the form
  $$
    \phi_T
     \;:\;
    \xymatrix{
    \mathfrak{l}( \mathcal{L}(\mathrm{KU}/BU(1))/S^1 )
      \ar[rr]^-{\phi_T}_-\simeq
      &&
    \mathfrak{l}( \mathcal{L}(\Sigma\mathrm{KU}/BU(1))/S^1 )
    }
  $$
  relating the cyclifications of the rational twisted KU-coefficients,
  which is given by
  $$
         c_2 \longleftrightarrow \tilde c_2\;, \qquad
         h_3 \longmapsto h_3\;, \qquad
         \omega_{p} \longmapsto \omega_p
         \;.
  $$
\end{prop}
\begin{proof}
By Definition \ref{def:cyclification}, as a polynomial algebra the CE-algebra of $\mathfrak{L}\mathfrak{l}(\mathrm{KU}/BU(1))$ is obtained from the CE-algebra of $\mathfrak{l}(\mathrm{KU}/BU(1))$ by adding a shifted copy of each generator. We denote by $\omega_{2p-1}$ the shifted copy of $\omega_{2p}$ and by $-\tilde{c}_2$ the shifted copy of $h_3$. The differential is then defined by
\[
d\omega_{2p+2}=h_3 \wedge \omega_{2p}\;, \qquad d\omega_{2p+1}=h_3 \wedge \omega_{2p-1} + \tilde{c}_2\wedge \omega_{2p}\;, \qquad dh_3=0,\qquad d\tilde{c}_2=0\;.
\]
Next, again by Definition \ref{def:cyclification}, the CE-algebra of $\mathfrak{L}\mathfrak{l}(\mathrm{KU}/BU(1))/\mathbb{R}$
 is obtained by adding a further degree 2 generator $c_2$ and defining the differential as
\begin{align*}
\begin{array}{ll}
d\omega_{2p+2}= h_3 \wedge \omega_{2p} + c_2\wedge \omega_{2p+1}\;, & \qquad
~ d\omega_{2p+1}= h_3 \wedge \omega_{2p-1} + \tilde{c}_2\wedge \omega_{2p}\;,
\\
\quad ~~dc_2=0\;,  \qquad d\tilde{c}_2=0\;,  &
 \qquad \quad ~~dh_3=  - c_2\wedge \tilde{c}_2\;.
  \end{array}
\end{align*}
The proof for $\mathfrak{L}\mathfrak{l}(\Sigma\mathrm{KU}/BU(1))/\mathbb{R}$ is completely analogous.
\end{proof}

\begin{example}\label{ReductionOfIIACocyclesTo9d}
  By Prop. \ref{IIACocyclesCoeffsInKU},
  the combined F1/D$p$-brane cocycles on type IIA super-Minkowski spacetime
  constitute an
  $L_\infty$-homomorphism (corresponding to rational twisted $K^0$)
  of the form
  $$
    \mu_{{}_{F1/Dp}}^{\mathrm{IIA}}
      \;:\;
    \xymatrix{
      \mathbb{R}^{9,1\vert \mathbf{16} + \overline{\mathbf{16}}}
      \ar[rrr]^-{\left\{ \mu_{{}_{F_1}}^{\mathrm{IIA}}, ~\{ \mu_{{}_{D2p}} \} \right\} }
      &&&
      \mathfrak{l}(\mathrm{KU}/BU(1))
  }
  $$
  with the coefficient $L_\infty$-algebra on the right from Def. \ref{RationalTwistedku}.
  By Prop. \ref{IIBAsExtension} and Theorem \ref{ddReduction} this is naturally identified with a
  dimensionally reduced cocycle
  of the form
  $$
    \mathfrak{L}(\mu_{{}_{F1/Dp}}^{\mathrm{IIA}})/\mathbb{R}
     \;\;:\;\;
    \xymatrix{
       \mathbb{R}^{8,1\vert \mathbf{16}+ \mathbf{16}}
       \ar[drrr]_{c_2^{\mathrm{IIA}}}
       \ar[rrrrr]^-{
         \left\{
          {\tiny
            \begin{array}{ll}
              \mu_{{}_{F_1}}^{\mathrm{IIA}}, ~\left\{\mu_{{}_{D2p}}\right\}
              \\
              (\pi_9^{\mathrm{IIA}})_\ast \mu_{{}_{F_1}}^{\mathrm{IIA}}, ~
              \left\{(\pi_9^{\mathrm{IIA}})_\ast \mu_{{}_{D2p}}\right\}
              \\
              c_2^{\mathrm{IIA}}
            \end{array}
          }
         \right\}
       }
       &&&&&
       \mathfrak{L}\big(\mathfrak{l}(\mathrm{KU}/BU(1))\big)/\mathbb{R}
       \ar[dll]^-{\omega_2}
       \\
       &&&
       b \mathbb{R}
    }
  $$
  with coefficients now given by Prop. \ref{cyclofku/BU1},
  where in braces on top we see the original cocycles without the pieces containing $e^9$, below that we
  see the $e^9$-components, and in the last line the cocycle that classifies the IIA extension.
Similarly,
  the combined F1/D$p$-brane cocycles on type IIB super-Minkowski spacetime
  constitute an
  $L_\infty$-homomorphism (corresponding to rational twisted $K^1$)
  of the form
  $$
    \mu_{{}_{F1/Dp}}^{\mathrm{IIB}}
      \;:\;
    \xymatrix{
      \mathbb{R}^{9,1\vert \mathbf{16} + {\mathbf{16}}}
      \ar[rrr]^-{\left\{ \mu_{{}_{F_1}}^{\mathrm{IIB}}, ~\{\mu_{{}_{D(2p+1)}}\} \right\} }
      &&&
      \mathfrak{l}(\mathrm{KU}/BU(1))
  }
  $$
  and by Prop. \ref{IIBAsExtension} and Theorem \ref{ddReduction} this is naturally identified with a
  dimensionally reduced cocycle
  of the form
  $$
    \mathfrak{L}(\mu_{{}_{F1/Dp}}^{\mathrm{IIB}})/\mathbb{R}
     \;\;:\;\;
    \xymatrix{
       \mathbb{R}^{8,1\vert \mathbf{16}+ \mathbf{16}}
       \ar[drrr]_{c_2^{\mathrm{IIB}}}
       \ar[rrrrr]^-{\tiny
         \left\{
          {
            \begin{array}{ll}
              \mu_{{}_{F_1}}^{\mathrm{IIB}}, ~\{\mu_{{}_{D(2p+1)}}\}
              \\
              (\pi_9^{\mathrm{IIB}})_\ast \mu_{{}_{F_1}}^{\mathrm{IIB}}, ~
              \left\{(\pi_9^{\mathrm{IIB}})_\ast \mu_{{}_{D(2p+1)}}\right\}
              \\
              c_2^{\mathrm{IIB}}
            \end{array}
          }
         \right\}
       }
       &&&&&
       \mathfrak{L}\big(\mathfrak{l}(\Sigma\mathrm{KU}/BU(1))\big)/\mathbb{R}
       \ar[dll]^-{\omega_2}
       \\
       &&&
       b \mathbb{R}
    }
    \,.
  $$
\end{example}

Now for the original F1/D$p$ cocycles
$\mu_{{}_{F_1/Dp}}^{\mathrm{IIA}}$ and $\mu_{{}_{F_1/Dp}}^{\mathrm{IIB}}$
from \cite[Sec. 4]{FSS16} it does not make sense to ask whether they are
equivalent, since their domains are not. However, their double dimensional reductions
in Example \ref{ReductionOfIIACocyclesTo9d},
to which by Theorem \ref{ddReduction} they bijectively correspond,
do have the same domain
$\mathbb{R}^{8,1\vert \mathbf{16}+ \mathbf{16}}$. This is ultimately due to the fact that the two inequivalent
chiral $\mathrm{Spin}(9,1)$-representations become isomorphic as $\mathrm{Spin}(8,1)$-representations
-- see Remark \ref{SpinRepsInDimensions11And10And9}.
Hence for these it makes sense to ask whether they are equivalent $L_\infty$-homomorphisms.
We now establish that indeed they are:

\begin{theorem}
  \label{TheTDualityIsoOnLInfinityCocycles}
  The $L_\infty$-algebra isomorphism of Proposition \ref{CoefficientsIsoOnTDuality}
  takes the dimensionally reduced type IIA F1/D$p$-cocycle $\mathfrak{L}(\mu_{{}_{F1/Dp}}^{\mathrm{IIA}})/\mathbb{R}$
  of Example \ref{ReductionOfIIACocyclesTo9d}
  to the dimensionally reduced IIB cocycle $\mathfrak{L}(\mu_{{}_{F1/Dp}}^{\mathrm{IIB}})/\mathbb{R}$, making the following diagram
  of super $L_\infty$-algebras commute:
$$
\hspace{-3cm}
  \xymatrix@C=9em{
    & & b \mathbb{R}
    \\
    & && \mathfrak{L}\mathfrak{l}(\Sigma\mathrm{KU}/ BU(1))/\mathbb{R}
    \ar[ul]_-{\omega_2}
    \\
    &
    \mathbb{R}^{8,1\vert \mathbf{16} + \mathbf{16}}
    \ar[drr]|{~
          \mathfrak{L}(\mu_{{}_{F1/Dp}}^{\mathrm{IIA}})/\mathbb{R}
    }
    \ar[urr]|{
          \mathfrak{L}(\mu_{{}_{F1/Dp}}^{\mathrm{IIB}})/\mathbb{R}~
    }
    \ar[ddr]|{c_2^{\mathrm{IIA}}}
    \ar[uur]|{ c_2^{\mathrm{IIB}} }
    &&
    \\
    &
    && ~ \mathfrak{L}\mathfrak{l}(\mathrm{KU}/B U(1))/\mathbb{R}\;.
    \ar@{->}[uu]^\simeq_{\phi_T}
    \ar[dl]^-{\omega_2}
    \\
    &
    & b \mathbb{R}
  }
$$
\end{theorem}
\begin{proof}
  First of all we need to check that
  $$
    \omega_2\big(\phi_T( \mathfrak{L}(\mu_{{}_{F1/Dp}}^{\mathrm{IIA}})/\mathbb{R} )\big) = c_2^{\mathrm{IIB}}
    \,.
  $$
  By Theorem \ref{ddReduction} and Prop. \ref{CoefficientsIsoOnTDuality} the left hand side here is
  the integration over the fiber $-(\pi_9^{\mathrm{IIA}})_\ast(\mu_{{}_{F_1}}^{\mathrm{IIA}})$. Hence we need to check that
  \begin{equation}\label{InfinitesimalTopologicalTDuality}
    -(\pi_9^{\mathrm{IIA}})_\ast(\mu_{{}_{F_1}}^{\mathrm{IIA}})
     =
     c_2^{\mathrm{IIB}}
     \,.
  \end{equation}
  In components this equality says that
  $
    i \, \overline{\psi} \Gamma_9 \Gamma_{10}\psi
      =
    \overline{\psi}\Gamma^{\mathrm{IIB}}_{9} \psi
   $.
  This holds by direct inspection -- see Remark \ref{SpinorToVectorPairingIIB}. Similarly, one has
\begin{equation}\label{InfinitesimalTopologicalTDuality2}
    -(\pi_9^{\mathrm{IIB}})_\ast(\mu_{{}_{F_1}}^{\mathrm{IIB}})
     =
     c_2^{\mathrm{IIA}}
     \,,
  \end{equation}
due to the identity $i \overline{\psi}\Gamma_9^{\mathrm{IIB}} \Gamma_{10}\psi=\overline{\psi}\Gamma_9\psi$.
Moreover, we also have
\[
 \mu_{F1}^{\mathrm{IIA}}\vert_{8+1}= \mu_{F1}^{\mathrm{IIB}}\vert_{8+1}
\]
due to the identity $\Gamma_a^{\mathrm{IIB}}=\Gamma_a$ for $0\leq a\leq 8$.
  In view of this,
  to see that $\mu_{{}_{F1}}^{\mathrm{IIA}}$ from Def. \ref{IIACocycles}
  is sent to $\mu_{{}_{F1}}^{\mathrm{IIB}}$ from Def. \ref{IIBBraneCEElements},
  use that $\phi_T$ swaps $c_2^{\mathrm{IIA}}$ with $c_2^{\mathrm{IIB}}$, while
  keeping the restriction $\mu_{{}_{F1}}|_{8+1}$ intact.
 Now we need to check that the D-brane charges are sent to each other.
  Unwinding the definitions, this means that interchanging the components of
  fiber integration $(\pi_9^{\mathrm{IIA}})_\ast(-)$ and restriction $(-)\vert_{8+1}$ then the process in
  expression \eqref{fiberintegration} turns the IIA-brane  elements from Def. \ref{IIACocycles} to the IIB-brane elements from Def. \ref{IIBBraneCEElements}. This is indeed the case, as the following explicit computations show
  (where under the braces we keep using remark \ref{SpinorToVectorPairingIIB}):
    $$
    \hspace{-3.7cm}
  \begin{aligned}
\noindent \fbox{D1} \quad    (\pi_9^{\mathrm{IIA}})_\ast(\mu_{{}_{D2}}) - e^9 \wedge ( \mu_{{}_{D 0}}|_{8+1} )
    & =
    i \sum_{a = 0}^8 \overline{\psi} \Gamma_{a} \Gamma_9 \psi \wedge e^{a}
    -
    \overline{\psi} \underset{= - i \Gamma_9^{\mathrm{IIB}} \Gamma_9}{\underbrace{\Gamma_{10}}} \psi \wedge e^{9}
    \\
    & =
    i \left(\overline{\psi} \Gamma_a^{\mathrm{IIB}} \Gamma_9 \psi\right) \wedge e^{a}
    \\
    & = \mu_{{}_{D1}}\;.
   \end{aligned}
  $$
  $$
  \hspace{-7mm}
    \begin{aligned}
 \fbox{D3}~~   & (\pi_9^{\mathrm{IIA}})_\ast (\mu_{{}_{D4}}) - e^9 \wedge ( \mu_{{}_{D2}}|_{8+1} )=
    \\
    & = \tfrac{1}{3!} \sum_{a_i = 0}^8 \overline{\psi} \Gamma_{a_1 a_2 a_3} \Gamma_9 \Gamma_{10} \psi
          \wedge e^{a_1} \wedge e^{a_2} \wedge e^{a_3}
    - \tfrac{i}{2} \sum_{a_i = 0}^8 \overline{\psi} \Gamma_{a_1 a_2} \psi \wedge e^{a_1} \wedge e^{a_2} \wedge e^{9}
    \\
    & = \tfrac{1}{3!} \sum_{a_i = 0}^8 \overline{\psi} \Gamma_{a_1 a_2 a_3} \Gamma_9 \Gamma_{10} \psi
          \wedge e^{a_1} \wedge e^{a_2} \wedge e^{a_3}
    + \tfrac{1}{3!} 3 \sum_{a_i = 0}^8
      \overline{\psi} \Gamma_{a_1 a_2} \underset{= -i}{\underbrace{\Gamma_9^{\mathrm{IIB}} (\Gamma_9 \Gamma_{10})}} \psi \wedge e^{a_1} \wedge e^{a_2} \wedge e^{9}
    \\
    & =
    \tfrac{1}{3!}
    \left(
      \overline{\psi}
        \Gamma^{\mathrm{IIB}}_{a_1 a_2 a_3} (\Gamma_9 \Gamma_{10})
      \psi
    \right)
    \wedge
    e^{a_1} \wedge e^{a_2} \wedge e^{a_3}
    \\
    & = \mu_{{}_{D3}}\;.
  \end{aligned}
 $$
 $$
   \begin{aligned}
   \fbox{D5}~~
       & (\pi_9)_\ast( \mu_{{}_{D6}} )
         -
       e^9 \wedge ( \mu_{{}_{D4}}|_{8+1} )=
     \\
     &=
     \tfrac{i}{5!}
       \sum_{a_i = 0}^8
       \left(
       \overline{\psi}
         \Gamma_{a_1 \cdots a_5} \Gamma_{9}
       \psi
       \right)
       \wedge e^{a_1}\wedge \cdots \wedge e^{a_5}
       -
       \tfrac{1}{4!}
         \sum_{a_i = 0}^8
         \left(
         \overline{\psi}
           \Gamma_{a_1 \cdots a_4} \Gamma_{10}
         \psi
         \right)
         \wedge e^{a_1}\wedge \cdots\wedge e^{a_4}
         \wedge e^{9}
     \\
     &=
     \tfrac{i}{5!}
       \sum_{a_i = 0}^8
       \left(
       \overline{\psi}
         \Gamma_{a_1 \cdots a_5} \Gamma_{9}
       \psi
       \right)
       \wedge
       e^{a_1}\wedge \cdots \wedge e^{a_5}
       -
       \tfrac{1}{5!}
         5
         \sum_{a_i = 0}^8
         \Big(
         \overline{\psi}
           \Gamma_{a_1 \cdots a_4}
           \underset{= -i \Gamma_9^{\mathrm{IIB}} \Gamma_9}{\underbrace{\Gamma_{10}}}
         \psi
         \Big)
         \wedge e^{a_1}\wedge \cdots\wedge e^{a_4}
         \wedge e^{9}
     \\
     & =
     \tfrac{i}{5!}
       \left(
       \overline{\psi}
         \Gamma^{\mathrm{IIB}}_{a_1 \cdots a_5} \Gamma_9
       \psi
      \right)
      \wedge
      e^{a_1}\wedge \cdots \wedge e^{a_5}
      \\
     & = \mu_{{}_{D5}}\;.
   \end{aligned}
 $$
 $$
   \begin{aligned}
    \fbox{D7}~~
      & (\pi_9)_\ast( \mu_{{}_{D8}} )
       -
       e^9 \wedge ( \mu_{{}_{D6}}|_{8+1} )=
      \\
      & =
      \tfrac{1}{7!}
        \sum_{a_i = 0}^8
        \left(
        \overline{\psi}
          \Gamma_{a_1 \cdots a_7} \Gamma_9 \Gamma_{10}
        \psi
        \right)
        \wedge
        e^{a_1}\wedge \cdots \wedge e^{a_7}
        -
        \tfrac{i}{6!}
        \sum_{a_i = 0}^8
        \left(
        \overline{\psi}
          \Gamma_{a_1 \cdots a_6}
        \psi
        \right)
        \wedge
        e^{a_1} \wedge
          \cdots
        \wedge e^{a_6} \wedge e^9
        \\
        & =
      \tfrac{1}{7!}
        \sum_{a_i = 0}^8
        \left(
        \overline{\psi}
          \Gamma_{a_1 \cdots a_7} \Gamma_9 \Gamma_{10}
        \psi
        \right)
        \wedge
        e^{a_1}\wedge \cdots \wedge e^{a_7}
        +
        \tfrac{1}{7!}
        7
        \sum_{a_i = 0}^8
        \Big(
        \overline{\psi}
          \Gamma_{a_1 \cdots a_6} \underset{= -i}{\underbrace{\Gamma_9^{\mathrm{IIB}} (\Gamma_9 \Gamma_{10})}}
        \psi
        \Big)
        \wedge
        e^{a_1} \wedge
          \cdots
        \wedge e^{a_6} \wedge e^9
      \\
      & =
      \tfrac{1}{7!}
        \left(
       \overline{\psi}
         \Gamma_{a_1 \cdots a_7} (\Gamma_9 \Gamma_{10})
       \psi
       \right)
       \wedge
       e^{a_1}
       \wedge
         \cdots
       \wedge
       e^{a_7}
      \\
    &= \mu_{{}_{D7}}\;.
   \end{aligned}
 $$
 $$
  \begin{aligned}
   \fbox{D9}~~
   & (\pi_9)_\ast( \mu_{{}_{D10}} )
     -
   e^9 \wedge ( \mu_{{}_{D8}}|_{8+1} )=
   \\
   & =
   \tfrac{i}{9!}
   \sum_{a_i = 0}^8
   \big(
   \overline{\psi}
     \Gamma_{a_1\cdots a_{9}} \Gamma_9
   \psi
   \big)
   \wedge
   e^{a_1}\wedge \cdots \wedge e^{a_9}
   -
   \tfrac{1}{8!}
   \Big(
   \overline{\psi}
     \Gamma_{a_1 \cdots a_8} \underset{= -i \Gamma_9^{\mathrm{IIB}} \Gamma_9}{\underbrace{\Gamma_{10}}}
   \psi
   \Big)
   \wedge
   e^{a_1} \wedge \cdots \wedge e^{a_{8}}
   \wedge e^9
   \\
   & =
   \tfrac{i}{9!}
   \sum_{a_i = 0}^8
   \left(
   \overline{\psi}
     \Gamma_{a_1\cdots a_{9}} \Gamma_9
   \psi
   \right)
   \wedge
   e^{a_1}\wedge \cdots \wedge e^{a_9}
   +
   \tfrac{i}{9!}
   9
   \sum_{a_i = 0}^8
   \left(
   \overline{\psi}
     \Gamma_{a_1 \cdots a_8} \Gamma_9^{\mathrm{IIB}} \Gamma_9
   \psi
   \right)
   \wedge
   e^{a_1} \wedge \cdots \wedge e^{a_{8}}
   \wedge e^9
   \\
   & =
   \tfrac{i}{9!}
    \left(
    \overline{\psi}
      \Gamma_{a_1\cdots a_9}\Gamma_9
    \psi
    \right)
    \wedge
    e^{a_1} \wedge \cdots \wedge e^{a_9}
    \\
    & = \mu_{{}_{D9}}\;.
  \end{aligned}
 $$

 \vspace{-5mm}
\end{proof}

\begin{remark}[Topological T-duality I]
\label{NatureOfTheAlgebraicTDualityIso}
{\bf (i)} The interpretation of our super-Minkowski spacetimes as
tangent spaces of spacetime manifolds $X$, and extensions of them as corresponding
fiber bundles $X_{10} \to X_9$ of spacetime manifolds, means that the diagram of super $L_\infty$ algebras
from Proposition \ref{IIBAsExtension}
$$
  \xymatrix{
    \mathbb{R}
     \ar[r]
    &
    \mathbb{R}^{9,1\vert \mathbf{16}+ \overline{\mathbf{16}}}
    \ar[dr]_{\hspace{-5mm}\mathrm{hofib}(c_2^{\mathrm{IIA}})}
    &&
    \mathbb{R}^{9,1\vert \mathbf{16} + \mathbf{16}}
    \ar[dl]^{~~\mathrm{hofib}(c_2^{\mathrm{IIB}})}
    &
    \mathbb{R}
    \ar[l]
    \\
    && \mathbb{R}^{8,1\vert \mathbf{16}+ \mathbf{16}}
    \ar[dl]|{c_2^{\mathrm{IIB}}}
    \ar[dr]|{c_2^{\mathrm{IIA}}}
    \\
    & b\mathbb{R} && b\mathbb{R}\;,
  }
$$
together with 3-cocycles $\mu_{{}_{F1}}^{\mathrm{IIA}/\mathrm{B}}$
globalizes to a diagram of manifolds of the form
$$
  \xymatrix{
    S^1_{{}_A}  \ar[r] & X_{10}^{{}^{\mathrm{IIA}}}
    \ar[dr]_{\pi_9^{{}^{\mathrm{IIA}}}}
    &&
    X_{10}^{{}^{\mathrm{IIB}}}
    \ar[dl]^{\pi_9^{{}^{\mathrm{IIB}}}}
    &
    S^1_{{}_B} \ar[l]
    \\
    && X_9
  }
$$
which carry closed differential super 3-forms $H_3^{A/B} \in \Omega^3_{\mathrm{cl}}( X_{10}^{A/B} )$.

Indeed, the consistency of the Green-Schwarz sigma-model
for the type II superstring on $X_{10}^{A/B}$
requires that the bispinorial component of the super-3-form $H$ here
is constrained to coincide on each tangent spacetime with  our
cocycle $\mu_{F_1}^{\mathrm{IIA/B}}$ (this follows from \cite[equation (2.11)]{BST86}).
In particular, if $H$ happens to have vanishing bosonic component, then
it is entirely fixed by restricting on each tangent space to our cocycle
$\mu_{{}_{F_1}}^{\mathrm{IIA/B}}$ (this is amplified in \cite[equation (2.15)]{BST86}).
Similarly, these circle bundles have first Chern classes whose representing forms
$C_2^{A/B}$ need to have bispinorial super-components that super-tangent-space wise
coincide with the cocycles $c_2^{{}^{\mathrm{IIA}/\mathrm{B}}}$.

\item {\bf (ii)} Hence the globalization of the super-tangent-space wise equivalence
that we see in equations (\ref{InfinitesimalTopologicalTDuality})
and (\ref{InfinitesimalTopologicalTDuality2})
in the proof of Theorem \ref{TheTDualityIsoOnLInfinityCocycles} imposes the global condition
$$
  C_2^{{}^{\mathrm{IIA/B}}} = -(\pi^{{}^{\mathrm{IIB/A}}}_9)_\ast \big(H_3^{{}^{\mathrm{IIB/A}}}\big)
  \,.
$$
This relation is what is used as an axiom for
``topological T-duality'' in \cite[(1.8)]{BouwknegtEvslinMathai04}, see \cite[lemma 2.12, lemma 2.33]{BunkeSchick05}.

\end{remark}

\section{T-Correspondence space and Doubled spacetimes}
\label{TFolds}

Above we considered T-duality as an equivalence of classifying maps (moduli) of
fields. Here we consider the incarnation of this equivalence in terms of the
higher extended super-Minkowski spaces that are classified thereby.

\medskip
With every $L_\infty$-cocycle on some superalgebra $\mathfrak{g}$, we get some $L_\infty$-extension
that it classifies, namely its homotopy fiber $\widehat{\mathfrak{g}} \longrightarrow \mathfrak{g}$,
accoding to Example \ref{homotopyfiberofLinfinityCocycles}.
If $\mathfrak{g}$ is some super-Minkowski super Lie algebra (Sec. \ref{Spinors}), then this extension
is the higher extended super-Minkowski spacetime which may be thought of as containing condensates
of those brane species that the cocycle classifies \cite[Remark 3.11]{FSS13}. In particular
 a 2-cocycle corresponds to a 0-brane and the corresponding extension is just an ordinary central
 extension, hence grows one extra dimension of spacetime, as befits a 0-brane condensate
 \cite[Remark 4.6]{FSS13}.

\medskip
Here we analyze this phenomenon for the cocycles that classify the type II branes
on $9d$ $N=2$ super-Minkowski spacetime, according to Example \ref{ReductionOfIIACocyclesTo9d}.
By Theorem \ref{TheTDualityIsoOnLInfinityCocycles} there is only one such
cocycle, up to equivalence, since the
dimensional reduction $\mathfrak{L}(\mu_{{}_{F1/Dp}}^{\mathrm{IIA}})/\mathbb{R}$
of the IIA branes and the dimensional reduction $\mathfrak{L}(\mu_{{}_{F1/Dp}}^{\mathrm{IIB}})/\mathbb{R}$
of the IIB branes agrees in 9d, by $L_\infty$-algebraic T-duality.
Hence we just write $\mathfrak{L}(\mu_{{}_{F1/Dp}}^{\mathrm{II}})/\mathbb{R}$ for either of them.
Now, as made explicit in the diagram in Theorem \ref{TheTDualityIsoOnLInfinityCocycles}, this
cocycle contains contributions from \emph{two} 0-brane species. One of these is the
D0-brane of type IIA in 10d, descended down to 9d, and the other is the double dimensional
reduction of the type IIB string 3-cocycle from 10d to 9d. Hence the condensation of these \cite[Remark 3.11, 4.6]{FSS13} grows
\emph{two} extra spacetime dimensions. By Theorem \ref{TheTDualityIsoOnLInfinityCocycles}, these
are the infinitesimal version of what in finite T-duality are the two circle
fibers $S^1_{{}_A}$ and $S^1_{{}_B}$ of Remark \ref{NatureOfTheAlgebraicTDualityIso}. Hence in the notation
of that remark, we obtain the fiber product spacetime
$$
  \xymatrix{
    & S^1_{{}_A} \times S^1_{{}_B} \ar@{^{(}->}[r] & X_{10}^{{}^{\mathrm{IIA}}}\times_{X_9} X_{10}^{{}^{\mathrm{IIB}}}
    \ar@{}[dd]|{\mbox{\tiny (pb)}}
    \ar[dl]
    \ar[dr]
    \\
    S^1_{{}_A}  \ar[r] & X_{10}^{{}^{\mathrm{IIA}}}
    \ar[dr]_{\pi_9^{{}^{\mathrm{IIA}}}}
    &&
    X_{10}^{{}^{\mathrm{IIB}}}
    \ar[dl]^{\pi_9^{{}^{\mathrm{IIB}}}}
    &
    S^1_{{}_B}\;, \ar[l]
    \\
    && X_9
  }
$$
which is an $S^1_{{}_A} \times S^1_{{}_B}$-fiber bundle over $X_9$.

\begin{defn} \label{doubledtypeIIspacetime}
  Write
  $$
     \mathbb{R}^{8+(1+1),1\vert 32}
      :=
         \mathbb{R}^{9,1\vert \mathbf{16} +\overline{\mathbf{16}}}
       \underset{\mathbb{R}^{8,1\vert \mathbf{16}+ \mathbf{16}}}{\times}
     \mathbb{R}^{9,1\vert \mathbf{16} + \mathbf{16}}
  $$
   as shorthand for the fiber product of the type IIA super-Minkowski spacetime
  with its IIB version, over their common 9d base, according to Prop. \ref{IIBAsExtension},
hence for the super Lie algebra fitting into the following
  fiber product diagram
  $$
     \xymatrix{
       & \mathbb{R}^{8+(1+1),1\vert 32}
       \ar@{}[dd]|{\mbox{\tiny (pb)}}
       \ar[dl]_{p_{{}_B}}
        \ar[dr]^{p_{{}_A}}
       \\
       \mathbb{R}^{9,1\vert \mathbf{16}+\mathbf{16}}
      \ar[dr]
      &&
       \mathbb{R}^{9,1\vert \mathbf{16} + \overline{\mathbf{16}}}\;.
      \ar[dl]
      \\
     & \mathbb{R}^{8,1\vert \mathbf{16}+ \mathbf{16}}
    }
  $$
\end{defn}

\begin{prop}\label{BunkeStyleTDualityInfinitesimally}
 We have a diagram of super $L_\infty$-algebras of the following form
 $$
   \xymatrix@=1.9em{
     &
     p_{{}_A}^\ast
     \widehat{
       \mathbb{R}^{9,1\vert \mathbf{16} +\overline{\mathbf{16}}}
     }
     \ar@{<--}[rr]_\simeq^{\nu}
     \ar[dl]_{\pi_9^{\mathrm{IIB}}}
     \ar[dr]
     \ar@{}[dd]|{\mbox{\tiny (pb)}}
     &&
     p_{{}_B}^\ast
     \widehat{
       \mathbb{R}^{9,1\vert \mathbf{16} +{\mathbf{16}}}
     }
     \ar[dl]
     \ar[dr]^{\pi_9^{\mathrm{IIA}}}
     \ar@{}[dd]|{\mbox{\tiny (pb)}}
          \\
     \widehat{\mathbb{R}^{9,1\vert \mathbf{16} +\overline{\mathbf{16}}}}
     \ar[dr]|{~~~~~\mathrm{hofib}(\mu_{{}_{F1}}^{\mathrm{IIA}})}
     &&
    \mathbb{R}^{8+(1+1),1\vert 32}
     \ar[dl]|{p_{{}_A}}
     \ar[dr]|{~p_{{}_B}}
     \ar@{}[dd]|{\mbox{\tiny (pb)}}
     &&
     \widehat{
       \mathbb{R}^{9,1\vert \mathbf{16} +{\mathbf{16}}}
     }\;,
     \ar[dl]|{\mathrm{hofib}(\mu_{{}_{F1}}^{\mathrm{IIB}}) }
     \\
     &
     \mathbb{R}^{9,1\vert \mathbf{16} +\overline{\mathbf{16}}}
     \ar[dr]|{ \pi_9^{\mathrm{IIA}}}
       &&
     \mathbb{R}^{9,1\vert \mathbf{16} + \mathbf{16}}
     \ar[dl]|{ \pi_9^{\mathrm{IIB}} }
     \\
     & & \mathbb{R}^{8,1\vert \mathbf{16}+ \mathbf{16}}
   }
 $$
where
 \begin{itemize}
    \item
        $\widehat{\mathbb{R}^{9,1\vert \mathbf{16}+ \overline{\mathbf{16}}}}
       \simeq  \mathfrak{string}_{\mathrm{IIA}}$
    is the super Lie 2-algebra extension of type IIA super-Minkowski spacetime by the
    3-cocycle for the type IIA superstring; i.e. the infinitesimal model of the
(super-)gerbe underlying the type IIA B-field;
  \item
        $\widehat{\mathbb{R}^{9,1\vert \mathbf{16}+ {\mathbf{16}}}}
       \simeq \mathfrak{string}_{\mathrm{IIB}}$
    is the super Lie 2-algebra extension of type IIB super-Minkowski spacetime by the
    3-cocycle for the type IIB superstring; i.e. the infinitesimal model of the
(super-)gerbe underlying the type IIB B-field;
 \item
   by slight abuse of notation, $\pi_9^{\mathrm{IIA}/\mathrm{IIB}}$ always denotes the map that projects out the 9th dimension
   of type IIA/IIB, respectively, hence dually, on underlying graded algebras, the canonical map $(-) \to (-)[e_9^{\mathrm{IIA}/\mathrm{IIB}}]$
   induced from adjoining the vielbein generator $e_9^{\mathrm{IIA}/\mathrm{IIB}}$
 \end{itemize}
 such that
 \begin{enumerate}
  \item a  horizonal isomorphism $\nu$ exists as shown, between the pullback of these two
extensions to the correspondence space, making the diagram commute;
 \item
   on the classifying 3-cocycles it is given by the \emph{Poincar{\'e} form}
   $$
     \mathcal{P}
       :=
     e_9^{\mathrm{IIA}} \wedge e_9^{\mathrm{IIB}}
   $$
     as
   $$
     p_B^\ast (\mu_{{}_{F1}}^{\mathrm{IIB}} )- p_A^\ast (\mu_{{}_{F1}}^{\mathrm{IIA}})
      =
     d  \mathcal{P}
     \,.
   $$
\end{enumerate}
\end{prop}
\begin{proof}
  By Example \ref{homotopyfiberofLinfinityCocycles}, the extended super-Minkowski
  super Lie 2-algebra on the far left is given by
  $$
    \mathrm{CE}\big(\widehat{\mathbb{R}^{9,1\vert \mathbf{16} +\overline{\mathbf{16}}}}\big)
    =
    \mathrm{CE}(\mathbb{R}^{9,1\vert \mathbf{16} +\overline{\mathbf{16}}})[f^{\mathrm{IIA}}_2]/(d f^{\mathrm{IIA}}_2 = \mu_{{}_{F1}}^{\mathrm{IIA}})
  $$
  and that on the far right by
  $$
    \mathrm{CE}\big(\widehat{\mathbb{R}^{9,1\vert \mathbf{16} +{\mathbf{16}}}}\big)
    =
    \mathrm{CE}(\mathbb{R}^{9,1\vert \mathbf{16} + {\mathbf{16}}})[f^{\mathrm{IIB}}_2]/(d f^{\mathrm{IIB}}_2 = \mu_{{}_{F1}}^{\mathrm{IIB}})
    \,.
  $$
  Their pullbacks along the projections $p_{{}_A}$ and $p_{{}_B}$ are directly
  seen to be given by further adjoining the generator $e_9^{\mathrm{IIB}}$ or $e_9^{\mathrm{IIA}}$, respectively.

  Then define $\nu$ by defining its dual $\nu^\ast$ by sending all generators in
  $p_A^\ast \mathrm{CE}(\widehat{ \mathbb{R}^{9,1\vert \mathbf{16} + \overline{\mathbf{16}}} })$
  to the generator of the same name in $p_B^\ast \mathrm{CE}(\widehat{ \mathbb{R}^{9,1\vert \mathbf{16} + {\mathbf{16}}} })$
  except for $f^{\mathrm{IIA}}_2$, for which we set
    \begin{equation}
      \label{shiftbyPoincare}
     \nu^\ast \;:\; f^{\mathrm{IIA}}_2 \mapsto f_2^{\mathrm{IIB}} - \mathcal{P}
     \,.
    \end{equation}
  To see that this indeed gives a homomorphism, recall equation (\ref{InfinitesimalTopologicalTDuality})
  which followed explicitly in components via remark \ref{SpinorToVectorPairingIIB}
  $$
    i \, \overline{\psi} \Gamma^{\mathrm{IIA}/\mathrm{IIB}}_9 \Gamma_{10}\psi
      =
    \overline{\psi}\Gamma^{\mathrm{IIB}/\mathrm{IIA}}_{9} \psi
    \,.
   $$
   Here this means that on the doubled correspondence space the type IIA and IIB superstring cocycles $\mu_{F1}^{A/B}$
   (Definition \ref{IIACocycles} and Definition \ref{IIBBraneCEElements})
   may be expressed in terms of each other as follows (as before in the proof of Theorem \ref{TheTDualityIsoOnLInfinityCocycles})
   \begin{equation}\label{F1ABOnCorrespondenceSpace}
     \mu_{F_1}^{A/B}
       =
       \mu_{F1}^9 + e_9^{A/B} \wedge ( \overline{\psi} \wedge \Gamma_9^{B/A} \psi )\,,
   \end{equation}
   where $\mu_{F1}^9 = \mu_{F1}^{\mathrm{IIA}}|_{8+1} = \mu_{F1}^{\mathrm{IIB}}\vert_{8+1}$ denotes their common summand, the one which involves only the generators $e_a$ for $a \leq 8$.
   This implies first of all the claimed coboundary between the 3-cocycles:
  $$
    \begin{aligned}
      \mu_{{}_{F_1}}^{\mathrm{IIA}} + d \mathcal{P}
      &=
      \underset{\mu_{{}_{F_1}}^{\mathrm{IIA}}}{\underbrace{\mu_{{}_{F1}}^9 + e_9^{\mathrm{IIA}} \wedge \overline{\psi}\Gamma_9^{\mathrm{IIB}} \psi}}
      +
      \underset{d(e_9^{\mathrm{IIA}} \wedge e_9^{\mathrm{IIB}})}{
      \underbrace{
        e_9^{\mathrm{IIB}} \wedge \overline{\psi}\Gamma_9^{\mathrm{IIA}} \psi
        - e_9^{\mathrm{IIA}} \wedge \overline{\psi}\Gamma_9^{\mathrm{IIB}} \psi
      }}
      \\
      &=
      \mu_{{}_{F1}}^9 + e_9^{\mathrm{IIB}} \wedge \overline{\psi}\Gamma_9^{\mathrm{IIA}} \psi
      \\
      &= \mu_{{}_{F1}}^{\mathrm{IIB}}\;.
    \end{aligned}
  $$
  In turn, this implies now that $\nu^\ast$ is indeed a dg-algebra homomorphism, because in the extended type IIA
  dg-algebra
  $p_B^\ast \mathrm{CE}\big(\widehat{\mathbb{R}^{9,1\vert \mathbf{16} + {\mathbf{16}}}}\big)$ we have
  the following relations:
  $$
    \begin{aligned}
      d(\nu^\ast(f^{\mathrm{IIA}}_2))
       & =
      d( f^{\mathrm{IIB}}_2 - \mathcal{P} )
      \\
       & =
      \mu_{F1}^{\mathrm{IIB}} - d \mathcal{P}
      \\
      & = \mu_{F1}^{\mathrm{IIA}}
      \\
      & = \nu^\ast(d(f^{\mathrm{IIA}}_2))
      \,.
    \end{aligned}
  $$
  Finally, $\nu$ clearly has an inverse $\nu^{-1}$. Its linear dual is the identity on all generators
  except $f^{\mathrm{IIB}}_2$, where it is
    $$
     (\nu^{-1})^\ast \;:\; f^{\mathrm{IIB}}_2 \mapsto f^{\mathrm{IIA}}_2 + \mathcal{P}
     \,.
  $$
  \vspace{-5mm}
\end{proof}


\medskip
\begin{remark}[Topological T-duality II]\label{BunkeStyleTDuality}
Proposition \ref{BunkeStyleTDualityInfinitesimally} is evidently the infinitesimal and supergeometric picture of
topological T-duality as considered in \cite[def. 2.8]{BunkeRumpfSchick06}.
There, one considers two circle bundles $X_{10}^{{}^{\mathrm{IIA}}}$ and $X_{10}^{{}^{\mathrm{IIB}}}$,
carrying a $U(1)$-gerbe $\mathcal{G}_{{}_{\mathrm{IIA}}}$ and $\mathcal{G}_{{}_{\mathrm{IIB}}}$, respectively, such
 that there is an equivalence
between these gerbes after pullback to the fiber product:
$$
  \xymatrix{
    & p_{{}_A}^\ast \mathcal{G}_{{}_{\mathrm{IIA}}}
    \ar@{<-}[rr]^{\simeq}
    \ar[dr]|{~p_{{}_A}~}
    \ar[dl]
     \ar@{}[dd]|{\mbox{\tiny (pb)}}
    &&
    p_{{}_B}^\ast \mathcal{G}_{{}_{\mathrm{IIB}}}
    \ar[dl]|{~p_{{}_B}}
    \ar[dr]
     \ar@{}[dd]|{\mbox{\tiny (pb)}}
    \\
    \mathcal{G}_{{}_{\mathrm{IIA}}}
    \ar[dr]
    & & X_{10}^{{}^{\mathrm{IIA}}}\times_{X_9} X_{10}^{{}^{\mathrm{IIB}}}
    \ar@{}[dd]|{\mbox{\tiny (pb)}}
    \ar[dl]
    \ar[dr]
    &&
    \mathcal{G}_{{}_{\mathrm{IIB}}}\;.
    \ar[dl]
    \\
    & X_{10}^{{}^{\mathrm{IIA}}}
    \ar[dr]_{\pi_9^{\mathrm{IIA}}}
    &&
    X_{10}^{{}^{\mathrm{IIB}}}
    \ar[dl]^{\pi_9^{{}^{\mathrm{IIB}}}}
    &
    \\
    && X_9
  }
$$
and such that this equivalence is fiberwise exhibited by the 2-cocycle which is the cup-product of the
canonical 1-classes on the two circle fibers.
\end{remark}
Therefore, we find both perspectives on topological T-duality
(Remark \ref{NatureOfTheAlgebraicTDualityIso} and Remark \ref{BunkeStyleTDuality}),
from the analysis of the super-tangent-space wise super $L_\infty$-cocycles
(Prop. \ref{TheTDualityIsoOnLInfinityCocycles} and Prop. \ref{BunkeStyleTDualityInfinitesimally},
respectively). (Previously these rules had been guessed, not derived from string theory.)
Of course these two perspectives are supposed to be equivalent, rationally. In section \ref{TDualityLie2Algebra}
we see how this equivalence arises within the homotopy theory of super $L_\infty$-algebras.

\medskip
Before we turn to that, we observe that the correspondence super-spacetime
$\mathbb{R}^{9+(1,1),1\vert \mathbf{32}}$ also serves to clarify the operation of
T-duality on the RR-fields:
Consider the standard fact
(see also \cite[section 4]{FSS16}) that a degree-3 twisted $\mathrm{KU}$-valued cocycle
on any super $L_\infty$-algebra
$\mathfrak{g}$
$$
 \left\{
 \raisebox{20pt}{
  \xymatrix{
    \mathfrak{g}
    \ar@{-->}[rr]^-{\{ \omega_{2p} \vert p \in \mathbb{N} \}}
    \ar[dr]_{\mu_3}
    &&
    \mathfrak{l}(\mathrm{KU}/BU(1))
    \ar[dl]
    \\
    & b^2 \mathbb{R}
  }
  }
  \right\}
  \;\;\;\;\leftrightarrow\;\;\;\;
  \left\{
    \omega_{2p} \in \mathrm{CE}(\mathfrak{g})\vert
    d_{\mathfrak{g}} \omega_{2p} = \mu_3 \wedge \omega_{2p-2}
  \right\}
$$
  is naturally identified with an untwisted $\mathrm{ku}$-cocycle on the higher central extension $\widehat{\mathfrak{g}}$.
 The latter is classified by the twist $\mu_3$ according to Example \ref{homotopyfiberofLinfinityCocycles}
$$
  \left\{
  \raisebox{20pt}{
  \xymatrix{
    \widehat{\mathfrak{g}}
     \ar[d]|{\mathrm{hofib}(\mu_3)}
     \ar@{-->}[rr]^-{\left\{ \tilde \omega_{2p} \vert p \in \mathbb{N} \right\}}
    &&
   \mathfrak{l}(\mathrm{KU})
    \\
    \mathfrak{g}
  }
  }
  \right\}
  \;\;\;\;\leftrightarrow\;\;\;\;
  \left\{
    \tilde \omega_{2p} \in \mathrm{CE}(\widehat{\mathfrak{g}}) \vert  d \tilde \omega_{2p} = 0
  \right\}
$$
under the relation
$$
  \tilde \omega_{2p} = [\exp(-f_2) \wedge C]_{2p}
  \,,
  \;\;\;\;
  C := \underset{p \in \mathbb{N}}{\sum} \omega_{2p}
  \,,
$$
where $f_2$ is the new generator of $\mathrm{CE}(\widehat{\mathfrak{g}}) = \mathrm{CE}(\mathfrak{g})[f_2, df_2 = \mu_3] $ according to
Example \ref{homotopyfiberofLinfinityCocycles}, and where $[-]_{2p}$ denotes taking the summand of homogeneous degree $2p$.
Under this identification, the top three morphisms in the diagram of Proposition \ref{BunkeStyleTDualityInfinitesimally}
define a linear map
$$
  (\pi_9^{\mathrm{IIA}})_\ast \circ \nu^\ast \circ (\pi_9^{\mathrm{IIB}})^\ast
$$
from the twisted cocycle of the type IIA super-Minkowski spacetime to that of type IIB.
Here the two pullback operations are just the dual morphisms on CE-algebras, while the
pushforward operation $(\pi_9^{\mathrm{IIA}})_\ast$ is defined as in expressions (\ref{decomp})
and (\ref{fiberintegration}).
\begin{prop}
  \label{TDualityViaPullPush}
  This operation is a well defined morphism on twisted cohomology groups
  $$
    (\pi_9^{\mathrm{IIA}})_\ast \circ \nu^\ast \circ (\pi_9^{\mathrm{IIB}})^\ast
      \;:\;
    H_{\mu_{F1}^{\mathrm{IIA}}} \left( \mathbb{R}^{9,1\vert \mathbf{16} + \overline{\mathbf{16}}}, \mathfrak{l}(\mathrm{KU})  \right)
      \longrightarrow
    H_{\mu_{F1}^{\mathrm{IIB}}} \left( \mathbb{R}^{9,1\vert \mathbf{16} + {\mathbf{16}}}, \mathfrak{l}(\Sigma \mathrm{KU})   \right)
  $$
  and is in fact an isomorphism. Moreover, it identifies the type IIA D-brane cocycles with those
  of type IIB, as in Theorem \ref{TheTDualityIsoOnLInfinityCocycles}:
  \begin{equation}
   \label{HoriFormulaInProp}
    \exp(-f_2^{\mathrm{IIB}}) \wedge  C^{\mathrm{IIB}}
    =
    (\pi_9^{\mathrm{IIA}})_\ast \circ \nu^\ast \circ (\pi_9^{\mathrm{IIB}})^\ast
      \left(
        \exp(-f_2^{\mathrm{IIA}}) \wedge C^{\mathrm{IIA}}
      \right)
    \,,
  \end{equation}
  where
  $$
    C^{\mathrm{IIA}} := \mu_{{}_{D0}} + \mu_{{}_{D2}} + \cdots + \mu_{{}_{D 10}}
    \;\,,
    \;\;\;\;\;
    C^{\mathrm{IIB}} := \mu_{{}_{D1}} + \mu_{{}_{D3}} + \cdots \mu_{{}_{D 9}}
    \,.
  $$
\end{prop}
\begin{proof}
  First, the pullback operation $(\pi_9^{\mathrm{IIB}})^\ast$ leaves the form of the cocycles unaffected, and
  simply regards them inside the larger CE-algebra which has the generator $e_9^{\mathrm{IIB}}$ adjoined.
  Next, by Proposition \ref{BunkeStyleTDualityInfinitesimally}, equation (\ref{shiftbyPoincare}),
  the pullback along $\nu$ amounts to substituting $f^{\mathrm{IIA}}_2$ by
  $f_2^{\mathrm{IIB}} - \mathcal{P} = f^{\mathrm{IIB}}_2 - e_9^{\mathrm{IIA}} \wedge e_9^{\mathrm{IIB}}$.
  This means that
  \begin{equation}
   \label{nuastonexpf2C}
    \nu^\ast
      \;:\;
    \exp(-f^{\mathrm{IIA}}_2) \wedge C^{\mathrm{IIA}}
      \mapsto
    \exp(-f^{\mathrm{IIB}}_2 + \mathcal{P})
    =
    \exp(-f^{\mathrm{IIB}}_2) \wedge \exp( \mathcal{P} ) \wedge C^{\mathrm{IIA}}\;.
  \end{equation}
  Since $\mathcal{P} = e_9^{\mathrm{IIA}} \wedge e_9^{\mathrm{IIB}}$ squares to zero, this in turn amounts to sending
  $$
    C^{\mathrm{IIA}}
      \mapsto
    C^{\mathrm{IIA}} + e_9^{\mathrm{IIA}} \wedge e_9^{\mathrm{IIB}} \wedge C^{\mathrm{IIA}}
    \,.
  $$
  Observe then that fiber integration $(\pi_9^{\mathrm{IIA}})_\ast$ applied to wedge products with $e_9^{\mathrm{IIA}} \wedge e_9^{\mathrm{IIB}}$
  amounts to wedge product with $e_9^{\mathrm{IIB}}$ of the part that does not contain a factor $e_9^{\mathrm{IIA}}$:
  $$
    (\pi_9^{\mathrm{IIA}})_\ast
    \left(
      e_9^{\mathrm{IIB}} \wedge e_9^{\mathrm{IIA}} \wedge (\cdots)
    \right)
    =
    -e_9^{\mathrm{IIB}} \wedge ( \cdots )|_{8+1}
  $$
  (since $e_9^{\mathrm{IIA}}$ squares to zero and using equations (\ref{decomp}) and (\ref{fiberintegration})).
  Summing up, it follows that acting on elements of the form $\exp(f^{\mathrm{IIA}}_2) \wedge C^{\mathrm{IIA}}$ we have
  \begin{equation}
    (\pi_9^{\mathrm{IIA}})_\ast \circ \nu^\ast \circ (\pi_9^{\mathrm{IIB}})^\ast
     \;=\;
    (\pi^{\mathrm{IIA}}_9)_\ast(-) - e^9_B \wedge(-)\vert_{8+1}
    \,.
  \end{equation}
  Comparison of the right hand side
  with the proof of theorem \ref{TheTDualityIsoOnLInfinityCocycles} (equations \fbox{D1} through \fbox{D9}) shows that this is precisely what establishes
  the T-duality isomorphism between twisted cocycles in general, and what identifies the type IIA D-brane cocycles
  with those of type IIB in particular.
\end{proof}

\begin{remark}[Topological T-duality III]
  \label{TopologicalTDualityIII}
  In view of Remark \ref{BunkeStyleTDuality}, the isomorphism in Proposition \ref{TDualityViaPullPush}
  is seen to be the super $L_\infty$-algebraic analog of the key result of classical topological
  T-duality \cite[Theorem 3.13]{BunkeSchick05}: The pull-push integral transform through a topological
  T-duality correspondence yields an isomorphism on twisted K-groups.
  Moreover, via the explicit equation (\ref{nuastonexpf2C})
  the identification from equation (\ref{HoriFormulaInProp}) is
  $$
    \exp(-f^{\mathrm{IIB}}_2)\wedge C^{\mathrm{IIB}}
      =
    (\pi_9^{\mathrm{IIA}})_\ast \left(
      \exp(\mathcal{P}) \wedge (\pi_9^{\mathrm{IIB}})^\ast \left( \exp(-f^{\mathrm{IIA}}_2) \wedge C^{\mathrm{IIA}} \right)
    \right)
    \,.
  $$
  This is precisely the form known as the \emph{Buscher rules for RR-fields}, or the \emph{Hori formula}
  \cite[equation (1.1)]{Ho}. Here we obtained this from just the form of the super-cocycles, to which the
  bifermionic summand of more general RR-fields are constrained to be equal at the level of super-tangent spaces.
\end{remark}

\section{T-Duality Lie 2-algebra}
\label{TDualityLie2Algebra}

It was proposed in \cite{Hull05} \cite{Hull07} that there ought to be a concept of ``T-folds''
which generalizes that of manifolds to a situation where diffemorphisms may be
accompanied by T-duality transformations. In \cite{Nikolaus14} it was
claimed that the correct mathematical formulation of this concept
is by spaces associated to principal 2-bundles (\cite{NSS12}) for structure 2-group a certain
``T-duality 2-group''.
We show now that the extended supergeometry implied by the cocycle
$\mathfrak{L}(\mu_{{}_{F1/Dp}}^{\mathrm{II}})/\mathbb{R}$
from Example \ref{ReductionOfIIACocyclesTo9d}
provides a systematic derivation of this structure, infinitesimally,
but including the supergeometric aspects.

\medskip
First of all, we consider the following sub-$L_\infty$-algebra of
$\mathfrak{L}(\mathrm{KU}/BU(1))/\mathbb{R}$:

\begin{defn}
 \label{Lie2AlgebraOfTDuality2Group}
 The {\it delooped T-duality Lie 2-algebra} is given by
  $$
    \mathrm{CE}(b\mathcal{T}_1)
      =
    \left\{\mathbb{R}[c_2,\tilde{c}_2,h_3];\quad
       \begin{array}{l}
         d c_2 = 0\,,\;\;\;\;d \tilde c_2 = 0
         \\
         \quad d h_3 = -c_2 \wedge \tilde c_2
       \end{array}
    \right\}
   \,.
  $$
\end{defn}


\begin{remark}[T-duality Lie 2-group]
  \label{TDuality2Group}
   The delooped T-duality Lie 2-algebra $b \mathcal{T}(1)$ from Definition \ref{Lie2AlgebraOfTDuality2Group}
   is the $L_\infty$-algebra corresponding to a smooth 2-group $T(1)$,
  the \emph{T-duality 2-group}, with smooth delooping 2-stack
   $\mathbf{B}T(1)$.
  The constrruction of the latter follows from \cite[section 3.2]{FSS13a}: Consider the homotopy fiber
  \[
  \xymatrix{
    \mathbf{B}T(1)
       \ar[rr]\ar[d]_{\mathrm{hofib}(\mathbf{c}_1 \cup \mathbf{c}'_1)}
       &&
       {*}\ar[d]
       \\
     \mathbf{B}U(1)\times \mathbf{B}U(1)
       \ar[rr]^-{\mathbf{c}_1 \cup \mathbf{c}'_1}
     && \mathbf{B}^3 U(1)
     \;,
  }
  \]
  where $\mathbf{B}U(1)$ is the smooth stack of principal $U(1)$-bundles, $\mathbf{B}^3U(1)$ is the smooth stack of $\mathbf{B}^2 U(1)$-principal 2-bundles (bundle 2-gerbes) and $\mathbf{c}_1 \cup \mathbf{c}'_1$ is
  external cup product of the universal first Chern class with itself,
  regarded as a morphism of smooth stacks. Here $\mathbf{c}_1 \cup \mathbf{c}'_1$
  is a homomorphism of smooth higher group stacks, and so its homotopy fiber inhetits group structure, too.
This means that a $T(1)$-principal bundle over a smooth (super-)manifold $X$ is the datum of two principal $U(1)$-bundles on $X$ together with a trivialization of the product of their first Chern classes. One manifestly sees that, translated in terms of Chevalley-Eilenberg algebras, this is precisely the content of Definition  \ref{Lie2AlgebraOfTDuality2Group}.

Finally, forgetting the smooth structure and passing to topological geometric realization ${\vert -\vert}$
(as in \cite{dcct}), then the
above homotopy fiber of smooth higher stacks becomes the ordinary homotopy fiber of topological spaces
$$
\xymatrix{
  \vert \mathbf{B}T(1)\vert
    \ar[rr]^-{\mathrm{hofib}(c_1 \cup c'_1)} &&
  K(\mathbb{Z},2)
   \times
  K(\mathbb{Z},2)
  \ar[r]^-{c_1 \cup c'_1} &
  K(\mathbb{Z},4)
  }
  \,,
$$
where now $c_1$ and $c'_1$ are the ordinary (non-stacky) universal first Chern classes and we form their
ordinary cup product. This
identifies the geometric realization ${\vert \mathbf{B}T(1)\vert}$ with the space that was identified
as the classifying space for topological T-duality pairs in \cite[thm. 2.17]{BunkeSchick05}.
\end{remark}
Indeed, we  now discuss how this is reflected on the level of super $L_\infty$-algebras, seeing that
the T-duality Lie 2-algebra $b \mathcal{T}_1$ sits inside the coefficients
$\mathcal{L}(\mathrm{KU}/\mathrm{BU}(1))/S^1$ as the classifying $L_\infty$-algebra for the ``T-duality'' pairs
consisting of circle fibrations and degree-3 classes, while the rest of $\mathcal{L}(\mathrm{KU}/\mathrm{BU}(1))/S^1$
encodes the K-theory classes on this background.

\begin{prop}
  \label{TDuality2GroupAction}
  There is a homotopy fiber sequence
  $$
    \xymatrix{
      \mathfrak{l}(\mathrm{KU} \oplus \Sigma \mathrm{KU})
      \ar[r]
      &
      \mathfrak{l}\mathcal{L}(\mathrm{KU}/BU(1))/ S^1
      \ar[d]
      \\
      &
      b\mathcal{T}_1
    }
  $$
    which exhibits the cyclified $L_\infty$-algebras in Prop. \ref{cyclofku/BU1}
  as homotopy quotients of $\mathfrak{l}(\mathrm{KU})$ and $\mathfrak{l}(\Sigma\mathrm{KU})$,
    respectively,
  by the Lie 2-algebra of the T-duality 2-group (Prop. \ref{Lie2AlgebraOfTDuality2Group})
  $$
    \mathfrak{l}( \mathcal{L}( \mathrm{KU} /BU(1))/ \mathbb{R} )
     \simeq
    \mathfrak{l}( (\mathrm{KU} \oplus \Sigma \mathrm{KU})/ \mathcal{T}_1 )\;.
  $$
\end{prop}
\begin{proof}
  Consider the corresponding dual diagram of CE-algebras:
  $$
  \hspace{-2mm}
    \begin{array}{ccccc}
      \mathrm{CE}(\mathfrak{l}(\mathrm{KU} \oplus \Sigma \mathrm{KU} ))
      &\longleftarrow&
      \mathrm{CE}(\mathfrak{l}( \mathcal{L}( \mathrm{KU} /BU(1))/ S^1 ))
      &\longleftarrow&
      \mathrm{CE}(b\mathcal{T}_1)
      \\
      \\
    \left\{
       \begin{array}{l}
         d \omega_{2p+2} = 0
         \\
         d \omega_{2p+3} = 0
       \end{array}
    \right\}
      &\longleftarrow&
    \left\{
       \begin{array}{l}
         d c_2 = 0\,,\;\;\;\;d \tilde c_2 = 0
         \\
         d h_3 = - c_2 \wedge \tilde c_2
         \\
         d \omega_{2p+2} = h_3 \wedge \omega_{2p} + c_2 \wedge \omega_{2p+1}
         \\
         d \omega_{2p+3} = h_3 \wedge \omega_{2p+1} + \tilde c_2 \wedge \omega_{2p+2}
       \end{array}
    \right\}
      &\longleftarrow&
    \left\{
       \begin{array}{l}
         d c_2 = 0\,,\;\;\;\;d \tilde c_2 = 0
         \\
         d h_3 = - c_2 \wedge \tilde c_2
       \end{array}
    \right\}.
    \end{array}
  $$
  The morphism on the right is the dual of a fibration, according to
  Prop. \ref{modelstructureOnLInfinityAlgebra}, since it is an inclusion of generators, and the morphism on the left is clearly its
  ordinary cofiber, hence is a model for its homotopy cofiber.
\end{proof}

%
%

By the construction from Section \ref{TDuality} one sees that
the triple $(c_2^{\mathrm{IIA}}, c_2^{\mathrm{IIB}}, \mu_{{}_{F1}}^9)$ defines a $\mathcal{T}_1$-valued cocycle on $\mathbb{R}^{8,1\vert \mathbf{16} + \mathbf{16}}$. This leads to the following
interrelations between models for 9-dimensional spacetime, their T-folds,
and the dimensionally reduced twisted K-theory.

\begin{defn}\label{HomotopyFiberOfReducedtypeIICocycle}
Define the 9-dimensional {\it T-fold super Lie 2-algebra }
  $\mathbb{R}^{8,1\vert \mathbf{16} + \mathbf{16}}_{\mathfrak{dbl}_A }$ to be the
  homotopy fiber of dimensional reduced IIA-fields,
  and define $\mathbb{R}^{8,1\vert \mathbf{16} + \mathbf{16}}_{\mathfrak{dbl}_B }$
  to be the homotopy fiber of the dimensional reduced IIB-fields according to
  Example \ref{ReductionOfIIACocyclesTo9d}
  $$
    \xymatrix@=4.2em{
      \mathbb{R}^{8,1\vert \mathbf{16} + \mathbf{16}}_{\mathfrak{dbl}^{\mathrm{IIA}}(K) }
      \ar[r]
      \ar[drr]_{\hspace{-15mm}\mathrm{hofib}(\mathfrak{L}(\mu_{{}_{F1/Dp}}^{\mathrm{IIA}})/\mathbb{R})}
      &
      \mathbb{R}^{8,1\vert \mathbf{16} + \mathbf{16}}_{\mathfrak{dbl}^{\mathrm{IIA}} }
      \ar[dr]
      && \mathfrak{L}(\Sigma \mathrm{KU}/BU(1))/\mathbb{R}
      \ar[r]
      &
      b \mathcal{T}_1
      \\
      &
      & \fbox{$\mathbb{R}^{8,1\vert\mathbf{16}+\mathbf{16}}$}
      \ar[dr]_{\hspace{-5mm}\mathfrak{L}(\mu_{{}_{F1/Dp}}^{\mathrm{IIA}})/\mathbb{R}}
      \ar[ur]^{\mathfrak{L}(\mu_{{}_{F1/Dp}}^{\mathrm{IIB}})/\mathbb{R}}
      \ar[urr]_{~~(c_2^{\mathrm{IIB}}, c_2^{\mathrm{IIA}}, \mu_{{}_{F1}}^9)}
      \ar[drr]^{~~(c_2^{\mathrm{IIA}}, c_2^{\mathrm{IIB}}, \mu_{{}_{F1}}^9)}
      &&
      \\
      \mathbb{R}^{8,1\vert \mathbf{16} + \mathbf{16}}_{\mathfrak{dbl}^{\mathrm{IIB}}(K) }
      \ar[r]
      \ar[urr]^{\hspace{-5mm}{\mathrm{hofib}(\mathfrak{L}(\mu_{{}_{F1/Dp}}^{\mathrm{IIB}})/\mathbb{R})}}
      &
      \mathbb{R}^{8,1\vert \mathbf{16} + \mathbf{16}}_{\mathfrak{dbl}^{\mathrm{IIB}} }
      \ar[ur]
      &&
      \mathfrak{L}(\mathrm{KU}/BU(1))/\mathbb{R}
      \ar[r]
      &
      b \mathcal{T}_1\;.
    }
  $$
\end{defn}

Of course by Theorem \ref{TheTDualityIsoOnLInfinityCocycles} these two homotopy fibers are going to be
equivalent, but for our purposes it is interesting to make explicit how they are equivalent:

\begin{prop}\label{TheTFoldAlgebra}
The canonical model for the
Chevalley-Eilenberg algebra of the super $L_\infty$-algebras
$\mathbb{R}^{8,1\vert \mathbf{16} + \mathbf{16}}_{\mathfrak{dbl}^{A/B}}$
from Def. \ref{HomotopyFiberOfReducedtypeIICocycle} is
$$
  \mathrm{CE}\big(\mathbb{R}^{8,1\vert \mathbf{16} + \mathbf{16}}_{\mathfrak{dbl}^{\mathrm{IIA}}}\big)
   \;=\;
  \left\{
    \begin{aligned}
      d\psi^\alpha & = 0
      \\
      d e^a & = \overline{\psi} \Gamma_a \psi \;\;\vert \;\; a \leq 8
      \\
      d e_9^{\mathrm{IIA}} &= \overline{\psi}\Gamma_9^{\mathrm{IIA}} \psi
      \\
      d e_9^{\mathrm{IIB}} & = \overline{\psi} \Gamma_9^{\mathrm{IIB}} \psi
      \\
      d f_2 & = \mu_{{}_{F1}}^\mathrm{IIA}
    \end{aligned}
  \right\}\;,\;\;\;
  \mathrm{CE}\big(\mathbb{R}^{8,1\vert \mathbf{16} + \mathbf{16}}_{\mathfrak{dbl}^{\mathrm{IIB}}}\big)
   \;=\;
  \left\{
    \begin{aligned}
      d\psi^\alpha & = 0
      \\
      d e^a & = \overline{\psi} \Gamma_a \psi \;\;\vert \;\; a \leq 8
      \\
      d e_9^{\mathrm{IIA}} &= \overline{\psi}\Gamma_9^{\mathrm{IIA}} \psi
      \\
      d e_9^{\mathrm{IIB}} & = \overline{\psi} \Gamma_9^{\mathrm{IIB}} \psi
      \\
      d f_2 & = \mu_{{}_{F1}}^\mathrm{IIB}
    \end{aligned}
  \right\}\;.
$$
These are just the objects in the top left and top right of the ``topological'' T-duality diagram in
Prop. \ref{BunkeStyleTDualityInfinitesimally}:
$$
  \mathbb{R}^{8,1\vert \mathbf{16} + \mathbf{16}}_{\mathfrak{dbl}^{\mathrm{IIA}}}
  =
     p_{{}_A}^\ast
     \widehat{
       \mathbb{R}^{9,1\vert \mathbf{16} +\overline{\mathbf{16}}}
     }
     \;\,,\qquad \qquad
  \mathbb{R}^{8,1\vert \mathbf{16} + \mathbf{16}}_{\mathfrak{dbl}^{\mathrm{IIB}}}
  =
     p_{{}_B}^\ast
     \widehat{
       \mathbb{R}^{9,1\vert \mathbf{16} +\overline{\mathbf{16}}}
     }
     \,.
$$
\end{prop}
\begin{proof}
We consider the case of type IIA, the case for IIB is of course
directly analogous.
To find the homotopy fiber, we may replace the point inclusion by a fibration
and then take the ordinary pullback. Dually, by Prop. \ref{modelstructureOnLInfinityAlgebra}
we need to find an inclusion of  $\mathfrak{L}(\mathrm{KU}/BU(1))/\mathbb{R}$
into CE-algebra whose underlying co-unary cochain complex is null, and then take the ordinary pushout
of $(c_a^{\mathrm{IIA}}, c_2^{\mathrm{IIB}}, \mu_{{}_{F1}}^9)$ along that,
$$
  \xymatrix{
    \mathrm{CE}\big( \mathbb{R}_{\mathfrak{dbl}}^{8,1\vert \mathbf{16}+ \mathbf{16}} \big)
    &&& A
    \ar[lll]
    \ar[r]^\simeq & 0
    \\
    \mathrm{CE}(\mathbb{R}^{8,1\vert \mathbf{16}+ \mathbf{16}})
    \ar[u]
    &&&
   \mathrm{CE}(b\mathcal{T}_1)\;.
    \ar@{^{(}->}[u]
    \ar@{}[ulll]|{\mbox{\tiny (po)}}
    \ar[lll]^-{\left( c_2^{\mathrm{IIA}},~ c_2^{\mathrm{IIB}},~  \mu_{{}_{F1}}^9 \right)}
  }
$$
To build such an $A$, we first need to adjoin generators $e_9^{\mathrm{IIA}}$ and $e_9^{\mathrm{IIB}}$ to $\mathrm{CE}(b\mathcal{T}_1)$ to
render the cocycles $c_2$ and $\tilde c_2$ trivial, by setting
$$
  d e_9^{\mathrm{IIA}}:= c_2 \;,\qquad d e_9^{\mathrm{IIB}}:= \tilde c_2
  \,.
$$
Furthermore, we need to kill $h_3$, which is a cocycle in the
underlying dual chain complex of $b\mathcal{T}_1$. However, adjoining a generator $f_2$
with $d f_2 = h_3$ would fail the condition that $d^2 \theta = 0$. However, to remedy that
we may set
$$
  d f_2 = h_3 + e_9^{\mathrm{IIA}} \wedge \tilde c_2\;,
$$
which is consistent with $d^2 \theta_2 = 0$
(recall that $d h_3 = - c_2 \wedge \tilde c_2$ according to Prop. \ref{CoefficientsIsoOnTDuality})
yet still removes $h_3$ in the cohomology of the underlying dual chain complex
(since the co-unary restriction of $d$ on $\theta_2$ is $h_3$).
In conclusion, a possible choice for $A$ is the quotient
$$
  A
   :=
    {{\raisebox{.2em}{${\mathrm{CE}(b\mathcal{T}_1)[ \{e_9^{\mathrm{IIA}}, e_9^{\mathrm{IIB}}, \theta_2\} ]}$}
\left/\raisebox{-.2em}{$\left(
    { d e_9^{\mathrm{IIA}} = c_2, ~ d e_9^{\mathrm{IIB}} = \tilde c_2 } \atop { d f_2 = h_3 + e_9^{\mathrm{IIA}} \wedge \tilde c_2 }
  \right)$}\right.}}\;.
$$
Now the pushout in question is directly read off, using the values of $(c_2^{\mathrm{IIA}}, c_2^{\mathrm{IIB}}, h_3)$
from section \ref{TDuality}. So
$c_2$ gets identified with $\overline{\psi}\Gamma_9 \psi$,
$\tilde c_2$ gets identified with $\overline{\psi}\Gamma_9^{\mathrm{IIB}} \psi$
and   $h_3$ gets identified with $\mu_{{}_{F1}}^{\mathrm{IIA}}|_{8+1}$ in the pushout,
while $\tilde c$ is identified with $-(\pi_9^{\mathrm{IIA}})_\ast(\mu_{{}_{F1}}^{\mathrm{IIA}})$. Hence $h_3 + e_9^{\mathrm{IIA}} \wedge \tilde C$ gets identified
with
$$
  \mu_{{}_{F1}}^{\mathrm{IIA}} |_{{}_{8+1}} - e_9^{\mathrm{IIA}} \wedge (\pi_9^{\mathrm{IIA}})_\ast(\mu_{{}_{F1}}^{\mathrm{IIA}}) =
  \mu_{{}_{F1}}
$$
as claimed.
\end{proof}

As a consequence, this shows how the equivalence in
Prop. \ref{BunkeStyleTDualityInfinitesimally} follows from
Theorem \ref{TheTDualityIsoOnLInfinityCocycles}, since the operation of forming
homotopy fibers sends weak equivalences to weak equivalences.
  $$
    
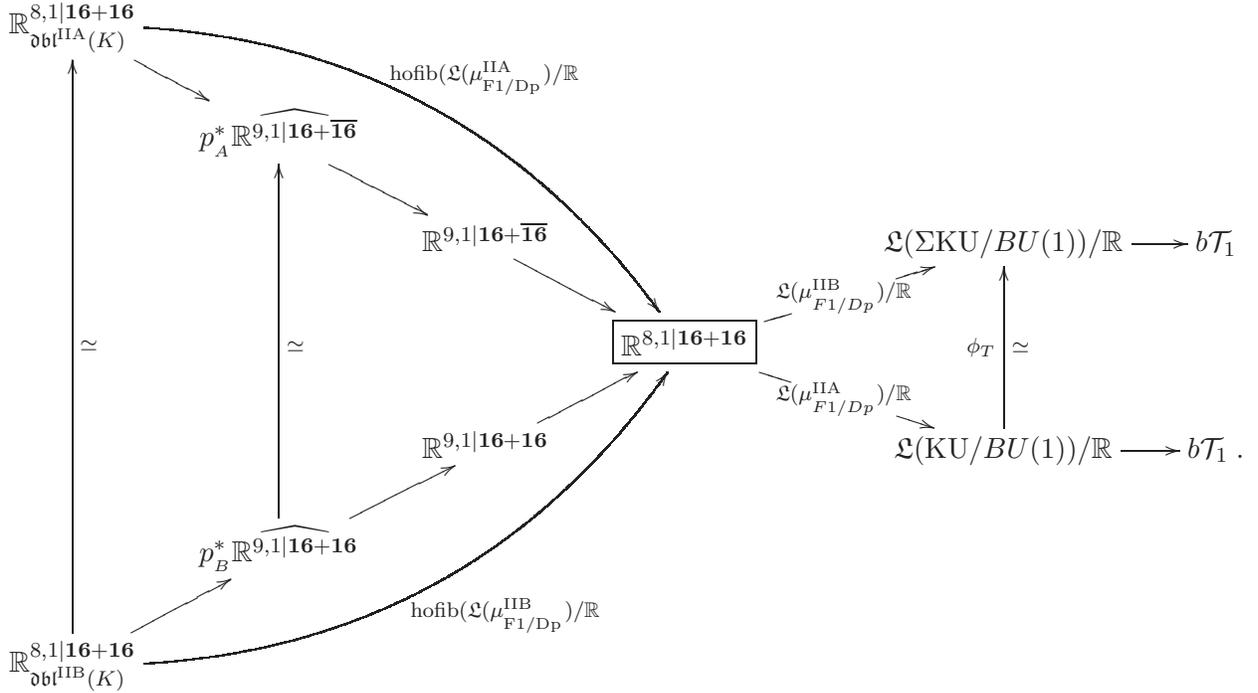
\xymatrix
   @=1.65em{
      \mathbb{R}^{8,1\vert \mathbf{16} + \mathbf{16}}_{\mathfrak{dbl}^{\mathrm{IIA}}(K) }
     \ar@/^3pc/[dddrrr]^{~\mathrm{hofib(\mathfrak{L}(\mu^{\mathrm{IIA}}_{{}_{F1/Dp}}) }/\mathbb{R}  }
      \ar[dr]
      \\
      &
     p_{{}_A}^\ast\widehat{ \mathbb{R}^{9,1\vert \mathbf{16} + \overline{\mathbf{16}}}}
      \ar[dr]
      \\
      &
      &
      \mathbb{R}^{9,1\vert \mathbf{16}+\overline{\mathbf{16}}}
      \ar[dr]
      &
      &
      & \mathfrak{L}(\Sigma \mathrm{KU}/BU(1))/\mathbb{R}
      \ar[r]
      &
      b \mathcal{T}_1
      \\
      &
      &
      &
      \fbox{$\mathbb{R}^{8,1\vert\mathbf{16}+\mathbf{16}}$}
      \ar[drr]|{\mathfrak{L}(\mu_{{}_{F1/Dp}}^{\mathrm{IIA}})/\mathbb{R}}
      \ar[urr]|{\mathfrak{L}(\mu_{{}_{F1/Dp}}^{\mathrm{IIB}})/\mathbb{R}}
      &&
      \\
      &&
      \mathbb{R}^{9,1\vert \mathbf{16} + \mathbf{16}}
      \ar[ur]
      &&&
      \mathfrak{L}(\mathrm{KU}/BU(1))/\mathbb{R}
      \ar[uu]_\simeq^{\phi_T}
      \ar[r]
      &
      b \mathcal{T}_1\;.
      \\
      &
      p_{{}_B}^\ast \widehat{ \mathbb{R}^{9,1\vert \mathbf{16} + {\mathbf{16}}}}
      \ar[ur]
      \ar[uuuu]_\simeq
      \\
      \mathbb{R}^{8,1\vert \mathbf{16} + \mathbf{16}}_{\mathfrak{dbl}^{\mathrm{IIB}}(K) }
      \ar@/_3pc/[uuurrr]_{~~~~\mathrm{hofib(\mathfrak{L}(\mu^{\mathrm{IIB}}_{{}_{F1/Dp}}) }/\mathbb{R}  }
      \ar[uuuuuu]_\simeq
      \ar[ur]
    }
  $$

\section{F-theory}
\label{SectionOnF}

We have introduced in Definitionn \ref{doubledtypeIIspacetime} the ``doubled'' correspondence super Lie algebra
$$
  \mathbb{R}^{8+(1+1),1\vert 32}
    \;:=\;
  \mathbb{R}^{9,1\vert \mathbf{16} + \mathbf{16}} \underset{\mathbb{R}^{8,1\vert \mathbf{16} + \mathbf{16}}}{\times}
\mathbb{R}^{9,1\vert \mathbf{16} + \overline{\mathbf{16}}}
$$
of bosonic dimenion $9+2$ that serves to interpolate between
type IIA and type IIB super-spacetime via T-duality. It is immediate to see that this
inherits the D0-brane 2-cocycle of its type IIA factor and hence extends to a
super Lie algebra of bosonic dimension 10+2:

\begin{defn} \label{FTheorySuperLieAlgebra}
  Write $\mathbb{R}^{9+(1+1),1\vert 32}$ for the central super Lie algebra extension of
  $\mathbb{R}^{8+(1+1),1\vert 32}$ (def. \ref{doubledtypeIIspacetime}) classified by the 2-cocycle $p_{{}_A}^\ast c_2^M$, i.e. fitting into
a homotopy fiber sequence of the following form:
  $$
    \xymatrix{
      \mathbb{R}^{9+(1+1),1\vert 32}
      \ar[d]_-{\mathrm{hofib}(p_{{}_A}^\ast c_2^M)}
      \\
      \mathbb{R}^{8+(1+1),1\vert 32}
      \ar[rr]^-{p_{{}_A}}
      &&
      \mathbb{R}^{9,1\vert \mathbf{16} +\overline{\mathbf{16}}}
      \ar[rr]^-{c_2^M}
      &&
      b\mathbb{R}\;.
    }
  $$
  where $c_2^M$ is from Def. \ref{IIBAsExtension}.
\end{defn}
This super Lie algebra reflects both the relation between
the M-theoretic 11d spacetime and the type IIA super-spacetime, as well as the correspondence
of the latter to the type IIB super-spacetime.

\begin{remark}\label{FtheoryAlgebraExplicitly}
  Unwinding the definition, the CE-algebra of the super Lie algebra $\mathbb{R}^{9+(1+1),1\vert 32}$
of Def. \ref{FTheorySuperLieAlgebra} has the following generating relations
  $$
         d e_0  = \overline{\psi}\Gamma_0\psi\;,
          \;\;
         d e_1  = \overline{\psi} \Gamma_1 \psi\;,
         \;\;
         \hdots
        \;,\;
         d e_8  = \overline{\psi}\Gamma_8 \psi\;,
         \;\;
          d e_9^{\mathrm{IIB}} = \overline{\psi}\Gamma_9^{\mathrm{IIB}} \psi\;,
         \;\;
          d e_9 = \overline{\psi} \underset{\sigma_1}{\underbrace{\Gamma_9}} \psi\;,
          \;\;
          d e_{10} = \overline{\psi} \underset{\sigma_3}{\underbrace{\Gamma_{10}}} \psi\;,
  $$
  where on the right we have the 12 algebra elements from Def. \ref{IIBCliffordGenerators}.
\end{remark}

\medskip

Since the homotopy fiber of $c_2^M$ alone is the local model for M-theoretic spacetime,
hence for non-perturbative type IIA string theory, this suggests that the above homotopy fiber
$\mathbb{R}^{9+(1+1),1\vert 32}$
of $p_A^\ast c_2^M$ is similarly related to a non-perturbative description of type IIB
string theory.
That type IIB string theory ought to have a non-perturbative description in terms of 10+2 dimensional fibrations
over 10-dimensional super-spacetime is known as the \emph{F-theory} conjecture, due to \cite{Vafa96}.
In order to formalize aspects of this, we recall how this conjecture is motivated, see also for instance \cite{Johnson97}:

\medskip
\noindent {\bf The motivation of the F-theory conjecture from the M-theory conjecture.}
\begin{enumerate}
 \item Assume that a non-perturbative completion of type IIA string theory exists, given by a geometric theory on a 10+1 dimensional Riemannian circle fibration (M-theory).
\item  Pass to the limit that the radius $R_M$ of the circle fiber is infinitesimal to obtain perturbative type IIA string theory
with coupling constant
$$
  g_{{}_{\mathrm{IIA}}} = R_{{}_M}/\ell_s
 \,,
$$
where $\ell_s$ is the string scale.
\item Consider the situation when the 10d IIA spacetime is itself a 9+1 dimensional circle fibration with circle fiber $S^1_{{}_A}$,
so that in total the original 11d spacetime is a Riemannian torus fibration over 9d with fiber $S^1_{{}_M} \times S^1_{{}_A}$.
\item Invoke perturbative T-duality to find an equivalent perturbative type IIB string theory on a 9+1-dimensional fibration with dual circle fiber $S^1_{{}_B}$ of radius
$$
  \begin{aligned}
    R_{{}_B} & = \ell_s^2 / R_{{}_A}
 \end{aligned}
$$
By the rules of perturbative T-duality, the coupling constant of the IIB theory is
$$
  \begin{aligned}
     g_{{}_{\mathrm{IIB}}} & = g_{{}_{\mathrm{IIA}}} \tfrac{\ell_s}{R_{{}_A}}
    \\
     & = R_{{}_M}/R_{{}_A}\,.
  \end{aligned}
$$
Note that the last two steps can be combined \cite{Sch96}: M-theory compactified on a torus is supposed to be equivalent to type
IIB superstring theory compactified on a circle in the limit of small  volume
of the torus.

\item While this was ``derived'' for infinitesimal $R_{{}_M}$, observe that the resulting formula for the IIB coupling
evidently extrapolates to finite $R_{{}_M}$ and then says that the IIB theory remains weakly coupled for finite (large)
$R_{{}_M}$ if only $R_{{}_A}$ is suitably scaled along. Regard this as evidence for a non-perturbative version of
T-duality which relates the non-perturbative IIA theory (M-theory) with some non-perturbative completion of
IIB string theory, to be called \emph{F-theory}.
\item Collect the geometric data that went into this construction:
Retain information both of the $S^1_{{}_A}$-fiber, hence of the doubled correspondence space, and
of the $S^1_{{}_M}$-fiber to obtain in total a 10+2-dimensional $S^1_{{}_M}
\times S^1_{{}_A}$-fiber bundle over 10d type IIB spacetime.
\item Since the coupling constant $g_{{}_{\mathrm{IIB}}}$ depends only on the ratio
$R_{{}_M}/R_{{}_A}$, hence only
on the complex structure of this torus, conclude that this is to be regarded as a 10+2-dimensional elliptic fibration.
\item
 Check that the action of S-duality on type IIB fields corresponds to the automorphisms of the
 elliptic fiber.
\end{enumerate}

\medskip

We now observe that the super Lie algebra from Def. \ref{FTheorySuperLieAlgebra}
has just the right structure to be the model for the super-tangent space of this F-theory elliptic fibration.

\begin{prop}\label{TheFTheorySpacetimeInContext}
The super Lie algebra $\mathbb{R}^{9+(1+1),1|32}$ from Def. \ref{FTheorySuperLieAlgebra}
fits into a diagram of super $L_\infty$-algebras of the following form
$$
  
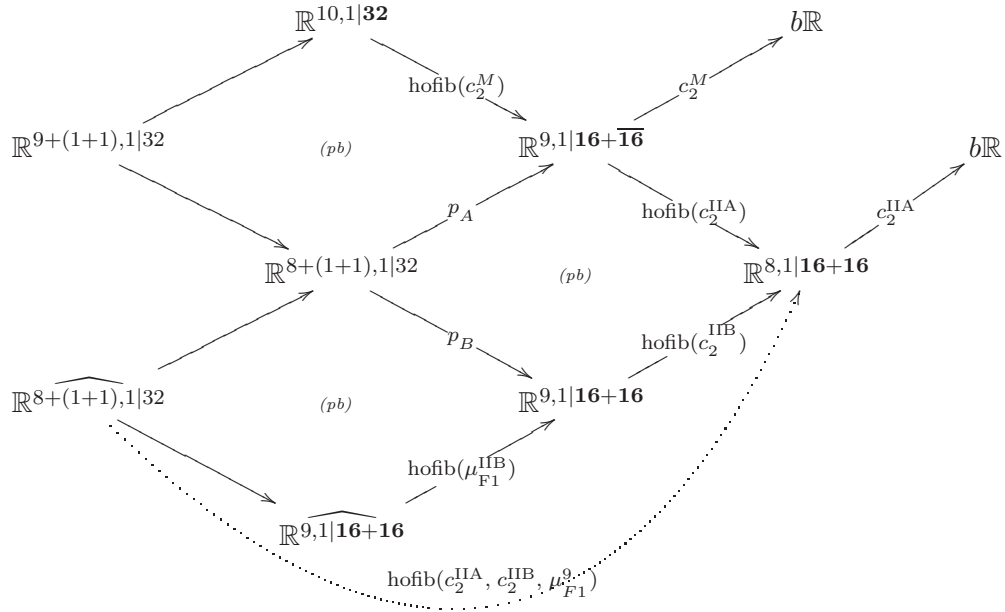
\xymatrix@=2.8em{
    & &
    \mathbb{R}^{10,1\vert \mathbf{32}}
    \ar[dr]|{\mathrm{hofib}(c^M_2)~}
    & & b \mathbb{R} &
    \\
    & \mathbb{R}^{9+(1+1),1\vert {32}}
    \ar@{}[rr]|{\mbox{\tiny (pb)}}
    \ar[ur]
    \ar[dr]
    &
    &
    \mathbb{R}^{9,1\vert \mathbf{16}+\overline{\mathbf{16}}}
    \ar[ur]|{c^M_2}
    \ar[dr]|{\mathrm{hofib}(c_2^{\mathrm{IIA}})}
    & & b \mathbb{R}&
    \\
    & &
    \mathbb{R}^{8+(1+1),1\vert 32}
    \ar[ur]|{p_{{}_A}}
    \ar[dr]|{p_{{}_B}}
    \ar@{}[rr]|{\mbox{\tiny (pb)}}
    &
    &
    \mathbb{R}^{8,1\vert \mathbf{16} + \mathbf{16}}
    \ar[ur]|{ c_2^{\mathrm{IIA}} }
    &&
    \\
    &
    \widehat{\mathbb{R}^{8+(1+1),1\vert 32}}
    \ar[ur]
    \ar[dr]
    \ar@{}[rr]|{\mbox{\tiny (pb)}}
      \ar@{..>}@/_8.5pc/[urrr]^{\hspace{5mm} \mathrm{hofib}(c_2^{\mathrm{IIA}},\;c_2^{\mathrm{IIB}},\; \mu_{{}_{F1}}^9)}
    &
    &
    \mathbb{R}^{9,1\vert \mathbf{16}+{\mathbf{16}}}
    \ar[ur]|{\mathrm{hofib}(c_2^{{}^{\mathrm{IIB}}})}
    &
    &&
    \\
    &
    &
    \widehat{\mathbb{R}^{9,1\vert \mathbf{16}+ \mathbf{16}}}
    \ar[ur]|{\mathrm{hofib( \mu_{{}_{F1}}^{\mathrm{IIB}} )}}
    &
    &
    &
  }
$$
where each square is a (homotopy) pullback square (homotopy Cartesian).
\end{prop}
\begin{proof}
 This follows immediately from the pasting law for homotopy pullbacks of super $L_\infty$-algebras.
 It is also directly checked explicitly.
\end{proof}

In particular, this says that
the composite diagonal morphism
$$
  \mathbb{R}^{8+(1+1),1\vert 32} \longrightarrow \mathbb{R}^{8,1\vert\mathbf{16} + \mathbf{16}}
$$
exhibits its domain as a $\mathrm{Lie}(S^1_{{}_M} \times S^1_{{}_A} )$-fibration
over the type IIB superspacetime.

\begin{remark}
Remark \ref{FtheoryAlgebraExplicitly} says that the bosonically 12-dimensional
$\mathbb{R}^{9+(1+1),1\vert 32}$ is a super Lie algebra, but not a
super-\emph{Minkowski} Lie algebra, hence not a super-symmetry algebra in the sense of
spacetime supersymmetry.
Instead, the diagram in Prop. \ref{TheFTheorySpacetimeInContext} shows how it projects
onto various genuine super-spacetime algebras.
\end{remark}

This is consistent with the observations and assertions
in the literature  that F-theory does not have a straightforward spacetime
 interpretation and, furthermore, that the alternative seems to
 emanate from superalgebras, but not of the usual type.
  In \cite{BVP}, is is shown how  the algebras in 10, 11, and 12 dimensions
  can be described by a web of dualities as different faces of the ortho-symplectic
 superalgebra ${\rm OSp}(1|32)$. In particular, for F-theory
 the corresponding algebra has no vector operator, so that there is
 no generator for translations.  This implies that F-theory, in contrast
 to M-theory, has no straightforward spacetime interpretation.
 It was noted in \cite[end of Sec. 6]{CJLP} that the fact that
 the type IIA and IIB theories have  ${\rm SL}(2,\R)$ and  ${\rm SL}(1|1)$
 global symmetry, respectively, highlights the similarity between the
 symmetries of the two, but more importantly  point to
a possibly fermionic twelve-dimensional origin.

\medskip
Hence, if we regard $\mathbb{R}^{9+(1+1),1\vert 32}$ as a $\mathbb{R}^2$-fibration
over the IIB-spacetime according to Prop. \ref{TheFTheorySpacetimeInContext}, then
the last check in the list of F-theory desiderata from above is that its fiber automorphisms
induce S-duality transformation on the type IIB fields:
\begin{defn}
  \label{IIBCocyclesAsFCocycles}
  Write
  $$
    \mu^F_{{}_{F1}},\; \mu^F_{{}_{Dp}} \;\; \in \mathrm{CE}(\mathbb{R}^{9+(1+1),1\vert 32})
  $$
  for the pullback of the type IIB F1/D$p$-brane cocycles from Def. \ref{IIBBraneCEElements}
  along the projection
  $$
    \mathbb{R}^{9+(1+1),1\vert 32} \longrightarrow \mathbb{R}^{9,1\vert \mathbf{16} + \mathbf{16}}
  $$
  from Prop. \ref{TheFTheorySpacetimeInContext}.
  \end{defn}

  Hence in terms of generators, via Remark \ref{FtheoryAlgebraExplicitly}, these CE-elements
  have the same form as in Def. \ref{IIBBraneCEElements}, but with all generators $e^a$ renamed
  as $e_{\mathrm{IIB}}^a$. Notably the cocycles for the F1- and D1-string on the F-theory space read
  $$
    \mu^F_{{}_{F1}} = i \left( \overline{\psi} \Gamma_a^{\mathrm{IIB}} \Gamma_{10} \psi \right) \wedge e_{\mathrm{IIB}}^a
\qquad \text{and} \qquad
    \mu^F_{{}_{D1}} = i \left( \overline{\psi} \Gamma_a^{\mathrm{IIB}} \Gamma_{9} \psi \right) \wedge e_{\mathrm{IIB}}^a
    \,.
  $$

\begin{prop}
  Rotation in the $(9,10)$-plane of the fiber $\mathbb{R}^2$ of the F-theory super Lie algebra
  from Def. \ref{FTheorySuperLieAlgebra}
  is a super Lie algebra automorphism
  $$
    \begin{array}{ccc}
    \phi_\alpha
     \;:\;
   \mathbb{R}^{9+(1+1),1\vert 32}
     &
     \longrightarrow
     &
   \mathbb{R}^{9+(1+1),1\vert 32}
      \\
      e^a_{\mathrm{IIB}} & \mapsto& e^a_{\mathrm{IIB}}
      \\
      \\
      \left(
        \begin{array}{c}
          e^9
          \\
          e^{10}
        \end{array}
      \right)
        & \mapsto&
        \left(
         \begin{array}{cc}
          \cos(\alpha) & \sin(\alpha)
          \\
          -\sin(\alpha) & \cos(\alpha)
          \end{array}
        \right)
        \cdot
      \left(
        \begin{array}{c}
          e^9
          \\
          e^{10}
        \end{array}
      \right)
      \\
      \\
      \psi &\mapsto& \exp(\tfrac{\alpha}{4} \Gamma_9 \Gamma_{10}) \psi\;.
    \end{array}
  $$
  Under this automorphism the F1-string and the D1-string are turned into superpositions
  of each other, in that their F-theoretic cocycles from Def. \ref{IIBCocyclesAsFCocycles}
  satisfy:
  $$
    \phi_\alpha^\ast
    \left(
      \begin{array}{c}
      \mu_{{}_{D1}}^F
      \vspace{1mm}\\
      \mu^F_{{}_{F1}}
      \end{array}
    \right)
    =
        \left(
         \begin{array}{cc}
          \cos(\alpha) & \sin(\alpha)
          \\
          -\sin(\alpha) & \cos(\alpha)
          \end{array}
        \right)
     \cdot
    \left(
     \begin{array}{c}
      \mu_{{}_{D1}}^F
      \vspace{1mm}\\
      \mu^F_{{}_{F1}}
      \end{array}
    \right) \;.
  $$
\end{prop}
\begin{proof}
  That $\phi_\alpha$ is indeed a super Lie algebra automorphism follows with
  Remark \ref{SpinorToVectorPairingIIB}.
  Hence acting with $\phi_\alpha$ on $\mu^F_{{}_{F1}}$ and $\mu^F_{{}_{D1}}$
  leaves all the factors in the  explicit formula in
  Def. \ref{IIBCocyclesAsFCocycles} invariant, except for $\Gamma_9$ and $\Gamma_{10}$,
  which are rotated into each other.
\end{proof}

Under the identification $\sigma_2 = - \Gamma_{9} \Gamma_{10}$ from Def. \ref{IIBCliffordGenerators}
this is the action of S-duality on the F1/D1 cocycles according to \cite[Remark 4.9]{FSS13}.

\medskip
\medskip

{\bf Acknowledgement.}

\medskip

We thank Jonathan Pridham for discussion of his model structure for $L_\infty$-algebras.
U.S. thanks Thomas Nikolaus for discussion of T-folds.

\noindent D.F. thanks New York University Abu Dhabi (NYUAD)
for kind hospitality during the writing of this paper.

\noindent H.S. would like to thank Ralph Kaufmann and Alexander Voronov,
the organizers of the AMS special sessions on Topology and Physics in Minneapolis in
October 2016, for the opportunity to present the results of this work.

\noindent U.S. thanks the Max Planck Institute for Mathematics (MPI) in Bonn
for kind hospitality during the writing of this paper.
U.S. thanks Alberto Cattaneo for kind hospitality at University of Zurich
and the opportunity to lecture on the results presented here.
U.S. was supported by RVO:67985840.



\begin{thebibliography}{99}

\bibitem{AETW87}
A. Ach{\'u}carro, J. Evans, P. Townsend, D. Wiltshire,
{\it Super $p$-Branes},
Phys. Lett. {\bf B 198} (1987), 441-446,
[\href{http://inspirehep.net/record/22286}{inspirehep.net/record/22286}].


%

\bibitem{AAL}
E. Alvarez, L. Alvarez-Gaum{\'e}, and Y. Lozano,
{\it An introduction to T-duality in string theory},
Nucl. Phys. Proc. Suppl. {\bf 41}  (1995), 1-20,
[\href{https://arxiv.org/abs/hep-th/9410237}{arXiv:hep-th/9410237}].


\bibitem{BS}
I. Bachas and K. Sfetsos,
{\it T-duality and world-sheet supersymmetry},
Phys. Lett. {\bf B349} (1995), 448-457,
[\href{https://arxiv.org/abs/hep-th/9502065}{arXiv:hep-th/9502065}].

\bibitem{BdR}
E. Bergshoeff and M. de Roo,
{\it D-branes and T-duality},
Phys. Lett. {\bf B380} (1996) 265--272,
[\href{https://arxiv.org/abs/hep-th/9603123}{arXiv:hep-th/9603123}].
	


\bibitem{BST86}
E. Bergshoeff, E. Sezgin and P. K. Townsend,
{\it Superstring actions in $D = 3,4,6,10$ curved superspace},
Phys. Lett. {\bf B 169} (1986), 191-196,
[\href{http://inspirehep.net/record/223138}{inspirehep.net/record/223138}].

\bibitem{BST87}
E. Bergshoeff, E. Sezgin and P. K. Townsend,
{\it Supermembranes and eleven dimensional supergravity},
Phys. Lett. {\bf B 189} (1987) 75-78.

 \bibitem{BVP}
E. Bergshoeff and A. Van Proeyen,
{\it The many faces of ${\rm OSp}(1|32)$},
 Class. Quantum Grav. {\bf 17} (2000) 3277,
[\href{http://arxiv.org/abs/hep-th/000326}{arXiv:hep-th/000326}].

\bibitem{Borceux94}
F. Borceaux,
{\it Handbook of categorical algebra I: Basic category theory},
Cambridge University Press (1994)

\bibitem{BouwknegtEvslinMathai04}
P. Bouwknegt, J. Evslin, and V. Mathai,
{\it T-duality: Topology change from $H$-flux},
Commun. Math. Phys. {\bf 249} (2004), 383-415
[\href{http://arxiv.org/abs/hep-th/0306062}{arXiv:hep-th/0306062}].

\bibitem{BM}
P. Bouwknegt and V. Mathai,
{\it D-branes, B-fields and twisted K-theory},
J. High Energy Phys. {\bf 03} (2000) 007,
[\href{https://arxiv.org/abs/hep-th/0002023}{arXiv:hep-th/0002023}].


\bibitem{BuijsFelixMurillo12}
U. Buijs, Y. F{\'e}lix, A. Murillo,
{\it $L_\infty$-models of based mapping spaces},
J. Math. Soc. Japan Volume 63, Number 2 (2011), 503-524.
[\href{https://arxiv.org/abs/1209.4756}{arXiv:1209.4756}]

\bibitem{BunkeSchick05}
U. Bunke, T. Schick,
{\it On the topology of T-duality},
Rev. Math. Phys. 17:77-112 (2005) [\href{https://arxiv.org/abs/math/0405132}{math/0405132}]

\bibitem{BunkeRumpfSchick06}
U. Bunke, P. Rumpf, and T. Schick,
{\it The topology of T-duality for $T^n$-bundles},
Rev. Math. Phys. {\bf 18} (2006), 1103,
\href{http://arxiv.org/abs/math.GT/0501487}{[arXiv:math.GT/0501487].}

\bibitem{Ca}
L. Castellani,
{\it Chiral $D=10$, $N=2$ supergravity on the group manifold I: Free differential algebra and solution of Bianchi identities},
Nucl. Phys. {\bf B 294} (1987), no. 4, 877--889, \newline
[\href{http://inspirehep.net/record/246358}{inspirehep.net/record/246358}].

\bibitem{CDF}
L.~Castellani, R.~D'Auria, and P.~Fr{\'e},
\newblock {\it Supergravity and Superstrings -- A geometric perspective},
\newblock World Scientific, Singapore (1991)

\bibitem{Ced}
M. Cederwall,
{\it Double supergeometry},
J. High Energy Phys. {\bf 06} (2016) 155,
[\href{http://arxiv.org/abs/1603.04684}{arXiv:1603.04684}].



\bibitem{CAIB00}
C. Chryssomalakos, J. de Azc\'arraga, J. Izquierdo, and C. P{\'e}rez Bueno,
{\it The geometry of branes and extended superspaces},
Nucl. Phys. {\bf  B 567} (2000), 293-330,
[\href{http://arxiv.org/abs/hep-th/9904137}{arXiv:hep-th/9904137}].


\bibitem{CJLP}
E. Cremmer, B. Julia, H. Lu, and C. N. Pope,
{\it Dualisation of Dualities, II: Twisted self-duality of doubled fields and superdualities},
 Nucl. Phys. {\bf B535} (1998), 242-292,
\href{https://arxiv.org/abs/hep-th/9806106} {[arXiv:hep-th/9806106].}

\bibitem{CLPS}
M. Cvetic, H. Lu, C. N. Pope, and K. S. Stelle,
{\it T-duality in the Green-Schwarz formalism, and the massless/massive IIA duality map},
	Nucl. Phys. {\bf B573} (2000), 149-176, \newline
\href{https://arxiv.org/abs/hep-th/9907202}{[arXiv:hep-th/9907202].}

\bibitem{DAuriaFre82}
R. D'Auria and P. Fr\'e,
{\it Geometric supergravity in $D=11$ and its hidden supergroup},
Nucl. Phys. {\bf B201} (1982) 101--140.


\bibitem{DFGT}
R. D'Auria, P. Fr\'e, P. Grassi,  and M. Trigiante,
{\it Pure spinor superstrings on generic type IIA supergravity backgrounds},
J. High Energy Phys. {\bf 0807} (2008), 059,
[\href{https://arxiv.org/abs/0803.1703}{arXiv:0803.1703}].


\bibitem{DeligneFreed}
P.~Deligne, D.~Freed,
\newblock {\it Supersolutions},
\newblock in P.~Deligne et al.
\newblock {\it Quantum fields and strings},
\newblock Amer. Math. Soc., Providence, RI, 1999,
\newblock [\href{https://arxiv.org/abs/hep-th/9901094}{arXiv:hep-th/9901094}].

\bibitem{DuffHoweInamiStelle87}
M. Duff, P. Howe, T. Inami, K. Stelle,
{\it Superstrings in $D = 10$ from Supermembranes in $D = 11$},
Phys. Lett. B191 (1987) 70
\href{http://inspirehep.net/record/245249}{inspirehep.net/record/245249}


\bibitem{FF}
P. de Medeiros, J. Figueroa-O'Farrill, E. M\'endez-Escobar, and P. Ritter,
{\it On the Lie-algebraic origin of metric 3-algebras}, Commun. Math. Phys. {\bf 290} (2009), 871--902,
\href{https://arxiv.org/abs/0809.1086}{[arXiv:0809.1086].}



\bibitem{FF2}
J. Figueroa-O'Farrill and J. Sim\'on,
{\it Supersymmetric Kaluza-Klein reductions of M2 and M5-branes},
	Adv. Theor. Math. Phys. {\bf 6} (2003) 703-793,
\href{https://arxiv.org/abs/hep-th/0208107}{[arXiv:hep-th/0208107].}


\bibitem{FSS13a}
 D. Fiorenza, H. Sati, and U. Schreiber,
{\it Extended higher cup-product Chern-Simons theories},
J. Geom. Phys. {\bf 74} (2013), 130--163,
\href{https://arxiv.org/abs/1207.5449}{[arXiv:1207.5449].}

\bibitem{FSS13}
 D. Fiorenza, H. Sati, and U. Schreiber,
\newblock {\it Super Lie $n$-algebra extensions, higher WZW models, and super $p$-branes with tensor multiplet fields},
\newblock \href{http://www.worldscientific.com/doi/abs/10.1142/S0219887815500188}{Intern. J. Geom. Meth. Mod. Phys. {\bf 12} (2015) 1550018 },
\newblock \href{http://arxiv.org/abs/1308.5264}{[arXiv:1308.5264].}

\bibitem{FSS15}
D. Fiorenza, H. Sati, and U. Schreiber,
{\it The WZW term of the M5-brane and differential cohomotopy},
J. Math. Phys. {\bf 56} (2015), 102301,
[\href{http://arxiv.org/abs/1506.07557}{arXiv:1506.07557}.]

\bibitem{FSS16}
D. Fiorenza, H. Sati, and U. Schreiber,
{\it Rational sphere valued supercocycles in M-theory and type IIA string theory},
Journal of Geometry and Physics, Volume 114, April 2017,
\href{https://arxiv.org/abs/1606.03206}{arXiv:1606.03206}



\bibitem{Ha}
S. F. Hassan,
{\it $SO(d,d)$ transformations of Ramond-Ramond fields and space-time spinors},
Nucl. Phys. {\bf B583} (2000) 431-453,
[\href{https://arxiv.org/abs/hep-th/9912236}{arXiv:hep-th/9912236}].


\bibitem{HKS}
M. Hatsuda, K. Kamimura, and W. Siegel,
{\it Superspace with manifest T-duality from type II superstring},
J. High Energy Phys.  {\bf 06} (2014) 039,
[\href{http://arxiv.org/abs/1403.3887}{arXiv:1403.3887}].


\bibitem{HM}
M.~Henneaux and L.~Mezincescu,
\newblock {\it A Sigma model interpretation of Green-Schwarz covariant superstring action},
\newblock Phys. Lett. {\bf B152} (1985), 340--342.

\bibitem{Hess06}
K. Hess,
{\it Rational homotopy theory: a brief introduction},
Interactions between homotopy theory and algebra, 175-202. Contemp. Math 436
\href{http://arxiv.org/abs/math.AT/0604626}{[arXiv:math.AT/0604626].}

\bibitem{Hinich}
V. Hinich,
{\it DG coalgebras as formal stacks}
J. Pure and Applied Algebra
{\bf 162} (2001), 209-250,
\href{https://arxiv.org/abs/math/9812034}{[arXiv:math/9812034].}


\bibitem{Hirschhorn}
P. Hirschhorn,
{\it Model categories and their localizations}, Amer. Math. Soc., Providence, RI, 2003.


\bibitem{Ho}
K. Hori, {\it D-branes, T-duality and index theory},
Adv. Theor. Math. Phys. {\bf 3} (1999) 281, \newline
[\href{https://arxiv.org/abs/hep-th/9902102}{arXiv:hep-th/9902102}].

\bibitem{HuertaSchreiber}
J. Huerta and U. Schreiber,
{\it M-Theory from the superpoint},
\href{http://arxiv.org/abs/1702.01774}
	{[arXiv:1702.01774] [hep-th].}

\bibitem{Hull05}
C. Hull,
{\it A geometry for non-geometric string backgrounds},
J. High Energy Phys. {\bf 0510} (2005), 065,
\href{http://arxiv.org/abs/hep-th/0406102}{[arXiv:hep-th/0406102].}


\bibitem{Hull07}
C. Hull,
{\it Doubled geometry and T-folds},
J. High Energy Phys. {\bf 0707} (2007) 080, \newline
[\href{http://arxiv.org/abs/hep-th/0605149}{arXiv:hep-th/0605149}].

\bibitem{Johnson97}
C. Johnson,
{\it From M-theory to F-theory, with branes},
Nucl. Phys. {\bf B507} (1997) 227-244,
[\href{http://arxiv.org/abs/hep-th/9706155}{arXiv:hep-th/9706155}].

\bibitem{Jones87}
J. D. S. Jones,
{\it Cyclic homology and equivariant homology}, Invent. Math. {\bf 87} (1987), 403-423.


\bibitem{Ka}
A. Kapustin, {\it D-branes in a topologically nontrivial B-field},
Adv. Theor. Math. Phys.
{\bf 4} (2000) 127,
\href{http://arxiv.org/abs/hep-th/9909089}{[arXiv:hep-th/9909089].}

\bibitem{KonechnySchwarz97}
A. Konechny and A. Schwarz,
{\it On $(k \oplus l\vert q)$-dimensional supermanifolds}
in J. Wess and V. Akulov (eds.), Supersymmetry and Quantum Field Theory,
Lecture Notes in Physics 509, Springer (1998),
\href{http://arxiv.org/abs/hep-th/9706003}{[arXiv:hep-th/9706003].}


\bibitem{KR}
B. Kulik and R. Roiban,
{\it T-duality of the Green-Schwarz superstring},
J. High Energy Phys. {\bf 0209} (2002) 007,
[\href{https://arxiv.org/abs/hep-th/0012010}{arXiv:hep-th/0012010}].



\bibitem{LadaStasheff93}
T. Lada and J. Stasheff,
{\it Introduction to sh Lie algebras for physicists},
Int. J. Theo. Phys. {\bf 32} (1993), 1087-1103,
[\href{https://arxiv.org/abs/hep-th/9209099}{arXiv:hep-th/9209099}].

\bibitem{LadaMarkl95}
T. Lada and M. Markl,
{\it  Strongly homotopy Lie algebras},
Commun. in Algebra {\bf 23} (1995), 2147-2161,
[\href{http://arxiv.org/abs/hep-th/9406095}{arXiv:hep-th/9406095}].

\bibitem{LSW}
J. A. Lind, H. Sati, and C. Westerland,
{\it A higher categorical analogue of topological T-duality for sphere bundles},
\href{http://arxiv.org/abs/1601.06285}{[arXiv:1601.06285].}

\bibitem{Loday11}
J.-L. Loday,
{\it Free loop space and homology},
\href{https://arxiv.org/abs/1110.0405}{[arXiv:1110.0405].}

\bibitem{LPSS}
H. Lu, C. N. Pope, E. Sezgin, and K. S. Stelle,
{\it Stainless super $p$-branes},
Nucl. Phys. {\bf B456} (1995), 669-698,
[\href{https://arxiv.org/abs/hep-th/9508042}{arXiv:hep-th/9508042}].


\bibitem{MS}
V. Mathai and H. Sati, {\it Some relations between twisted K-theory and
$E_8$ gauge theory}, J. High
Energy Phys. {\bf 0403} (2004), 016,
[\href{https://arxiv.org/abs/hep-th/0312033}{arXiv:hep-th/0312033}].


\bibitem{Moore14}
G. Moore,
{\it Physical Mathematics and the Future},
talk at Strings 2014
\newline
\url{http://www.physics.rutgers.edu/~gmoore/PhysicalMathematicsAndFuture.pdf}


\bibitem{MooreWitten00}
G. Moore, E. Witten,
{\it Self-duality, Ramond-Ramond Fields, and K-theory},
JHEP 0005 (2000) 032,
[\href{https://arxiv.org/abs/hep-th/9912279}{hep-th/9912279}]

\bibitem{Nikolaus14}
T. Nikolaus,
{\it T-duality in K-theory and elliptic cohomology},
talk at String Geometry Network Meeting, Feb 2014, ESI Vienna.

\bibitem{NSS12}
T. Nikolaus, U. Schreiber, and D. Stevenson,
{\it Principal $\infty$-bundles --  General theory},
J. Homotopy Related Structr. {\bf 10} (2015),  749--801,
[\href{https://arxiv.org/abs/1207.0248}{arXiv:1207.0248}].

\bibitem{NSS12b}
T. Nikolaus, U. Schreiber, and D. Stevenson,
{\it Principal $\infty$-bundles --  Presentations},
J. Homotopy and Related Structr. {\bf 10} (2015), 565-622,
[\href{http://arxiv.org/abs/1207.0249}{arXiv:1207.0249}].

\bibitem{Saz}
B. Nikoli\'c and B Sazdovi\'c,
{\it T-dualization of type II superstring theory in double space}, \newline
[\href{http://arxiv.org/abs/1505.06044}{arXiv:1505.06044}].



\bibitem{PSa}
S. Palmer and C. Saemann,
{\it M-brane models from non-abelian gerbes},
J. High Energy Phys. {\bf 1207} (2012), 010,
[\href{http://arxiv.org/abs/1203.5757}{arXiv:1203.5757}].


\bibitem{PS}
 R. Percacci and E. Sezgin, {\it On target space duality in $p$-branes},
 Mod. Phys. Lett. {\bf A10} (1995) 441-450,
 \href{https://arxiv.org/abs/hep-th/9407021}{[arXiv:hep-th/9407021].}

\bibitem{pridham}
J. P. Pridham,
{\it Unifying derived deformation theories},
Adv. Math. {\bf 224} (2010), 772--826,
[\href{https://arxiv.org/abs/0705.0344}{arXiv:0705.0344}].

\bibitem{Quillen69}
D. Quillen,
{\it Rational homotopy theory},
Ann. of Math., 2nd Series, {\bf 90}, no. 2 (1969), 205-295,
\href{http://www.jstor.org/stable/1970725}{[www.jstor.org/stable/1970725].}



\bibitem{IIBAlgebra}
M. Sakaguchi,
{\it IIB-branes and new spacetime superalgebras},
J. High Energy Phys. {\bf 0004} (2000) 019,
[\href{https://arxiv.org/abs/hep-th/9909143}{arXiv:hep-th/9909143}].



\bibitem{SatiSchreiber15}
H. Sati, U. Schreiber,
\newblock {\it Lie $n$-Algebras of BPS charges},
\newblock J. High Energy Phys.  2017:87
\newblock \href{http://arxiv.org/abs/1507.08692}{[arXiv:1507.08692].}

\bibitem{SSS09}
H. Sati, U. Schreiber, and J. Stasheff,
{\it $L_\infty$-algebra connections and applications to String- and Chern-Simons $n$-transport},
Quantum field theory, 303--424, Birkh\"auser, Basel (2009)
 	[\href{http://arxiv.org/abs/0801.3480}{arXiv:0801.3480}].

\bibitem{SSSIII}
H.~Sati, U.~Schreiber, and J.~Stasheff,
\newblock {\it Twisted differential String- and Fivebrane structures},
\newblock Commun. Math. Phys. {\bf 315} (2012), 169-213,
\newblock \href{http://arxiv.org/abs/0910.4001}{[arXiv:0910.4001].}



\bibitem{dcct}
U.~Schreiber,
\newblock{\it Differential cohomology in a cohesive $\infty$-topos},
\newline
\href{https://ncatlab.org/schreiber/show/differential+cohomology+in+a+cohesive+topos}{ncatlab.org/schreiber/show/differential+cohomology+in+a+cohesive+topos}



\bibitem{StructureTheory}
U. Schreiber,
{\it Structure theory for higher WZW terms},
lectures at \href{http://www.esi.ac.at/activities/events/2015/higher-structures-in-string-theory-and-quantum-field-theory}{\it Higher Structures in String Theory and Quantum Field Theory}, ESI Vienna 2015
\newline
\href{https://ncatlab.org/schreiber/show/Structure+Theory+for+Higher+WZW+Terms}{ncatlab.org/schreiber/show/Structure+Theory+for+Higher+WZW+Terms}

\bibitem{SchreiberBristol}
U. Schreiber,
{\it Modern Physics formalized in Modal Homotopy Type Theory},
talk at J. Ladyman (org.), {Applying homotopy type theory to physics},
Bristol, April 2015
\newline
\href{https://ncatlab.org/schreiber/show/Modern+Physics+formalized+in+Modal+Homotopy+Type+Theory}{ncatlab.org/schreiber/show/Modern+Physics+formalized+in+Modal+Homotopy+Type+Theory}


\bibitem{Sch96}
J. H. Schwarz,
{\it M theory extensions of T-duality},
[\href{https://arxiv.org/abs/hep-th/9601077}{arXiv:hep-th/9601077}].



\bibitem{Se}
 A. Sen, {\it T-duality of $p$-branes}, Mod. Phys. Lett. {\bf A11} (1996) 827-834,
[\href{https://arxiv.org/abs/hep-th/9512203}{arXiv:hep-th/9512203}].

\bibitem{Siegel}
W. Siegel,
{\it Superspace duality in low-energy superstrings},
Phys. Rev. {\bf D48} (1993) 2826-2837,
[\href{https://arxiv.org/abs/hep-th/9305073}{arXiv:hep-th/9305073}].

\bibitem{Sullivan77}
D. Sullivan
{\it Infinitesimal computations in topology},
Publ. Math. de I.H.{\' E}.S. {\bf 47} (1977),  269-331.

\bibitem{Vafa96}
C. Vafa,
{\it Evidence for F-theory},
Nucl. Phys. {\bf B469} (1996), 403-418,
[\href{http://arxiv.org/abs/hep-th/9602022}{arXiv:hep-th/9602022}].

\bibitem{vN}
P. van Nieuwenhuizen,
{\it Free graded differential superalgebras},
 Group Theoretical Methods in Physics, 228--247,
Lecture Notes in Physics 180, M. Serdaroglu and E. Inonu (eds.),
Springer-Verlag, Berlin, Germany,  1983.


\bibitem{VigueBurghelea}
M. Vigu{\'e}-Poirrier and D. Burghelea,
{\it A model for cyclic homology and algebraic K-theory of 1-connected topological spaces},
J. Differential Geom. {\bf 22} (1985), 243-253,
\newline
\href{https://projecteuclid.org/euclid.jdg/1214439821}{[projecteuclid.org/euclid.jdg/1214439821].}



\bibitem{VigueSullivan}
M. Vigu{\'e}-Poirrier and D. Sullivan,
{\it The homology theory of the closed geodesic problem},
J. Differential Geom. {\bf 11} (1976) 633-644,
\href{https://projecteuclid.org/euclid.jdg/1214433729}
{[projecteuclid.org/euclid.jdg/1214433729].}



\bibitem{Wi}
E. Witten, {\it D-branes and K-theory}, J. High Energy Phys. {\bf 12} (1998) 019,
\newline
\href{http://arxiv.org/abs/hep-th/9810188}{[arXiv:hep-th/9810188].}


\end{thebibliography}
\end{document}